\RequirePackage[loading]{tracefnt}

\documentclass[jair,twoside,11pt,theapa]{article}
\usepackage{jair, theapa, rawfonts}

\let\chapter\section
\usepackage[ruled,linesnumbered,noresetcount,vlined]{algorithm2e}
\renewcommand\AlCapNameSty{\slshape}
\SetKwFor{Repeat}{repeat}{:}{endw}
\usepackage{amsmath, amssymb, amsfonts, mathrsfs, mathtools}
\usepackage{times, helvet, courier}
\usepackage{graphicx, epsfig, epstopdf}
\usepackage{wrapfig, multirow, multicol}
\usepackage{rotating, booktabs}
\usepackage{url}
\usepackage{stmaryrd,xspace,arydshln}
\usepackage{picinpar}
\usepackage{todonotes}

\def\thmspace{0.0em}
\newtheorem{theorem}{\hspace{\thmspace}{\bf Theorem}\!}

\newtheorem{definition}{\hspace{\thmspace}{\bf Definition}\!}

\newtheorem{lemma}{\hspace{\thmspace}{\bf Lemma}\!}
\newenvironment{proof}{{\textit{Proof}.}}{\hfill$\Box$}
\newtheorem{example}{\hspace{\thmspace}{\bf Example}\!}

\def \iff{\;\Leftrightarrow\;}

\newcommand{\benumerate}{\begin{list}{$\bullet$}{\topsep=0pt \parsep=0pt \itemsep=1pt \labelwidth=1.5em \labelsep=0.5em \leftmargin=20pt}}
\newcommand{\eenumerate}{\end{list}}
\newcommand{\bitemize}{\begin{list}{$\bullet$}{\topsep=0pt \parsep=0pt \itemsep=1pt \leftmargin=10pt}}
\newcommand{\eitemize}{\end{list}}

\newcommand{\setf}[1]{{\bf{#1}}}

\newcommand{\cA}{\mathcal{A}}
\newcommand{\cD}{\mathcal{D}}
\newcommand{\cI}{\mathcal{I}}
\newcommand{\cO}{{\mathcal{O}}}	%

\newcommand{\cR}{\mathcal{R}}
\newcommand{\cT}{\mathcal{T}}

\newcommand{\sU}{{\mathscr{U}}}	
\newcommand{\sF}{{\mathscr{F}}} 
\newcommand{\sR}{{\mathscr{R}}}

\newcommand{\sd}{\setf{d}}
\newcommand{\MLap}{\mathcal{M}_{\mathrm{Lap}}}

\newcommand{\bF}{\setf{F}}
\newcommand{\bc}{\setf{c}}
\newcommand{\bx}{\setf{x}}

\newcommand{\bD}{\setf{D}}
\newcommand{\bDp}{\setf{D}'}
\newcommand{\bxp}{\setf{x}'}

\newcommand{\EE}{\mathbb{E}}

\newcommand{\RR}{\mathbb{R}}

\newcommand{\Lap}{{\text{Lap}}}

\def \algname{\textsc{OptStream}}

\graphicspath{{images/}}

\jairheading{1}{1}{1}{01/01}{01/01}
\ShortHeadings{OptStream: Releasing Time Series Privately}
{F.~Fioretto \& P.~Van Hentenryck}
\firstpageno{1}

\begin{document}

\title{OptStream: Releasing Time Series Privately}
\author{\name Ferdinando Fioretto \email fioretto@gatech.edu \\
       \addr School of Industrial and Systems Engineering \\
       Georgia Institute of Technology \\
       Atlanta, GA  30332, USA
       \AND
       \name Pascal Van Hentenryck \email pvh@isye.gatech.edu \\
       \addr School of Industrial and Systems Engineering \\
       Georgia Institute of Technology \\
       Atlanta, GA  30332, USA
}
\maketitle\allowdisplaybreaks\sloppy

\begin{abstract} 
Many applications of machine learning and optimization operate on data streams. While these datasets are fundamental to fuel decision-making algorithms, often they contain sensitive information about individuals, and their usage poses significant privacy risks. Motivated by an application in energy systems, this paper presents \algname{}, a novel algorithm for releasing differentially private data streams under the w-event model of privacy. \algname{} is a 4-step procedure consisting of sampling, perturbation, reconstruction, and post-processing modules. First, the \emph{sampling} module selects a small set of points to access in each period of interest. Then, the \emph{perturbation} module adds noise to the sampled data points to guarantee privacy. Next, the \emph{reconstruction} module re-assembles non-sampled data points from the perturbed sample points. Finally, the \emph{post-processing} module uses convex optimization over the private output of the previous modules, as well as the private answers of additional queries on the data stream, to improve accuracy by redistributing the added noise. \algname{} is evaluated on a test case involving the release of a real data stream from the largest European transmission operator. Experimental results show that \algname{} may not only improve the accuracy of state-of-the-art methods by at least one order of magnitude but also supports accurate load forecasting on the private data.
\end{abstract}


\section{Introduction}

Differential privacy \cite{dwork:06} has emerged as a robust framework to release datasets while limiting the disclosure of 
participating individuals. Informally, it ensures that what can be 
learned about an individual in a differential private dataset 
is, with high probability, limited to what could have been learned about the individual in the same dataset but without her data.

Many applications of machine learning and optimization, in areas such
as healthcare, traffic management, and social networks, operate over
streams of data. The use of differential privacy for the private
release of time series has attracted increased attention in recent
years (e.g.,
\cite{dwork:10,dwork:2010a,fanti:16,ding:17,chen2017pegasus}) where
aggregated statistics are continuously reported. Two common approaches for time series data release are the \emph{event-level} and
\emph{user-level} privacy models \cite{dwork:10}. The former focuses
on protecting a single \emph{event}, while the latter aims at
protecting \emph{all} the events associated with a single user, i.e.,
it focuses on protecting the presence of an individual in the
dataset. Additionally, \citeA{kellaris:14} proposed the notion of
$w$-event privacy to achieve a balance between event-level and
user-level privacy, trading off utility and privacy to protect event
sequences within a time window of $w$ time steps.

This paper was motivated by a desire to release private streams of
energy demands, also called \emph{loads}, in transmission systems. 
The goal is that of protecting changes in consumer loads up to some 
desired amount within critical time intervals. Although customer 
identities are typically considered public information (e.g., each facility is served by some energy provider), their loads can be highly sensitive as they may reveal the economic activities of grid customers. For example, changes in load consumption may indirectly reveal production levels and strategic investments and other similar information. 
Moreover, these time series are often input to complex analytic 
tasks,  e.g., demand forecasting algorithms \cite{nogales2002forecasting} and optimal power flows
\cite{ochoa2011minimizing}. As a result, the accuracy of the private
datasets is critical and, as shown later in the paper, existing
algorithms for time series fall short in this respect for this
application.

The main contribution of this paper is a new privacy mechanism that
remedies these limitations and is sufficiently precise for use in
forecasting and optimization applications. The new algorithm, called
\algname{}, is presented under the framework of $w$-event privacy and
is a 4-step procedure consisting of sampling, perturbation,
reconstruction, and post-processing modules.  The \emph{sampling}
module selects a small set of points for private measurement in each
period of interest, the \emph{perturbation} module introduces noise to
the sampled data points to guarantee privacy, the
\emph{reconstruction} module reconstructs the non-sampled data points
from the perturbed sampled points, and the \emph{post-processing}
module uses convex optimization over the private output of the
previous modules, as well as the private answers of additional queries
on the data stream, to redistribute the noise to ensure consistency of
salient features of the data. \algname{} is also generalized to answer
queries over hierarchical streams, allowing data curators to monitor
simultaneously streams produced by energy profile data at different
levels of aggregation. It is important to emphasize that, although
\algname{} was motivated by an energy application, it is potentially
useful for many other domains, since its design is independent of the
underlying problem.

\algname{} is evaluated on real datasets from \emph{R\'{e}seau de
  Transport d'\'{E}lectricit\'{e}}, the French transmission operator
and the largest in Europe. The dataset contains the energy consumption
for one year at a granularity of 30 minutes. \algname{} is also
compared with state-of-the-art algorithms adapted to $w$-event
privacy.  Experimental results show that \algname{} improves the
accuracy of state-of-the-art algorithms by at least one order of
magnitude for this application domain. The effectiveness of the
proposed algorithm is measured, not only in terms of the error between
the reported private streams and the original stream but also in the
accuracy of a load forecasting algorithm based on the private
data. Finally, the paper shows that the sampling and
optimization-based post-processing steps are critical in achieving the
desired performance and that the improvements are also observed when
releasing hierarchical streams of data.

The rest of this paper is organized as follows. Section
\ref{sec:preliminaries} discusses the stream model, summarizes the
privacy goals of this work, and reviews the notion of differential
privacy over streams.  Section \ref{sec:algorithm} describes
\algname{} and the design choices of its components.  Section
\ref{sec:error_analysis} analyzes the accuracy of the proposed
framework and shows how it reduces the error introduced to preserve
privacy when compared to a standard solution.  Section
\ref{sec:extensions} extends \algname{} to the $\alpha$-\emph{indistinguishability} privacy model, allowing privacy protection of arbitrary quantities and which is critical for the motivating application. Additionally, \algname{} is extended
to handle hierarchical stream data.  Section \ref{sec:experiments}
performs a comprehensive experimental analysis of real data streams
from energy load profiles.  Section \ref{sec:discussion} discusses key aspects of the privacy model adopted to privately releasing streams of data, as well as differences with the event-based model for data streams.  Section~\ref{sec:related_work} discusses the related work and, finally, Section \ref{sec:conclusions} concludes the work.

\def\Dw{D_{\rhd w}}
\def\xw{\bx_{\rhd w}}

\section{Preliminaries} \label{sec:preliminaries} 

This section first reviews basic concepts in differential privacy.  It then presents the $w$-event privacy model used to protect privacy in data streams and its definition of differential privacy.

\subsection{Differential Privacy}
\label{sec:dp}

This section reviews the standard definition of \emph{differential
privacy} \cite{dwork:06}.  Differential privacy focuses on protecting
the privacy of an individual user participating to a dataset. Such
notion relies on the definition of a \emph{adjacency relation
$\sim$} between datasets. Two datasets $X, X'$ are called $neighbors$
if their content differs in at most one tuple: $X\sim X' \iff \| X -
X'\| \leq 1$.
\begin{definition}[Differential Privacy]
   Let $\cA$ be a randomized algorithm that takes as input a dataset
    and outputs an element from a set of possible responses.  $\cA$
    achieves $\epsilon$-differential privacy if, for all sets
    $O \subseteq \cO$ and all neighboring datasets $X,
    X' \in \cD$: 
    \begin{equation} \label{eq:dp_def} \frac{Pr[\cA(X) \in
    O]}{Pr[\cA(X') \in O]} \leq \exp(\epsilon).  \end{equation}
\end{definition}

\noindent The level of privacy is controlled by the parameter
$\epsilon \geq 0$, called the \emph{privacy budget}, with values close
to $0$ denoting strong privacy. The adjacency relation $\sim$ captures
the participation of an individual into the dataset. While
differential privacy algorithms commonly adopts the $1$-Hamming
distance as adjacency relation between datasets, the latter can be
generalized to be any symmetric binary relation $\sim \in \cD^2$. In
particular, the relation $\sim$ captures the aspects of the private
data $D$ that are considered sensitive and has been generalized to
protect locations of individuals \cite{fawaz2014location} and quantities
in general \cite{chatzikokolakis:13}. When a single entry is
associated with a user in the database, an algorithm satisfying
Equation \eqref{eq:dp_def} prevents an attacker with access to the
algorithm's output from learning anything substantial about any
individual.

\subsubsection{Properties of Differential Privacy}

Differential privacy enjoys several important properties, including
composability and immunity to post-processing. 

\paragraph{Composability} 
\emph{Composability} ensures that a combination of differentially private
algorithms preserve differential privacy \cite{dwork:13}.

\begin{theorem}[Sequential Composition]
\label{th:composabilty}
The composition $\cA(D) = (\cA_1(D), \ldots, \cA_k(D))$ of a collection
$\{\cA_i\}_{i=1}^k$ of $\epsilon_i$-differentially private algorithms
satisfies $(\sum_{i=1}^{k}
\epsilon_i)$-differentially privacy.
\end{theorem}

\begin{theorem}[Parallel Composition] 
\label{th:composabilty2} 
Let $D_1$ and $D_2$ be disjoint subsets of $D$ and $\cA$ be an
$\epsilon$-differentially private algorithm.  Then computing $\cA(D
\cap D_1)$ and $\cA(D \cap D_2)$ satisfies $\epsilon$-differential privacy.
\end{theorem}

\paragraph{Post-processing Immunity}
\emph{Post-processing immunity} ensures that privacy guarantees are
preserved by arbitrary post-processing steps \cite{dwork:13}. 

\begin{theorem}[Post-Processing Immunity] 
\label{th:postprocessing} 
Let $\cA$ be an $\epsilon$-differentially private algorithm and $g$ be an arbitrary mapping from the set of possible output sequences $\cO$ to an arbitrary set. Then, $g \circ \cA$ is $\epsilon$-differentially private.
\end{theorem}

\noindent
We now introduce two useful differentially private algorithms. 

\subsubsection{The Laplace Mechanism}

A numeric query $Q$, mapping a dataset to $\RR^d$, can be made
differentially private by injecting random noise to its output. The
amount of noise to inject depends on the \emph{sensitivity} of the query, denoted by $\Delta_Q$ and defined as, 
\[
\Delta_Q = \max_{D \sim D'} \left\| Q(D) - Q(D')\right\|_1.
\]
In other words, the sensitivity of a query is the maximum
$L_1$-distance between the query outputs from any two neighboring
dataset $D$ and $D'$. For instance, $\Delta_Q = 1$ for a query $Q$
that counts the number of users in a dataset.

The Laplace distribution with 0 mean and scale $b$, denoted by
$\Lap(b)$, has a probability density function $\Lap(x|b) = \frac{1}{2b}
e^{-\frac{|x|}{b}}$. It can be used to obtain an
$\epsilon$-differentially private algorithm to answer numeric
queries \cite{dwork:06}. 
In the following, we use $\Lap(\lambda)^d$ to denote the i.i.d.~Laplace distribution over $d$ dimensions with parameter $\lambda$.

\begin{theorem}[Laplace Mechanism $\MLap$]
\label{th:m_lap} 
Let $Q$ be a numeric query that maps datasets to $\RR^d$. The Laplace
mechanism that outputs $Q(D) + z$, where $z \in \RR$ is drawn from the Laplace distribution $\textrm{\Lap}\left(\frac{\Delta_Q}{\epsilon}\right)^d$, achieves $\epsilon$-differential privacy.
\end{theorem}

\subsubsection{Sparse Vector Technique}

The Sparse Vector Technique (SVT) is an important tool of 
differential privacy \cite{dwork:13,hardt2010multiplicative} that allows to answer a sequence of queries without incurring a high privacy cost. The SVT
mechanism is given a sequence of queries $Q = q_1, q_2, \ldots$ and a
sequence of real valued thresholds $\Theta
= \theta_1, \theta_2, \ldots$ and outputs a vector indicating whether
each query answer $q_i$ is above or below the corresponding threshold $\theta_i$.  In other words, 
the output is a vector $\{\top, \bot\}^\ell$ where $\ell$ is the
number of queries answered and $\top$ (resp. $\bot$) indicates that the answer to a noisy query is (resp. is not) above a noisy
threshold.

The SVT mechanism is summarized in Algorithm \ref{alg:sparse_vector}.
It takes as input a dataset $D$, a sequence of queries $Q = q_1,
q_2,\ldots$, each with sensitivity no more than $\Delta_Q$, a sequence
of thresholds $\Theta = \theta_1, \theta_2, \ldots$, and a constant
$k$, denoting the maximum number of queries to be answered with value
$\top$. Its output consists of a sequence of answers $a_1,
a_2, \ldots$ with each $a_i \in \{\top, \bot\}$.  For each query, SVT
perturbs the corresponding threshold and checks if the perturbed
individual query answer is above the noisy threshold.

\begin{theorem}[SVT]
\label{thm:svt}
The SVT mechanism achieves $\epsilon$-differential
privacy \cite{hardt2010multiplicative}.
\end{theorem}

\noindent
The SVT mechanism is useful especially in situations when one expects
that most answers fall below the threshold, since the noise depends
on $k$. The SVT mechanism will be used in this paper to find good
sampled points in a stream. 

\begin{algorithm}[!t]
    \SetKwInOut{Input}{input}
    \SetKwInOut{Output}{output}
    \Input{$D$: the dataset\\
           $Q = q_1, q_2, \ldots$: a sequence of queries\\
           $\Theta = \theta_1, \theta_2, \ldots$: a sequence of thresholds\\
           $k$: the maximum number of queries to be answered positively\\
           $\epsilon$: the privacy budget}

    $\rho = Lap(2\Delta_Q / \epsilon)$\;
    count = 0\;
    \ForEach{query $q_i \in Q$} {
        $\nu_i = Lap(4k\Delta_Q / \epsilon)$\;
        \eIf{ $q_i(D) + \nu_i \geq \theta_i + \rho$} {
            output $a_i = \top$\;
            count = count + 1\;
            \textbf{Break} if count $ \geq k$\;
        } {
        output $a_i = \bot$\;
        }
    }
    \caption{\textsc{SVT} - Sparse Vector Technique}
    \label{alg:sparse_vector}
\end{algorithm}

\subsection{The $w$-Event Privacy Model for Data Streams}
\label{sec:stream-model}

This section presents the privacy model for streams adopted in this
paper. A \emph{data stream} $D$ is an infinite sequence of elements in
the \emph{data universe} $\sU = \cI \times \cT$, where $\cI$ denotes
the set of user identifiers and $\cT$ is a possibly unbounded set of
time steps. In other words, each tuple $(i, t)$ describes an event
reported by user $i$ that occurred at time $t$. Time is represented
through discrete steps $\cT = \{1, 2, \ldots\}$ and user events are
recorded periodically (e.g., every 30 minutes).  In a data stream $D$,
tuples are ordered by arrival time. If tuple $(i, t)$ arrives after
tuple $(i', t')$, then $t \geq t'$. Additionally, In the following, $D[t]$ denotes a \emph{stream prefix}, i.e., the sequence $D_1, \ldots, D_t$ of all
tuples observed on or before time $t$. Additionally, $\cD$ denotes the set of all datasets describing collections of tuples in $\sU$.

\begin{example} 
\label{ex:1}
Consider a data stream system that collects location data from WiFi
access points (APs) in a collection of buildings. Users correspond to
MAC addresses of individual devices that connect to an AP. An item $(i, t)$ in the data stream represents that the fact that user $i$ has 
made at least one connection to an access point at time step $t$. In 
the following table, user \texttt{02:18:98:09:1a:a4} reports a 
connection at times $t=1$ and $t=2$. The example represents a data 
stream prefix $D[3]$.
\begin{center}
    \resizebox{0.5\linewidth}{!} 
    {
    \begin{tabular}{c | c | c | c}
    \toprule
    $D_1$ & $D_2$ & $D_3$ & \ldots\\
    \midrule
    02:18:98:09:1a:a4 & 02:18:98:09:1a:a4 & 05:12:11:0a:30:03 & \\
    05:12:11:0a:30:03 & 05:12:11:0a:30:03 & 11:28:18:a9:1b:a3 & \\
    11:28:18:a9:1b:a3 & 11:28:18:a9:1b:a3 & 07:22:1a:12:31:33 & \\
                      & 07:22:1a:12:31:33 & 01:11:e2:14:43:b2 & \\
                      & 01:11:e2:14:43:b2 &                   &  \\
    \bottomrule
    \end{tabular}
    }
\end{center}
\end{example}

\noindent
The data curator receives information from an unbounded data stream
$D$ in discrete time steps. At time $t$, the curator collects a
dataset $D_t$ of tuples $(i, t)$ where every row corresponds to
a \emph{unique} user. The curator reports the result of a count query
$Q : \cD \to \RR$ that reports how many users are in the dataset at
time $t$, i.e., $Q(D_t) = | \{i : (i,t) \in D_t\}|$.

\begin{example}
\label{ex:2}
Consider the data stream prefix of Example \ref{ex:1}. The associated
result of the count queries executed onto each dataset of the data
stream is given in the following table.

\begin{center}
    \resizebox{0.3\linewidth}{!} 
    {
    \begin{tabular}{ c | c | c | c}
    \toprule
    $Q(D_1)$ & $Q(D_2)$ & $Q(D_3)$ & \ldots\\
    \midrule
    3 & 5 & 4 & \\
    \bottomrule
    \end{tabular}
    }
\end{center}
\end{example}

\noindent
In our target application, the data curator is interested in
publishing every element $Q(D_t)$ for a recurring period of $w$ time
steps. A \emph{$w$-period} is a set of $w$ contiguous time steps
$t\!-\!w\!+\!1, \!\ldots,\!t$ ($w \!\geq\!  1$). Thus, the answers to
each query $Q(D_t)$ are generated in real time for windows of $w$ time
steps. As a result, this paper adopts the $w$-event privacy framework
by \citeA{kellaris:14}.

\subsection{Differential Privacy on Streams} 
\label{sec:dp_on_streams}

The $w$-event privacy framework \cite{kellaris:14} extends the
definition of differential privacy to protect data streams and has
become a standard privacy notion for data streams (see, for instance,
\cite{rastogi:10,dwork:13,andres2013geo,bolot:13,chan:11}).
The framework operates on stream prefixes and two data streams
prefixes $D[t]$ and $D'[t]$ are \emph{w-neighbors}, denoted by
$D[t] \sim_w D'[t]$, if
\begin{enumerate}
\item [{\it i}.] for each $D_i, D_i'$ with $i \in [t]$, $D_i \sim
D_i'$, and

\item [{\it ii}.] for each $D_i, D_i', D_j, D_j'$ such that $i < j \in
    [t]$ and $D_i \neq D_i', D_j \neq D_j'$, $j-i+1 \leq w$ holds.
\end{enumerate}

\noindent In other words, two stream prefixes are w-neighbors if their
elements are \emph{pairwise} neighbors and all the differing elements
are within a time window of up to $w$ time steps. As a result, when
ensuring the privacy guarantees, the $w$-event framework does not
consider data streams where the differences are beyond a time window
of size $w$: It only needs to consider windows of $w$ elements.

\begin{definition}[$w$-privacy]
\label{def:wdp}
Let $\cA$ be a randomized algorithm that takes as input a stream
prefix $D[t]$ of arbitrary size and outputs an element from a set of
possible output sequences $\cO$. Algorithm $\cA$ satisfies $w$-event
$\epsilon$-differential privacy ($w$-privacy for short) if, for all
$t$, all sets $O \subseteq \cO$, and all $w$-neighboring stream
prefixes $D[t]$ and $D'[t]$:
\begin{equation}
\label{eq:wdp_def}
\frac{Pr[\cA(D[t]) \in O]}{Pr[\cA(D'[t]) \in O]} \leq \exp(\epsilon).
\end{equation}
\end{definition}
An algorithm satisfying $w$-privacy protects the sensitive information
that could be disclosed from a sequence of finite length $w$.  When
$w\!=\!1$, $w$-privacy reduces to event-level privacy
\cite{dwork:10} that protects the disclosure of events in a single
time step.

All the properties of differential privacy discussed above carry over
to $w$-privacy. The Laplace mechanism which adds Laplace noise to
each element of the stream with parameter $w \Delta_Q/\epsilon$
achieves $w$-privacy~\cite{kellaris:14}. 

\begin{table}[!t]
    \centering
    {
    \begin{tabular}{rlcrl}
    \toprule
        Symbol & Definition & ~~~ & Symbol & Definition \\
    \midrule
    $\mathcal{U}$   & Data universe                         && $\sim_w$        & neighboring relation \\
    $D$             & Data stream                           && $x_t$           & Result of query $Q(D_t)$\\
    $D_t$           & Dataset at time $t$                   && $\hat{x}_t$     & Private estimate of $x_t$\\
    $D[t]$          & Stream prefixes                       && $\bx$           & Data stream        \\
    $Q$             & A query on a dataset                  && $\hat{\bx}$     & Private data stream\\
    $\Delta_Q$      & Sensitivity of query $Q$              && $r_u^t$         & User's $u$ value at time $t$ \\
    \bottomrule
    \end{tabular}
    }
    \caption{Commonly used symbols and notations.
    \label{tab:symbols}}
\end{table}

\medskip 
The algorithms studied in this paper satisfy Definition \ref{def:wdp}. When clear from the context, the paper uses
$\epsilon$-differential privacy to denote $w$-event privacy.  A list
of common symbols is summarized in Table \ref{tab:symbols}.


\section{OptStream For Stream Release}
\label{sec:algorithm}

This section describes \algname, a novel algorithm for private data
stream release. \algname{} consists of four steps: (1) data sampling,
(2) perturbation, (3) reconstruction of the non-sampled data points,
and (4) optimization-based post-processing.  The algorithm takes as
input the data stream, denoted by $\bD^I = D^I_1, D^I_2,\ldots$, the
period size $w$ whose privacy is to be protected, the privacy budget
$\epsilon$, and some hyper-parameters used by its procedures, which
will be described later.  Its output is a data stream \emph{summary}
$\hat{\bx} = (\hat{x}_1, \hat{x}_2, \ldots)$ where each $\hat{x}_t$
represents a private version of the aggregated real data in
${D}^I_t$. Here, an aggregation operation is one that transforms the
dataset into a numerical representation (e.g., a count). The four
steps can be summarized as follows:

\begin{enumerate}

\item \textsc{Sample}: selects a small set of points for each
  $w$-period. Its goal is to perform a dimensionality reduction over
  the data stream whose sample points can be used to generate private
  answers to each query with low error.

\item \textsc{Perturb}: adds noise to the sampled data points to guarantee privacy.

\item \textsc{Reconstruct}: reconstructs the non-sampling data points
from the perturbed sampled points. Its goal is to map the
di\-men\-sion-reduced data stream back to the original space,
generating thus $w$ data points.

\item \textsc{Post-process}: uses the private output of the above
modules, as well as private answers to additional queries on the
data stream to ensure consistency of salient features of the data.
\end{enumerate}

\noindent
\algname{} balances two types of errors: a \emph{perturbation error},
introduced by the application of additive noise at the sampling
points, and a \emph{reconstruction error}, introduced by the
reconstruction procedure at the non-sampled points.  The higher the
number of samples in a $w$-period, the more perturbation error is
introduced while the reconstruction error may be reduced, and
vice-versa. Section \ref{sec:error_analysis} describes the error
generated by these two components and analyzes the number of samples
that minimizes the error.

\newcommand{\tpmod}[1]{{\@displayfalse\pmod{#1}}}
\let\oldnl\nl
\newcommand{\nonl}{\renewcommand{\nl}{\let\nl\oldnl}}
\def\funname{\textsc{Release}}

\begin{algorithm}[!t]
    \caption{\textsc{\algname} data stream release}
    \label{alg:stream_release}
    \SetKwInOut{Input}{input}
    \SetKwInOut{Output}{output}
    \setcounter{AlgoLine}{0}

    \Input{${\bf D}^I = (D^I_1, D^I_2, \ldots)$: the data stream \\
              $w$: the size of the period\\
              $\epsilon$: the privacy budget \\
              $k, \theta, \sF$: hyperparameters}
    \Output{${\bf \hat{x}} = (\hat{x}_1, \hat{x}_2, \ldots)$: a private data stream summary}

    \DontPrintSemicolon
    \nonl Let $t$ be the current time step\\
    
    \If{$t\pmod{w} \equiv 0$ \label{a1:if}}  {
        $\bD \gets D^I_{t-w+1}, \ldots D^I_{t}$         
            \label{a1:Ddecl}\;
        
        \funname$\,$($\bD$, $k$, $w$,  $\epsilon$, $\theta$, $\sF$)
            \label{a1:funcall}\;
    }

    \SetKwFunction{FMain}{{\textrm\funname}}
    \SetKwProg{Fn}{Function}{:}{}
    \label{a1:fundecl}
    \nonl \Fn{\FMain{$\bD$, $w$, $\epsilon$, $k$, $\theta$, $\sF$}}
        {
        
        $\bx \gets \textit{aggr}\,(D_{1}), \ldots \textit{aggr}\,(D_{w})$
            \label{a1:xdecl}\; 
        $\epsilon_s, \epsilon_p, \epsilon_o  \gets \,$ split budget $\epsilon$ 
        \label{a1:budget}\;

        $S = \textsc{Sample}\,(\bx, \epsilon_s, k, \theta)$
            \label{a1:sample}\;
        
        $\tilde{\bx}_{S} = \textsc{Perturb}\,(\bx[S], \epsilon_p)$
            \label{a1:perturb}\;
        
        $\tilde{\bx} = \textsc{Reconstruct}\,(\tilde{\bx}_{S}, w)$
            \label{a1:reconstruct}\;
        
        $\hat{\bx} = \textsc{PostProcess}\,(\tilde{\bx}, \epsilon_o, \{Q_{\bF}(\bD)\}_{\bF \in \sF})$
            \label{a1:postprocess}\;
        

        release $\hat{\bx}$
            \label{a1:end}\;
    }
\end{algorithm}


\algname{} processes the data stream in consecutive and disjoint
$w$-periods. To simplify notation, throughout this section, the paper
uses $\bD = D_1, \ldots, D_w$ and $\bx = x_1, \ldots, x_w$ to denote,
respectively, the current w-period being processed, and its univariate
discrete series representation, where each $x_i$ denotes the result of
an aggregation function $\textit{aggr}(D_i)$ that transforms the
dataset into a numerical representation. Similarly, we use $\bDp$ and
$\bxp$ to denote their neighboring counterparts.  Additionally, given
a set of time steps indexes $S$, $\bx[S]$ denotes the collection of
points $\{x_i | i \in S\}$. For simplicity, we will assume that the
sensitivity of the aggregation operations is one. The results directly
generalize to arbitrary sensitivities.

\algname{} is depicted in Algorithm \ref{alg:stream_release}.  When a
new $w$-period is observed (line \ref{a1:if}), the algorithm extracts
the relevant portion of the stream (line \ref{a1:Ddecl}) and calls
function \funname{}, which releases a private version of the data
stream in the current w-period (line \ref{a1:funcall}).

In addition to the portion $\bD$ of the data stream to release, the size
of the w-period $w$, and the privacy budget $\epsilon$, function
\funname{} takes, as inputs, three hyperparamers: $k, \theta,$ and
$\sF$. Parameters $k$ and $\theta$ are used by procedure \textsc{Sample}
and represent the maximum number of data points to extract
in the $w$-period and a threshold value respectively. Parameter $\sF$
is a set of \emph{features queries} used by procedure
\textsc{PostProcess}; these queries and their uses will be covered in
detail later. The four steps of Function \funname{} operate on a
discrete series representation $\bx$ of the data stream (line
\ref{a1:xdecl}).  The privacy budget to be used in each step is
computed in line \eqref{a1:budget}. Procedure \textsc{Sample} 
takes as input the sequence of values $\bx$, the maximum number $k$ of
data points to sample, and the portion $\epsilon_s \geq 0$ of the
privacy budget $\epsilon$ used in the sampling
process\footnote{$\epsilon_s$ can be 0 if the sampling processes does
  not access the real data to make a decision on which points to
  sample.} (line \ref{a1:sample}). It outputs the set $S$ of indexes
associated with the values of $\bx$ whose privacy must be protected.
Procedure \textsc{Perturb} takes, as input, the vector $\bx[S]$ of
sampled data points from $\bx$ and outputs a noisy version
$\tilde{\bx}_{S}$ of $\bx$, using a portion $\epsilon_p > 0$ of the
overall privacy budget (line \ref{a1:perturb}).  Procedure
\textsc{Reconstruct} takes, as inputs, the vector $\tilde{\bx}_{S}$
from the perturbation step and the size $w$ of the period and outputs
a vector $\tilde{\bx}$ of size $w$ whose values are private estimates
of the data stream in the $w$-period (line
\ref{a1:reconstruct}). Next, procedure \textsc{PostProcess} takes as
input the vector of points $\tilde{\bx}$, additional feature queries
$Q_{\bF}({\bD})$ for each feature $\bF$ in the set of data features
$\sF$ (which are defined in detail in Section
\ref{sec:post_process_procedure}), and uses a portion $\epsilon_o > 0$
of the overall privacy budget $\epsilon$ to compute a final estimate
$\hat{\bx}$ of $\bx$ (line \ref{a1:postprocess}).  Finally, the private estimates $\hat{\bx}$ are released (line \ref{a1:end}).

\begin{lemma}
    \label{thm:privacy_algorithm}
    Let $\epsilon_s + \epsilon_p + \epsilon_o = \epsilon$.  When the \textsc{Sample},  \textsc{Perturb}, and \textsc{Post-Process} procedures satisfy $\epsilon_s$-, $\epsilon_p$-, and $\epsilon_o$-differential privacy, respectively, Algorithm \ref{alg:stream_release} satisfies
    $\epsilon$-differential privacy.
\end{lemma}

\noindent 
Any sampling, perturbation, and reconstruction algorithms can be used
within Algorithm \ref{alg:stream_release}, provided that they achieve
the intended purpose and satisfy the required privacy guarantees. The
next section describes two variants of the sampling and reconstruction
algorithms that may reduce the error (Section
\ref{sec:error_analysis}) and are shown to perform well experimentally
(Section \ref{sec:experiments}). The perturbation procedure is a
standard application of the Laplace mechanism (Theorem \ref{th:m_lap})
on the set of sampled data points and parallel composition (Theorem
\ref{th:composabilty2}). The post-processing step is described in
Section \ref{sec:post_process_procedure}.

\subsection{The Sampling Procedures}
\label{sec:sample_procedure}

The goal of the sampling procedure is to select $k$ points  of the given $w$-period that summarize the entire data stream period well.  This section considers two strategies.


\paragraph{Equally-Spaced Sampling}
\label{sec:equally_spaced_sampling}

A first strategy is to sample $k$ equally-spaced data points in the
$w$-period. Since this approach does not inspect the values of the
data stream to make its decisions, it does not consume any privacy
budget, i.e., $\epsilon_s=0$.

\paragraph{$L_1$-Based Sampling}
\label{sec:auc_based_sampling}

The second strategy also selects $k$ out of $w$ points but tries to
minimize the error between the original data points and the points
generated by a linear interpolation of the $k$ selected points.  Let
$\xi^{[i,j]}$ be a function capturing the line segment between two
points $x_i$ and $x_j$, i.e.,
\[
\xi^{[i,j]}(t) = (t - i) \frac{x_j - x_i}{j - i} + x_i 
\]
for $t \in[i, j]$. Consider an ordered sequence $S = (\iota_1, \ldots,
\iota_k)$ of $k$ indexes in $[w]$ and define $\xi^S$ as a piecewise
linear function whose pieces are line segments between every two
adjacent points in $S$,
\begin{equation}
\xi^S(t) = 
\begin{cases}
   \xi^{[\iota_1, \iota_2]}(t) & \text{if } t \in [\iota_1, \iota_2] \\
   \xi^{[\iota_2, \iota_3]}(t) & \text{if } t \in [\iota_2, \iota_3] \\
   \quad\vdots & \notag \\
   \xi^{[\iota_{k-1}, \iota_{k}]}(t) & \text{if } t \in [\iota_{k-1}, \iota_{k}].
\end{cases}
\end{equation} 
Define the $L_1^{[a,b]}$-scoring function for $1 \leq a < b \leq w$ as,
\begin{equation}
    \label{eq:f1score}
    L_1^{[a, b]}(\bD) =  \sum_{i=a}^{b} \left| \xi^{[a, b]}(i) - x_i \right|.
\end{equation}

\noindent
The \emph{$L_1$-based sampling} procedure aims at selecting a sequence
$S$ of $k$ indexes $\iota_1 < \iota_2 < \ldots < \iota_k$ in a
$w$-period that minimizes the \emph{$\textrm{L}_1^S$-scoring function}
defined as, 
\begin{equation}
\label{eq:auc_score}
\min_{S} \textrm{L}_1^S (\bD) = \min_{\iota_1 < \iota_2 < \ldots < \iota_k} \sum_{j=1}^{k-1} L_1^{[\iota_j, \iota_{j+1}]}(\bD).
\end{equation}

\begin{figure}[!t]
\centering
\includegraphics[width=0.6\linewidth]{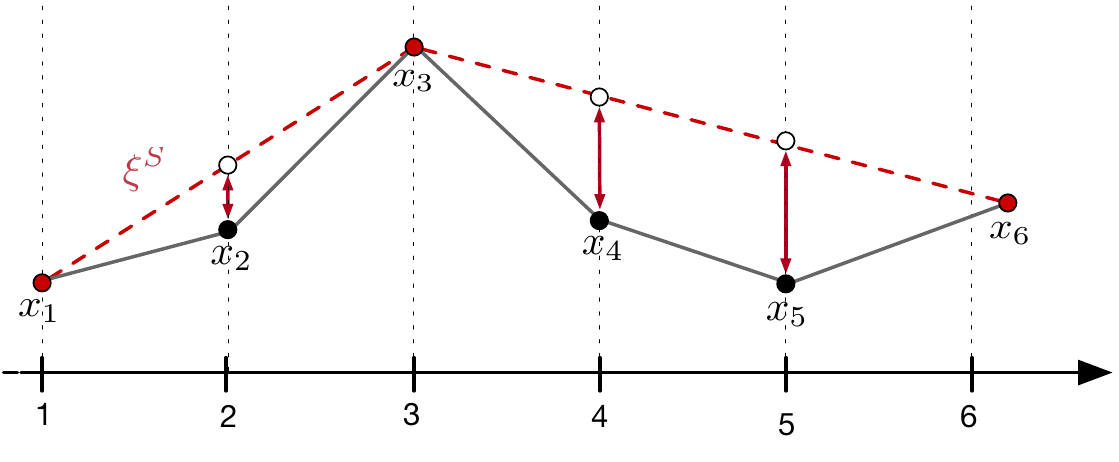}
\caption{Illustration example of the $L_1^S$  function, with $S = {1, 3, 6}$. The solid black curve connects the data point $\bx$, the dashed red curve denotes the function $\xi^S$, and the red arrows denote the distance between the points in $\bx$ and the function $\xi^S$. The $L_1^S(\bD)$ value is the sum of these distances.
\label{fig:l1error}}
\end{figure}


\noindent Figure \ref{fig:l1error} provides an example with a
graphical illustration of the $L_1^S$ function. The black solid curve
delineate $\bx$. The set  $S=(1,3,6)$ and the sampled points $x_1, x_3$, and
$x_6$ are colored red. $\xi^S$ is the piecewise linear function
passing through the sampled points and is represented with dashed red
lines. The red arrows denote the distance between the points in $\bx$
and their corresponding values in $\xi^S$ and the $L_1^S$-score for
the set $S$ is the sum of these distances.  The procedure assumes that
the first and the last points of $\bx$ are in $S$.

Intuitively, the set of $k$ points that minimizes
$\eqref{eq:auc_score}$ produce the minimal \emph{reconstruction error}
when linear interpolation is used as a reconstruction
procedure. However, finding the set of $k$ points minimizing the
$L_1$-scoring function may be computationally expensive. Hence this
paper presents a differentially-private greedy algorithm that
approximates \eqref{eq:auc_score}. The sampling procedure is depicted
in Algorithm \ref{alg:sampler} and is an instantiation of the SVT
mechanism where the queries compute the $L_1$-scores of the potential
interpolation steps.
\begin{algorithm}[!t]
    \SetKwInOut{Input}{input}
    \SetKwInOut{Output}{output}
    \Input{$\bx$: the data stream  seen in the current $w$-period\\  
           $k$: the number of samples\\
           $\epsilon_s$: the privacy budget\\
           $\theta$: a threshold } 
    \Output{$S$: A sample set of $k$ points in $[w]$}
    \setcounter{AlgoLine}{0}

    $S = \{1\}$; \ $t_p = 1$ \\
    $\rho = \text{Lap}(2\Delta_L/\epsilon_s)$\\
    \For{$i = 2 \ldots, w$} {
        $\mu_i = \text{Lap}(4k \Delta_L / \epsilon_s)$\\
        \If{$L_1^{[t_p, i]}(\bx) + \mu_i \geq \theta + \rho$} {
            $S = S \cup \{ i \}$\\
            $t_p = i$
        }
	\If{$w - i \leq k - |S|$} {
            $S = S \cup \{ j \ \mid \ i < i \leq w \}$ \\
	}		
        \textbf{Break} if $|S| = k$ 
    }
    \Return{$S$}
    \caption{\textsc{Sample} - Adaptive AUC-based sampler}
    \label{alg:sampler}
\end{algorithm}

The mechanism takes, as inputs, the data stream $\bx$ processed for
the current $w$-period, the number of points $k$ to sample for
measurements, along with the privacy budget $\epsilon_s$ and a user
defined threshold $\theta \geq 0$ which influences the acceptable
$L_1$-score for choosing the next point to sample. The values choice
for $\theta$ in our application of interest are detailed in Section
\ref{sec:experiments}. Line (1) initializes the set of sample points
$S$ with the first element of the $w$-period (this choice is necessary
for executing the interpolation in the next steps), and it tracks the
last point $t_p$ selected for sampling.  The algorithm first
generates the noise $\rho$ (line 2) for the threshold $\theta$ (line
5). For all but the first time step, the algorithm adds Laplace noise
with parameter $4k\Delta_L/\epsilon$ to the $L_1^{[t_p, i]}$-query
(line 4), where $\Delta_L$ is the largest sensitivity associated with
any of the $L_1$ scoring functions invoked by the algorithm (see
Theorem \ref{thm:auc_sensitivity}).  If the result is above the
threshold, then point $i$ is added to $S$ (line 6) and the
last selected point $t_p$ is updated (line 7). The mechanism keeps
track of the number of index points already stored. It stops when the
size of $S$ matches $k$ (line 10). The mechanism also tests if there
are enough points to reach $k$ (line 8) and adds the remaining points
if needed (line 9).

To run the mechanism, it is necessary to determine the sensitivity
$\Delta_{L_1}^{[a,b]}$ of the $L_1$-score defined in
\eqref{eq:f1score}, i.e.,
\[
\Delta_{L_1}^{[a,b]} = \displaystyle \max_{\bD \sim_w \bD'} \left| L_1^{[a,b]}(\bD') - L_1^{[a,b]}(\bD) \right|
\]
where $\bD$ and $\bD'$ are two w-neighboring data streams in a
$w$-period.

The following theorem shows that $\Delta_{L_1}^{[a,b]}$ can be
bounded, yielding a procedure to privately sample $k$ points in the
$w$-period using an efficient, suboptimal, version of
Equation~\eqref{eq:auc_score}.

\begin{theorem}
\label{thm:auc_sensitivity}
For an arbitrary $w$-period and fixed indexes $a, b \in [w]$ with $a<b$, the sensitivity $\Delta_{L_1}^{[a,b]}$ of the $L_1^{[a,b]}$ score is bounded by $2(b-a)$.
\end{theorem}

\def \auc{L_1^{[a,b]}(\bD)}
\def \aucp{L_1^{[a,b]}(\bD')}

\begin{proof} Consider two data stream $w$-periods $\bD$ and $\bD'$
such that $\bD \sim_w \bD'$ and focus, without loss of generality, on
their associated stream counts $\bx$ and $\bx'$. Let $\xi$ and $\xi'$
be shorthands for $\xi^{[a,b]}$ and $\xi'^{[a,b]}$. The goal is to
bound 
\[ 
    \left|\aucp - \auc\right| = \sum_{i = 1}^b 
        \left| |\xi'(i) - x'_i| - |\xi(i) - x_i| \right|,
\] 
and each term of the summation can be bounded independently. If $i = a\, \lor\, i = b$, then $\xi'(i) = x'_i$ and
$\xi(i) = x_i$, since $x'_i$ and $x_i$ are interpolated exactly and
$\displaystyle\left| |\xi'(i) - x'_i| - |\xi(i) - x_i| \right| = 0.$ Otherwise,
note that $|x'_i - x_i| \leq 1$ since $D_i \sim D'_i = 1$ for all $i
\in [w]$ by definition of $w$-event privacy. Moreover, since the pairs
($\xi(a), \xi'(a))$ and $(\xi(b), \xi'(b))$ differ by at most 1, it
follows that $|\xi'(i) - \xi(i)| \leq 1$ for every point $i \in
[a,b]$. There are four cases to consider.

\def\ai{\xi(i)}
\def\aj{\xi'(i)}
\def\ci{x_i}
\def\cj{x'_i}
\noindent
({\it 1}) If $\xi'(i) \geq x'_i$ and $\xi(i) \geq x_i$, then 
\begin{align*}
\left| |\xi'(i) - x'_i| - |\xi(i) - x_i| \right|
      &=  \left| (\xi'(i) - x'_i) - (\xi(i) - x_i) \right| \\ 
      &=  \left| (\xi'(i) - \xi(i)) - (x'_i - x_i) \right| \\
      &\leq 2.
\end{align*}
({\it 2}) The case $\xi'(i) < x'_i$ and $\xi(i) < x_i$ is symmetric. ({\it 3}) The case
$\xi'(i) \geq x'_i$ and $\xi(i) < x_i$ requires a further case analysis.
\begin{itemize}

\item[(\textit{i})] If $\aj \leq \ai$, we have $\cj \leq \aj \leq \ai < \ci$. Since
$|\ci-\cj| \leq 1$, it follows that $|\aj - \cj| \leq 1$ and $|\ai -
\ci| \leq 1$. Therefore
\[
\displaystyle \left| |\xi'(i) - x'_i| - |\xi(i) - x_i| \right| \leq 1.
\]
\item[(ii)] If $\aj > \ai$, then $\ai+1 \geq \aj > \ai$, since $|\ai - \aj| \leq 1$. It follows that
\[
                \cj \leq \aj \leq \ai + 1 < \ci + 1.
\]
Since $|\ci -\cj| \leq 1$ and  $|\ai - \aj| \leq 1$, $|\aj - \cj| \leq 1$ and $|\ai - \ci| \leq 1$ and therefore
\[
\displaystyle \left| |\xi'(i) - x'_i| - |\xi(i) - x_i| \right| \leq 1.
\]
\end{itemize}

\noindent
({\it 4}) Finally, the last case, $\xi'(i) < x'_i$ and $\xi(i) \geq x_i$, is symmetric to previous one. Therefore
\begin{align*}
\Delta_{L_1}^{[a,b]} = \max_{\bx \sim_w \bx'} |\aucp - \auc| \leq 2(b-a).
\end{align*}
\end{proof}

Let $S = (\iota_1, \ldots, \iota_k)$ be a sequence of $k$ indexes in
$[w]$. Because of the stopping criteria, the largest
contiguous interval $[\iota_i, \iota_j]$ for $i, j \in [k]$ has size
$w-k$. The sensitivity of the $L_1$-score on such an interval is
bounded by $2(w-k)$. Algorithm \ref{alg:sampler} thus uses
\[
\Delta_L = 2(w-k).
\]

\begin{theorem}
\label{thm:privacy_sampling}
Algorithm~\ref{alg:sampler} is $\epsilon_s$-differential private.
\end{theorem} 

\begin{proof} The proof is a direct consequence of the correctness of
the SVT mechanism \cite{hardt2010multiplicative}, Theorem
\ref{thm:auc_sensitivity}, and the definition of $\Delta_L$.
\end{proof}

\subsection{The Perturbation Procedure}
\label{sec:perturb_procedure}

Given the set $S$ of $k$ sampling indexes for a $w$-period, the
perturbation process takes as input the $k$-ary vector of the data
stream measurements $\bx[S]$ and outputs a noisy version
$\tilde{\bx}_S$ of such vector satisfying $\epsilon_p$-differential
privacy. The process simply applies the Laplace mechanism
with parameter $k \Delta_A/\epsilon_p$, where $\Delta_A$ is the
sensitivity of the aggregation query.

\begin{theorem}
\label{thm:privacy_perturb}
   \textsc{Perturb} satisfy $\epsilon_p$-differential privacy.
\end{theorem}
The above result follows by straightforward application of the Laplace mechanism (Theorem \ref{th:m_lap}) on a set of $k$ points and parallel composition (Theorem \ref{th:composabilty2}).

\subsection{The Reconstruction Procedure}
\label{sec:reconstruct_procedure}

The {\sc Reconstruct} procedure takes as input the noisy measurements
$\tilde{\bx}_{S} \in \RR^k$ at the sample points $S$ in $\bx$ and
outputs a vector $\tilde{\bx} \in \RR^w$ of private estimates for the
sub-stream $\bx$.  Each value $\tilde{x}_i$ of $\tilde{\bx}$ is
obtained evaluating the function $\xi^S$ at $i$. Section
\ref{sec:error_analysis} analyzes how well the polynomial approximates
the data stream at any point $x_i$.  The reconstruction
procedure is not required to query the real data stream and uses
exclusively private information to compute its output. Hence, the
output $\tilde{\bx}$ remains $\epsilon_s+\epsilon_p$-differential
private by post-processing immunity of differential privacy (Theorem
\ref{th:postprocessing}).

\subsection{The Optimization-based Post-Processing}
\label{sec:post_process_procedure}

\begin{algorithm}[!t]
    \caption{\textsc{PostProcess} -- Optimization-based post-processing}
    \label{alg:postproess}

    \SetKwInOut{Input}{input}
    \SetKwInOut{Output}{output}
    \Input{
        $\tilde{\bx}$: the private data stream seen in the current $w$-period\\ 
       $Q_{\bF}(\bD)$; the set of feature queries \\
       $\epsilon_o$: the privacy budget}
    \Output{$\hat{\bx}$: a private post-processed data stream}
    \setcounter{AlgoLine}{0}

    ${\bf c} = \MLap(\bx; Q_{\bF_i}, \epsilon/(p-1))$\\
    ${\bf x}^* = {\bf argmin}_{\bf \dot{x}} \| {\bf \dot{x}} - \tilde{\bc} \|_{2,\lambda}^2 = $
    {
    \begin{align}
    &
    \sum_{i=1}^{p} \frac{1}{m_i} \sum_{j = 1}^{m_i} (\dot{x}_{ij} -\tilde{c}_{ij})^2 \; \tag{O1} \\
    \textbf{subject~to}&: \notag \\
    & \forall i', i: \bF_{i'} \prec \bF_{i}, 
        ~~~ j \in [m_i]:  \dot{x}_{ij} = \sum_{l: \setf{d}_{i'l} \subseteq \setf{d}_{ij}} \dot{x}_{i'l} \label{cstr-feature} \tag{O2} \\
    \hspace{-30pt} &  \forall i, j: \dot{x}_{ij} \geq 0. \tag{O3}
    \end{align}
    }
    \Return{$\hat{\bx} = \dot{x}_{11}, \ldots, \dot{x}_{1w}$}\\
\end{algorithm}

The noise introduced in the previous steps may substantially alter the
values of the elements in the w-period, so that some global properties
of interest, such as the total sum of elements in the w-period may
differ from its original value.  The goal of the post-processing step
is to redistribute the noise introduced by the previous steps using
noisy information on aggregated values of the w-period. It does so by
casting the noise redistribution as an optimization problem whose
solution guarantees the consistency of different estimates of
identical quantities.

The {\sc Post-process} step computes the final estimates $\hat{\bx}$ of
the $w$-period data $\bx$ using the private data $\tilde{\bx}$ and
additional queries over $\bx$. The procedure is summarized in
Algorithm \ref{alg:postproess}.  It uses the concept of {\em features}
to capture semantic properties of the application of interest and
queries these features in addition to using the noisy input
$\tilde{\bx}$ of the original data stream.  For example, two important
features in the analysis of WiFi connections profiles are those periods
when peaks typically occur, as well as the total amount of connections
occurring in an $w$-period.

Formally, a \emph{feature} is a partition of the $w$-period and the
size of the feature is the number of elements in the partition.  We
say that a feature $\bF'$ is a \emph{sub-feature} of $\bF$, denoted by
$\bF' \prec \bF$, if $\bF'$ is obtained by sub-partitioning $\bF$.
The \emph{feature query} $Q_{\bF}(\bD)$ on data stream $\bD$
associated with feature $\bF = \{\sd_1, \ldots, \sd_m\}$ returns an
$m$-dimensional vector $(c_1, \ldots, c_m)$ where each $c_i$ is the
sum of the values $x_j$ of $\bx$ for $j \in \sd_i$.

\begin{example} Consider a w-period $\bx \!=\! (10, 15, 20, 23, 41, 72,
55, 50, 88, 72, 40, 18)$ of size 12 reporting the number of WiFi
connections to an access point in a day-period. Consider a feature
$\bF_1 = \{1, \ldots, 12\}$ that includes the indexes for all the
elements of the w-period. The associated feature query $Q_{\bF_1}(\bD)
= (10, 15, 20, 23, 41, 72, 55, 50, 88, 72, 40, 18)$ returns each
frequency value observed during the day.  Next, we consider a
sub-feature $\bF_2$ of $\bF_1$, defined as $\bF_2 = \{ \{1, 2, 3, 4\},
\{5, 6, 7, 8, 9\}, \{10, 11, 12\} \}$. Its associated feature query
$Q_{\bF_2}(\bD) = (68, 306, 130)$ describes the sums of all connection
frequencies associated with time steps in $\{1,2,3,4\}$, $\{5, 6, 7,
8, 9\}$, and $\{10, 11, 12\}$, representing morning, afternoon, and
evening hours respectively.
\end{example}

The optimization-based post-processing takes as input the noisy data
stream $\hat{\bx}$ from the reconstruction procedure and a collection
of features queries $Q_{\bF}(\bD)$ for each $\bF$ in the set of
features $\sF = \{\bF_1, \ldots, \bF_p \}$.  For notational
simplicity, we assume that the first feature always partitions the
data stream $w$-period into singletons, i.e., $\bF_1 = \{\{i\}: i \in
[w]\}$. The noisy answer to this query is the output $\tilde{\bx}$ of
the perturbation procedure (line \eqref{a1:reconstruct} of Algorithm
\ref{alg:stream_release}).  When viewed as queries, the inputs to the
mechanism can be represented as a set of values $Q_{\bF_i}(\bD) = {\bf
c}_i = (c_{i1}, \ldots, c_{im_{i}})$ $(1 \leq i \leq p)$ or, more
concisely, as $\bc = (c_{11},\ldots,c_{pm_p})$.  Finally, we assume
that the partial ordering $\prec$ of features is given.

This first step of Algorithm \ref{alg:postproess} (line 1) applies the
Laplace mechanism with privacy parameter $\frac{\epsilon}{p-1}$ to
each feature query (excluding the first one whose answer
$\tilde{\bx}$ is already private), i.e.,
\[ 
    \MLap(\bx; Q_{\bF_i}, \epsilon/{p-1}) = \tilde{c}_i =
    (\tilde{c}_{i1}, \ldots, \tilde{c}_{im_{i}}) \;\;\; (2 \leq i \leq p).  
\] 
The resulting values $\tilde{\bc} =
(\tilde{c}_{11},\ldots,\tilde{c}_{pm_p})$ are then post-processed by
the optimization algorithm depicted in line (2) to obtain the values
$\bx^* = (x^*_{11},\ldots,x^*_{p{m_p}})$. Finally, the mechanism
outputs a data stream $\hat{\bx} = (\dot{x}^*_{11}, \ldots,
\dot{x}^*_{1w})$.

The essence of Algorithm \ref{alg:postproess} is the optimization
model depicted in line (2).  Its decision variables are the
post-processed values $\dot{\bx} = (\dot{x}_{11}, \ldots,
\dot{x}_{p{m_p}})$, and ${\bf \lambda} = (\lambda_1, \ldots,
\lambda_p) \in (0,1]^p$ is a vector of reals representing weights for
  the terms of the objective function.
The objective minimizes the squared weighted L$_2$-Norm of $\dot{\bx} - \tilde{\bc}$, 
where the weight $\lambda_{i}$ of element $x_{ij} - \tilde{c}_{ij}$ is $\frac{1}{m_i}$. 

The optimization is subject to a set of {\em consistency constraints}
among comparable features and non-negativity constraints on the
variables. For each pair of features $(\bF_{i'},\bF_{i})$ with
$\bF_{i'} \prec \bF_{i}$, constraint \ref{cstr-feature} selects an
element $\setf{d}_{ij} \in \bF_{i}$ and all its subsets $\setf{d}_{i'l} \in
\bF_{i'}$ and imposes the constraint 
\[ 
    \dot{x}_{ij} = \sum_{l:
    \setf{d}_{i'l} \subseteq \setf{d}_{ij}} \dot{x}_{i'l},
\] 
which ensures that
the post-processed value $\dot{x}_{ij}$ is consistent with the sum of the
post-processed values of its partition in $\bF_{i'}$. By definition of
sub-features, there exists a set of elements in $\bF_{i'}$ whose union
is equal to $\setf{d}_{ij}$. 

\begin{theorem}
	\label{thm:postproc}
    The optimization-based post-process achieves $\epsilon_o$-differential privacy.
\end{theorem}

\begin{proof}
    Since each feature partitions the $w$-period over the data stream, each feature query is a count query with sensitivity 1. Thus, each $\tilde{c}_{ij}$ $(i>1)$ obtained from the Laplace mechanism is $\epsilon_o$-differential-private by Theorem~\ref{th:m_lap} and the values $\tilde{\bx} = \tilde{c}_{11}, \ldots, \tilde{c}_{1w}$ are differential private (Theorems \ref{thm:privacy_sampling} and \ref{thm:privacy_perturb}).  
    Additionally, $(\tilde{c}_{11}, \ldots, \tilde{c}_{pm_p})$ is $\epsilon_o$-differential-private by Theorem~\ref{th:composabilty}. Finally, the result follows from post-processing immunity (Theorem~\ref{th:postprocessing}).
\end{proof}

Observe that the mechanisms considered in this paper all operate over the universe of the data stream in the $w$-period. This is the case for instance of the Laplace mechanisms which runs in polynomial time in the size of the $w$-period. 
The next theoretical result characterizes the complexity of the optimization model depicted in line (2) and hence the complexity of Algorithm \ref{alg:postproess}.  Recall that a $\delta$-solution to an optimization problem is a solution whose objective value is within distance $\delta$ of the optimum.

\begin{theorem}
A $\delta$-solution to the optimization to the optimization model in 
line (2) (Algorithm \ref{alg:postproess}) can be obtained in time
polynomial in $w$, the number of features, and
$\frac{1}{\delta}$.
\end{theorem}
\begin{proof}
First observe that the number of variables and constraints in the
optimization model are bounded by a polynomial in size of the period 
$w$ and the number of features $p$. Indeed, since the features are
partitions, every set $\setf{d}_{i'l}$ in Constraint
(\ref{cstr-feature}) is a subset of exactly one $\setf{d}_{ij}$. The
result then follows from the fact that the optimization model is
convex, which implies that a $\delta$-solution can be found in time
polynomial in the size of the universe, the number of features, and
$\frac{1}{\delta}$ \cite{Nemirovski2004}.
\end{proof}


\section{Error analysis}
\label{sec:error_analysis}

This section analyzes how the sampling and reconstruction procedure
can improve the accuracy of the output.  It characterizes the error
$\EE \| \hat{\bx} - \bx \|_2^2$ of the results of Algorithm
\ref{alg:stream_release} over a $w$-period stream release, when the
equally-spaced approach is selected as a sampling procedure and no
post-processing is applied (i.e., when the algorithm releases
$\tilde{\bx}$ (line \ref{a1:reconstruct}, Algorithm
\ref{alg:stream_release})).

\subsection{Sample and Reconstruct}


Consider a $w$-period of the data stream $\bx$ using the same
notational assumptions introduced in the previous period, and thus we
focus in a $w$-period in $[1, w]$. Let $k =|S|$ be the number of
samples selected for \emph{measurements}, by the equally-spaced
sampling in addition to the initial point.  This determines the length
of each \emph{segment} $m = w / k$ during which the
\textsc{Reconstruct} procedure interpolates values without extra
measurements. We assume $m$ to be an even integer. Let $L$ be the
Lipschitz constant defined as \(L := \sup_{t \in [w]} x_t - x_{t -
  1}\).

\begin{theorem}
The error introduced by Algorithm \ref{alg:stream_release} (ignoring
post-processing) with the equally-spaced sampling procedure with
parameter $k$ is bounded by $O\left(m^2 L^2 w + 2\frac{w^2
  L}{\epsilon} + \frac{w^3}{m^{2}\epsilon^{2}}\right)$.
\end{theorem}

\begin{proof}
For notational simplicity, consider the first interval $I =
\{1,\ldots, m\}$ of the $w$-period, where the sample points $1$ and
$m$ are selected. Thus, $x_1$ and $x_m$ are measured privately. The
reconstruct procedure uses linear interpolation between $x_1$ and
$x_m$ to recover the values of the non-sampling points $x_2, \ldots,
x_{m-1}$.
    
For each point $i \in [m]$, there are two sources of error: the
\emph{perturbation error} $e_p$ and the \emph{reconstruction error}
$e_r$.  The worst-case reconstruction error is bounded by $e_r = mL$.
The perturbation error $e_p$ is the additive Laplace noise.  There are
$k$ measurements taken, and hence the privacy budget $\epsilon$ must
be divided by $k$, which results in $e_p = | \tilde{x}_i - x_i | =
\mathrm{Lap}(k/\epsilon)$.  This error adds to the perturbation error
at every point in the $w$-period.  Therefore, for all $i \in [m]$, the
expected error is:
\begin{align*}
\EE \| \tilde{x}_i - x_i \|_2^2  
  &\leq \EE \left[(e_r + e_p)^2 \right]\\
  &\leq    \EE_{Z \sim \mathrm{Lap}(k/\epsilon)} 
        \left[(mL  + Z)^2\right] \\
  &=    m^2L^2 + 2mL\frac{k}{\epsilon} + (\frac{k}{\epsilon})^2.
\end{align*}
\noindent
Multiplying this quantity by the number $m$ of points in the interval gives
\begin{align*}
\EE \| \tilde{\bx}[I] - \bx[I] \|_2^2  
  &\leq m \EE \| \tilde{x}_i - x_i \|_2^2 \leq m^3L^2 + 2m^2L\frac{k}{\epsilon} + m(\frac{k}{\epsilon})^2.
\end{align*}
  As a result, the error $\| \tilde{\bx}[I] - \bx[I] \|_2^2$ is bounded by   $O(m^3 L^2 + 2m^2L\frac{k}{\epsilon} + m (\frac{k}{\epsilon})^2)$. Multiplying the above by the number of intervals $\frac{w}{m}$ in the $w$-period gives the final error which is bounded by
  \[
      O\left(m^2 L^2 w +  2\frac{w^2L}{\epsilon} + \frac{w^3}{m^{2}\epsilon^{2}}\right).
  \]
\end{proof}

\noindent
For $m = \sqrt{\frac{w}{\epsilon L}}$, the above expression generates  an error of $O(w^2 L / \epsilon)$.\footnote{Here, to simplify notation, we consider $m$ as a real value.}  
In comparison, applying the Laplace mechanism to produce a private sub-stream in the $w$-period produces an error of 
$ 
w\, \EE_{\mathrm{Lap}(w/\epsilon)}[Z^2] = \frac{w^3}{\epsilon^{2}}.
$

\smallskip
The result above shows that choosing a sampling parameter $k$ to
sample uniformly every $m$ time steps may allow Algorithm
\ref{alg:stream_release} to produce outputs with a substantially lower
error than those obtained by the Laplace mechanism. Although this
result applies to the equally-spaced sampling procedure, Section
\ref{sec:experiments} demonstrates experimentally that the
$L1$-sampling procedure outperforms its equally-spaced counterpart.

\subsection{Optimization-based Post-Processing}

The following result is from \cite{AAMAS2018}. It bounds the error of
the optimization-based post-processing. It proves that the post-processing step can accommodate any side-constraints
without degrading the accuracy of the mechanism significantly. 

\begin{theorem}
\label{th:opt}
The optimal solution to the optimization model in line \eqref{a1:postprocess} of Algorithm \ref{alg:stream_release} satisfies
\[
\| \bx^* - \bx \|_{2,\lambda} \leq 2 \| \tilde{\bx} - \bx \|_{2,\lambda}.
\]
\end{theorem}

\begin{proof}
   It follows that:
   \begin{align*}
   \| \bx^* - \bx \|_{2,\lambda} 
        &\leq \| \bx^* - \tilde{\bx} \|_{2,\lambda} + \| \tilde{\bx} - \bx \|_{2,\lambda} \label{eq:p2}\\
         &\leq 2 \| \bx - \tilde{\bx}\|_{2,\lambda}.
   \end{align*}
   where the first inequality follows from the triangle inequality on weighted L$_2$-norms and the second inequality follows from
   \[
       \| \bx^* - \tilde{\bx} \|_{2,\lambda} \leq \| \bx - \tilde{\bx}\|_{2,\lambda}
   \]
   by optimality of $\bx^*$ and the fact that $\bx$ is a feasible solution to constraints (O2) and (O3).
   \end{proof}




\section{Privacy Model Extensions}
\label{sec:extensions}

This section first presents an extension of the privacy model that
protects disclosure of arbitrary quantities within the $w$-event
privacy model.  It then generalizes the theoretical results of
\algname{} to this extended privacy model.  Finally, it describes an
extension of \algname{} that supports hierarchical data streams.

\subsection{$\alpha$-Indistinguishability for Data Streams}
\label{sec:indistig}

The application studied in this paper requires the private release of
energy load consumption streams. Unlike the previous setting, where
the participation of users is revealed at each time step, the new
privacy setting requires the \emph{load consumption} of each customer
to be revealed for each time steps.  To capture this requirement, the
data stream setting presented in Section \ref{sec:stream-model} is
extended to an infinite sequence of elements in the \emph{data
  universe} $\sU = \cI \times \cR \times \cT$, where $\cI$ and $\cT$
are defined as in Section \ref{sec:stream-model}, and $\cR$ is the set
of possible \emph{events} (i.e., the data items generated by
users). In other words, each tuple $(i, r, t)$ describes an event that
occurred at time $t$ in which user $i$ reported value $r$.

Each user transmits an aggregate summary of the stream periodically
(e.g., every 30 minutes) and the discrete time step $k$ represents an
aggregate (e.g., the sum or the average) of all values describing the
event occurring in time period $k$. This setting implies that a
tuple $(i, r, t)$ describes the aggregate behavior of user $i$ during
the time period $t$.  Here, a numerical count query $Q: \cD \to
\RR^{|\cI|}$ returns an $|\cI|$-dimensional vector corresponding to
the aggregated reports of each user in $\cI$.  The following example
illustrates such behavior.

\begin{example} Consider a data stream system that collects power
consumption data from customers of an electric company distributed on
a wide geographical region. Customers may correspond to facilities
(such as homes, hospitals, industrial buildings) or electrical
substations, transmitting their power consumption at regular intervals
(e.g., every $30$ minutes). The value $r_t \in \cR$ transmitted by a
customer at time $t$ is a real number denoting the average amount of
power (in MegaWatts) the customer required during time interval $t$.
The table below represents a the scenario in which user $\text{customer}_1$ consumes $3.1$ MW at time $t=1$, $3.2$ MW at time $t=2$, and $2.8$ MW at time $t=3$. 
\begin{center}
    \resizebox{0.4\linewidth}{!} 
    {
    \begin{tabular}{c | c | c | c | c}
    \toprule
    $\cI$ & $Q(D_1)$ & $Q(D_2)$ & $Q(D_3)$ & \ldots\\
    \midrule
    $\text{customer}_1$ & 3.1 & 3.2 & 2.8 & \\
    $\text{customer}_2$ & 1.0 & 1.5 & 1.6 & \\
    $\text{customer}_2$ & 0.7 & 1.1 & 1.2 & \\
    \bottomrule
    \end{tabular}
    }
\end{center}
\end{example}

\noindent
Customer identities are assumed to be a public information, as every
facility consumes power. However, their load fluctuations are
considered to be highly sensitive. Indeed. changes in power
consumption may indirectly reveal production levels and hence
strategic investments, decreases in sales, and other similar
information. These changes should not be revealed within some
application-specific period of $w$ time steps.  Thus, within each
$w$-period, the privacy goal of the data curator is to protect
observed increases or decreases of power consumptions.  More
precisely, consider $r_u^t$ and ${r'}_u^t$ be two distinct values that
may be reported by user $u$ at time $t$ and that satisfy $|r_u^t -
{r'}_u^t| \leq \alpha$ for some positive real value $\alpha$. An
attacker should not be able to confidently determine that $u$ reported
value $r_u^t$ instead of value ${r'}_u^t$ and vice-versa.

This goal can be achieved by using a more general adjacency relation
between datasets. This relation needs to capture the distance between reported quantities whose magnitudes must be protected up to some given value
$\alpha > 0$. The value $\alpha$ is called the
\emph{indistinguishability level}.  This generalized definition of
differential privacy has been adopted in several applications and has
sound theoretical foundations (e.g.,
\cite{koufogiannis:15,andres2013geo,chatzikokolakis:13}\footnote{We
  refer the interested reader to \cite{chatzikokolakis:13} for a broad
  analysis of indistinguishability when differential privacy uses
  different notions of distance.}).  Consider a dataset $X$ to which
$n$ individuals contribute their real-valued data $r_i \in \RR$, i.e.,
$X = (r_1, \ldots, r_n) \in \RR^n$.  For $\alpha > 0$, an adjacency
relation that captures the indistinguishability of a single individual
to the aggregating scheme is:
\begin{equation} \label{eq:adj_rel} 
    X \sim_\alpha X' \iff \exists i
    \textrm{~s.t.~} | r_i - r'_i | \leq \alpha \textrm{~and~} r_j = r'_j,
    \forall j \neq i.  
\end{equation} 
Such adjacency definition, in our target domain, is useful to hide
increases or decreases of loads up to some quantity $\alpha$.  The
generalization of differential privacy that uses the above adjacency
relation is referred to as $\alpha$-indistinguishability
\cite{chatzikokolakis:13}.

The $w$-privacy definition, introduced in Definition \ref{def:wdp},
can be extended to use $\alpha$-indistinguishability by replacing the
standard neighboring relation between datasets with the adjacency
relation defined in Equation \eqref{eq:adj_rel}.  More precisely, for
a given $\alpha > 0$, two data streams prefixes $D[t], D'[t]$ are said
\emph{($w, \alpha$)-neighbors}, written $D[t] \sim_w^\alpha D'[t]$, if
     {\it (i)} for each $D_i, D_i'$ with $i \in [t]$, $D_i \sim_\alpha
     D_i'$, and {\it (ii)} for each $D_i, D_i', D_j, D_j'$, with $i <
     j \in [t]$ and $D_i \neq D_i', D_j \neq D_j'$, it must be the
     case that $j-i+1 \leq w$ holds.

\begin{definition}[($w, \alpha$)-indistinguishability]
    \label{def:alphawdp}
    Let $\cA$ be a randomized algorithm that takes as input a stream
    prefix $D[t]$ of arbitrary size and outputs an element from a set of
    possible output sequences $\cO$. $\cA$ satisfies $w$-event
    $\epsilon$-differential privacy under $\alpha$-indistinguishability ($(w, \alpha$)-indistinguishability for short) if, for all sets $O \subseteq \cO$, all $(w, \alpha)$-neighboring stream prefixes $D[t],
    D'[t]$, and all $t$:
    \begin{equation}
        \label{eq:alphawdp_def}
        \frac{Pr[\cA(D[t]) \in O]}{Pr[\cA(D'[t]) \in O]} \leq \exp(\epsilon).
    \end{equation}
\end{definition}

\noindent
This privacy definition combines $w$-event privacy \cite{kellaris:14}
and a metric-based generalization of differential privacy
\cite{chatzikokolakis:13} in order to meet the requirements of the
motivating application.  Notice that ($w,
\alpha$)-indistinguishability is a generalization of $w$-privacy, and
it reduces to $w$-privacy when $\alpha = 1$.

\algname{} can be naturally extended to satisfy this definition, by calibrating the noise introduced by the Laplace mechanism, used during the sampling (Theorem \ref{thm:auc_sensitivity}), perturbation (Theorem \ref{thm:privacy_perturb}), and post-processing (Theorem \ref{thm:postproc}) steps, to satisfy the above definition. 
More practically, it is achieved by updating the definition of query sensitivity as:
\[
    \Delta_Q = \max_{D[t] \sim_w^\alpha D'[t]} \| Q(D[t]) - Q(D'[t])\|_1.
\]
Thus, in the ($w, \alpha$)-indistinguishability model, the sensitivity
of each query $Q$ used during the \emph{perturb} (see Section
\ref{sec:perturb_procedure}) and the \emph{post-process} (see Section
\ref{sec:post_process_procedure}) procedures is $\Delta_Q = \alpha$,
and the sensitivity $\Delta_L$ of of the queries adopted in the
$L1$-based \emph{reconstruction} procedure (see Section
\ref{sec:reconstruct_procedure}) is $\Delta_L = 2\alpha(w-k)$ .
The behavior of \algname{} for different indistinguishability levels
is presented in Section \ref{sec:experiments}.


\subsection{Hierarchical Data Streams}
\label{sec:algorithm_ext}

\algname\ can also be generalized to support \emph{hierarchical data
  streams}. A data stream $D$ is called hierarchical if, for any time
$t$, there is an \emph{aggregation entry} $a = (i_a, r_a, t)$
associated with an \emph{aggregation} set $S_a \subseteq \cI$ that
reports the sum of values reported by all entries in $S_a$ at time
$t$, i.e., $r_a = \sum_{u \in S_a} r_u^t$.  More formally, a data
stream is hierarchical if, for any two aggregation entries $a_1, a_2$,
either $S_{a_1} \subset S_{a_2}$ or $S_{a_2} \subset S_{a_1}$ holds or
$S_{a_1} \cap S_{a_2} = \emptyset$ is true.
	
A set of aggregation entries $\mathcal{S}$ can be represented
hierarchically through a tree (called an \emph{aggregation tree}) in
which the root is defined by the entry that is not contained in any
other entries and the children of an aggregation entry $p$ are all the
entries in $D$ whose identifiers are contained in $S_p$. The
\emph{height} of a hierarchical data is the maximum path length from
the root to any leaf of its tree.

\begin{example} 
	In our target application, a data curator is
	interested in the aggregated load consumption stream at the level of a
	small geographical region, at a group of regions within the same
	electrical sub-station, and finally at the nation-level. These
	aggregations form a hierarchy whose root is represented by the
	nation-level aggregation entry, its children by the electrical
	sub-station entries, and the tree leaves by the region-level entries.
\end{example}

\noindent
To answer each query of the hierarchical stream, \algname\ is 
extended as follows.  For each $w$-period, the algorithm runs the
sampling, perturbation, and reconstruction procedure for each level of
the aggregation tree representing the hierarchical stream. Finally,
the post-processing optimization takes as input the answers to all the
aggregation entries queries with the set of features being equal to
the set of aggregation entries.  The post-processing step thus enforces 
consistency between the aggregation counts at a node and the
sum of counts at its children nodes in the tree, in addition to the features
described earlier in the paper.

\begin{theorem} For a hierarchy of height $h$, \algname{}, when using a
privacy budget of $\epsilon/h$ for each level of the hierarchy,
satisfies $\epsilon$-differential privacy.
\end{theorem}

\noindent
The result follows from the privacy guarantee of \algname\ (Theorem
\ref{thm:privacy_algorithm}), parallel composition of differential
privacy across each level of the hierarchy (Theorem
\ref{th:composabilty2}), and sequential composition of differential
privacy, for each of the $h$ hierarchical queries (Theorem
\ref{th:composabilty}).


\section{Evaluation}
\label{sec:experiments}

This section evaluates \algname{} on real data streams for a number of
tasks. It first evaluates the accuracy of the private release of a
stream of data within the $w$-privacy model. Then, it studies the
accuracy of each of the four individual components of the
algorithm. Finally, it evaluates the algorithm within the ($w,
\alpha$)-indistinguishability model.

\paragraph*{\bf Dataset} 

The source data was obtained through a collaboration with
\emph{R\'{e}seau de Transport d'\'{E}lectricit\'{e}}, the largest
energy transmission system operator in Europe. It consists of a
one-year national-level load energy consumption data with a
granularity of 30 minutes. The data is aggregated at a regional level
and $\sR$ denotes the set of regions. Each data point in the stream
represents the total load consumption of the customers served within a
region during a 30 minute time period. Thus, for every region $R \in
\sR$, a stream of data $\bx(R) = x_1(R), x_2(R), \ldots$ is generated
where $x_t(R)$ represents the energy demand to supply in order to
serve the region $R$ at time $t$. When the streams of all regions are
aggregated, the resulting load consumption data constitutes a data
stream of the energy profile at a national level.

For evaluation purposes, the experiments often consider a
representative region (Auvergne - Rh\^{o}ne-Alpes) to analyze the data
stream release.  For the evaluation of hierarchical data streams, the
experiments consider a hierarchical aggregation tree of height $2$,
where the root node is the national level and each of the leaves
corresponds to one region. Table \ref{tab:france_regions} lists an
overview of the data streams derived from the real energy load
consumption data for each French region in 2016. Each data stream
contains 17,520 entires.

\begin{table}[!t]
    \centering
    \resizebox{\linewidth}{!} 
    {
    \begin{tabular}{rrcrr}
    \toprule
    Stream Data Regions ($\sR$) & Daily Average (MW) & ~~ & Stream Data Regions ($\sR$) & Daily Average (MW)\\
    \midrule
    Auvergne-Rh\^{o}ne-Alpes & 7717.58  && Ile de France (Paris)   & 8315.13 \\
    Bretagne                 & 2554.23  && Nouvelle Aquitaine      & 4985.68 \\
    Bourgogne - Franche-Comt\`{e} & 2498.23&& Normandie          & 3267.22 \\
    Centre - Val de Loire   & 2157.97   && Occitanie               & 4314.7 \\
    Grand Est               & 5286.24   && Pays de la Loire        & 3174.17\\
    Hauts de France         & 5832.26   && Provence - Cote d'Azur  & 4782.4\\
    \bottomrule
    \end{tabular}
    }
    \caption{Overview of the power load consumption stream data derived from the France regions in 2016. 
    \emph{Daily Average} refers to the average power (in MW) demanded daily. 
    The number of time steps available for all regions is 17,520.
    \label{tab:france_regions}}
\end{table}

\paragraph*{\bf Algorithms}

The following sections evaluate two versions of \algname{}, with
equally-spaced sampling and with the $L1$-sampling. They are compared
against the \emph{Laplace mechanism} and the \emph{Discrete Fourier
  Transform (DFT)} algorithm \cite{rastogi:10}.  All the algorithms
release private data associated with the sub-streams $\bx$ for each
$w$-period.

To ensure $\epsilon$-differential privacy for each $w$-period, the
Laplace mechanism applies Laplace noise with parameter $\epsilon/w$ to
each value in the period. To ensure $\epsilon$-differential privacy,
the DFT algorithm works as follows. For each $w$-period, and given a
value $k < w$, it first computes the \emph{Fourier coefficients}
$DFT(\bx)_j = \sum_{i=1}^w \exp(\frac{2\pi\sqrt{-1}}{w}ji x_i)$ of
each sub-stream $\bx$ and each $j \in [w]$.  It then considers the $k$
coefficients of lowest frequencies -- which represent the high-level
trends in $\bx$ -- and perturbs them using Laplace noise with
parameter $\sqrt{k}\Delta_{2Q}/\epsilon$, where $\Delta_{2Q}$ is the
$L_2$ sensitivity of the count query.  It then pads the vector of
noisy coefficients with a $(w-k)$-dimensional vector of $0$'s to
obtain a $w$-dimensional vector. Finally, it applies the \emph{Inverse
  Discrete Fourier Transform} (IDFT) to this $w$-dimensional vector to
obtain a noisy estimate of the elements in the $w-$period. In more
details, if $\hat{f}_j$ is the j-th perturbed coefficient of the
series, the noisy estimate for the value $x_j$ is obtained through the
IDFT function as: $\frac{1}{w} \sum_{i=1}^w
\exp(\frac{2\pi\sqrt{-1}}{w}ji)\hat{f}_j$.


\subsection{Stream Data-Release in the $w$-Privacy Model}
This section evaluates the accuracy of privately releasing a data
stream within the $w$-privacy model.  It first studies the prediction
error of the algorithms.  It then analyzes the accuracy of privately
releasing hierarchical streams on aggregated queries.  Finally, it
analyzes the accuracy of forecasting tasks from the released private
data streams.

\def \wlen{\linewidth/4}
\def \wlenLabel{\linewidth}
\begin{figure}[!t]
    {\small\centering
    \hspace{35pt}
    ~~Real~ \hspace{75pt}
    Laplace \hspace{70pt}
    ~~DFT~~ \hspace{60pt}
    \algname
    }\\
    \centering
    \includegraphics[width=\wlen]{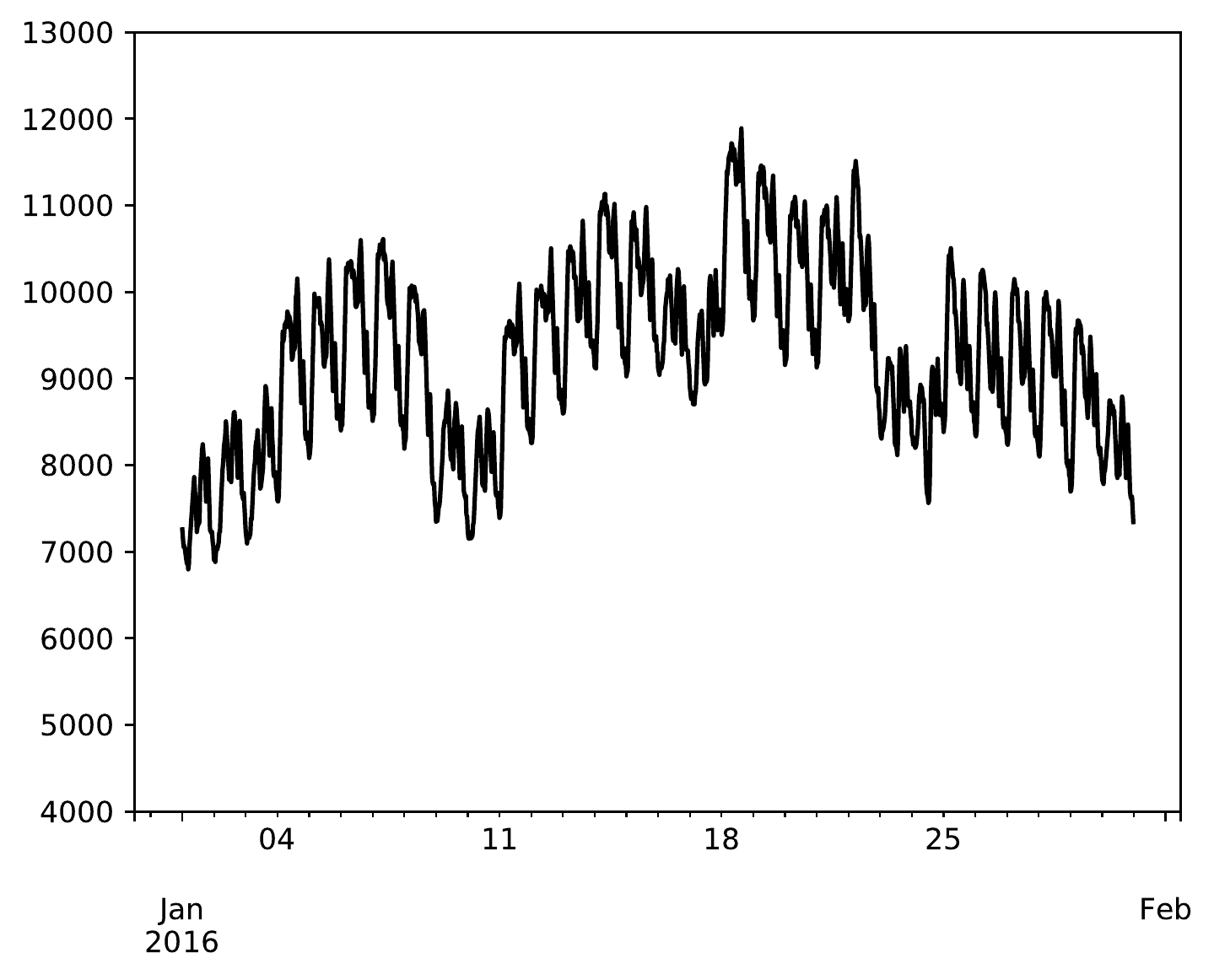} \hspace{-6pt}
    \includegraphics[width=\wlen]{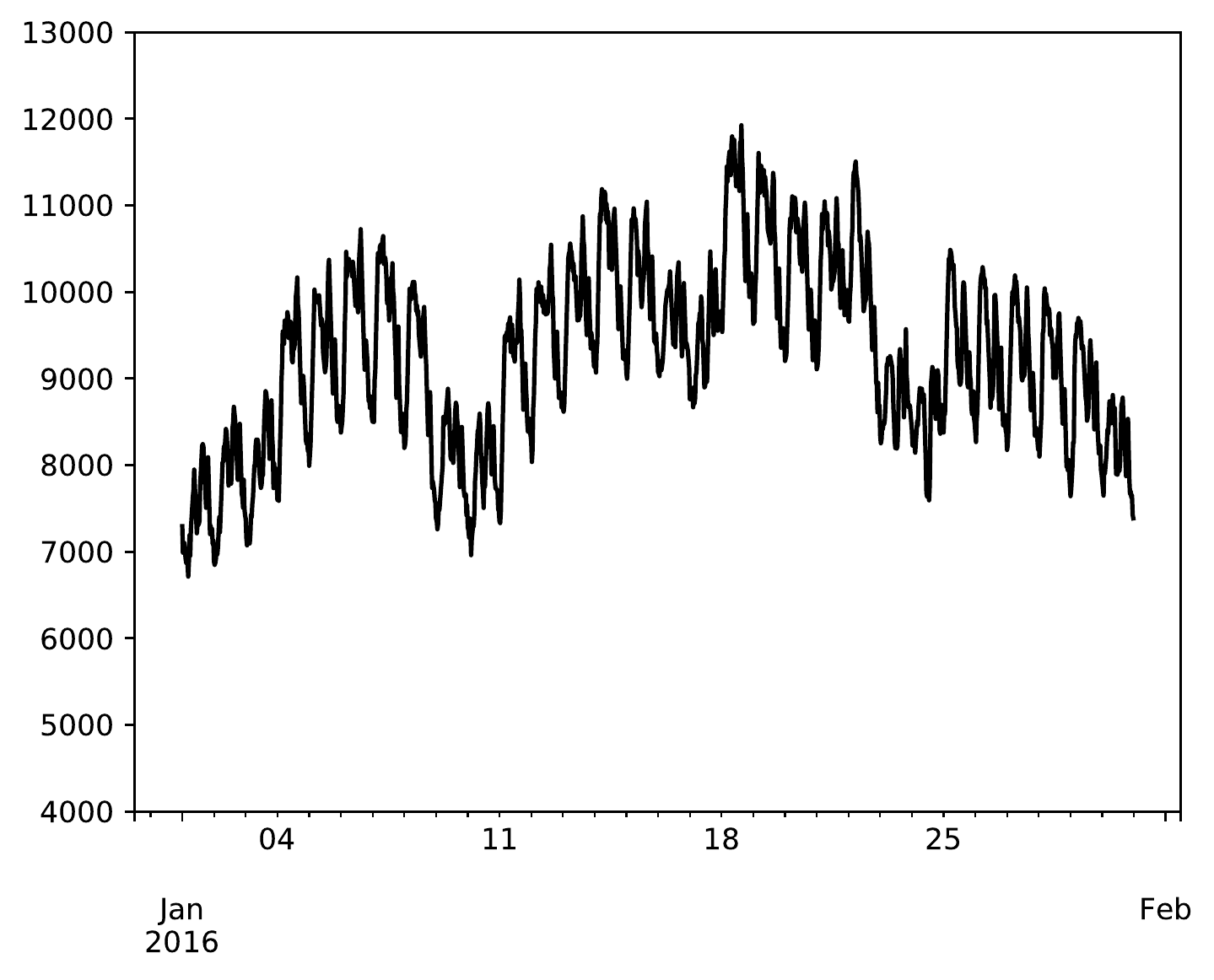} \hspace{-6pt}
    \includegraphics[width=\wlen]{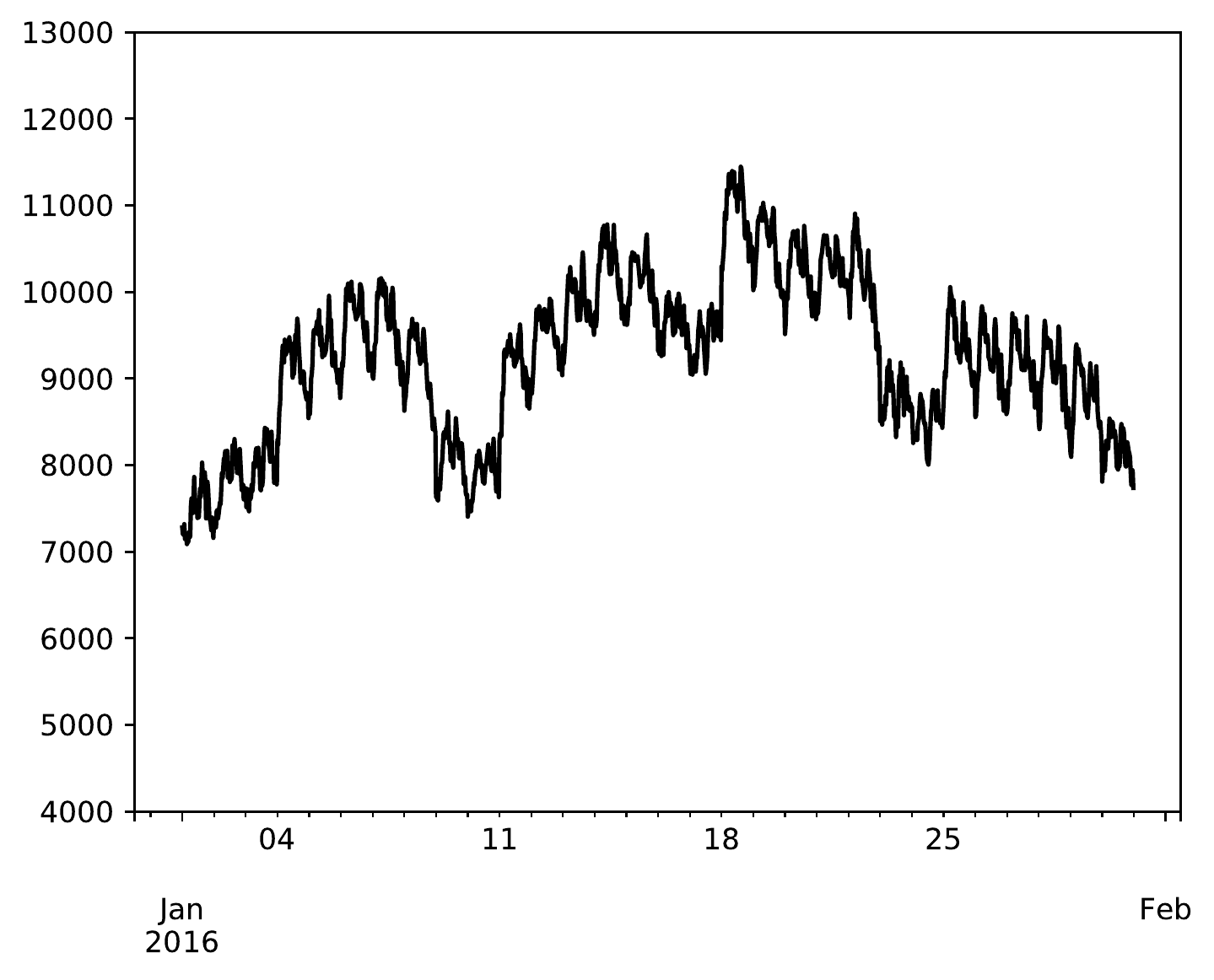}\hspace{-6pt}
    \includegraphics[width=\wlen]{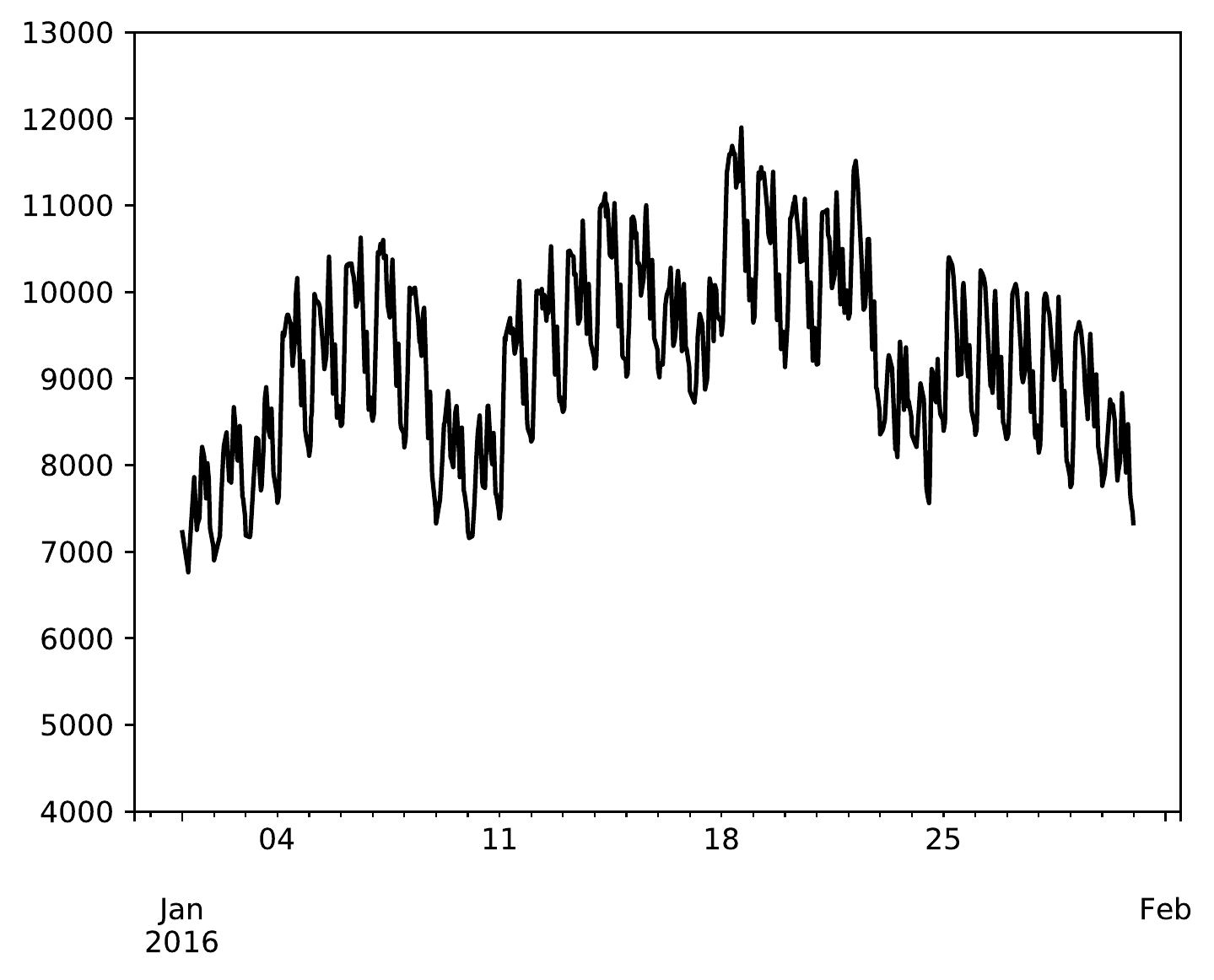}\\
    \includegraphics[width=\wlen]{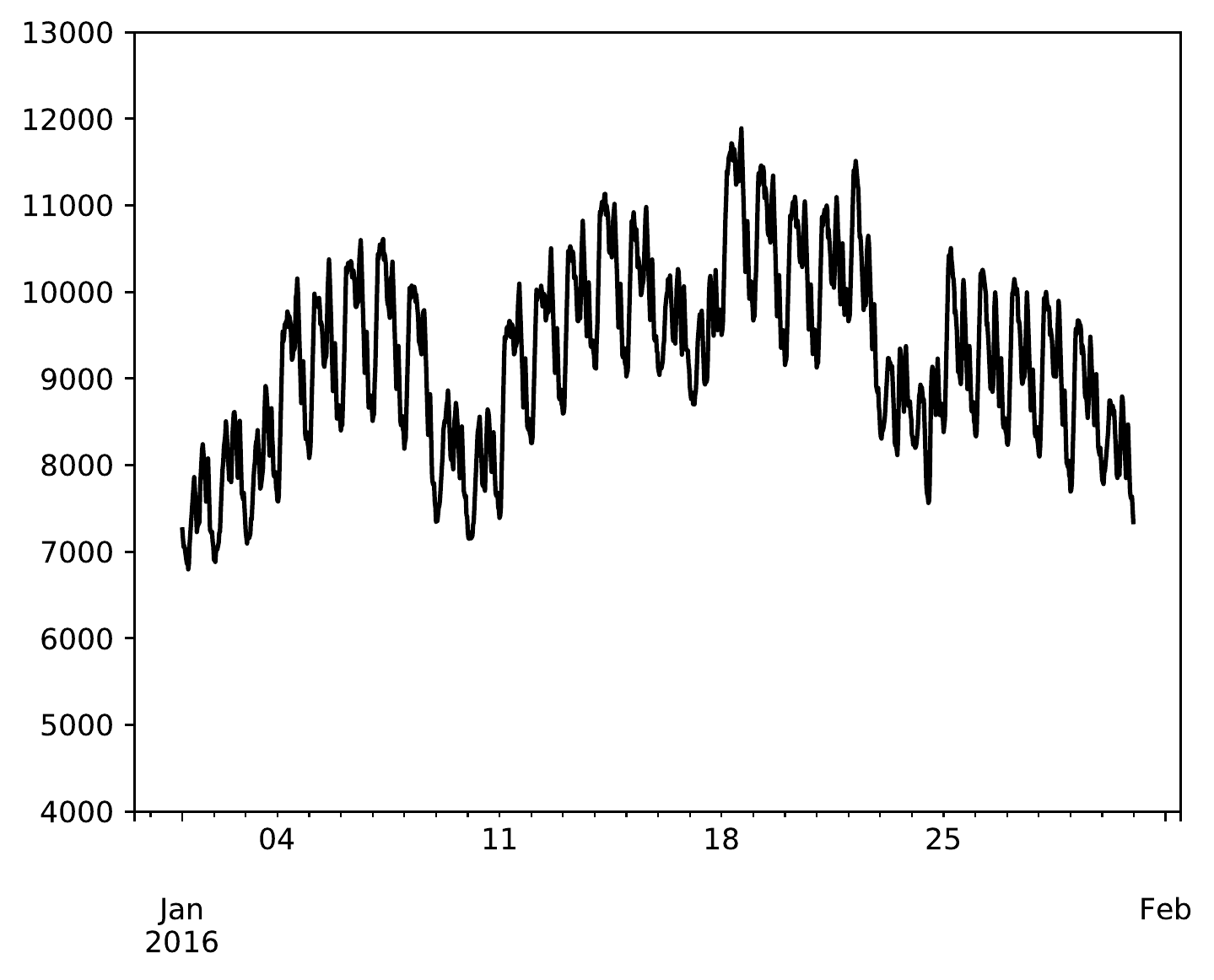} \hspace{-6pt}
    \includegraphics[width=\wlen]{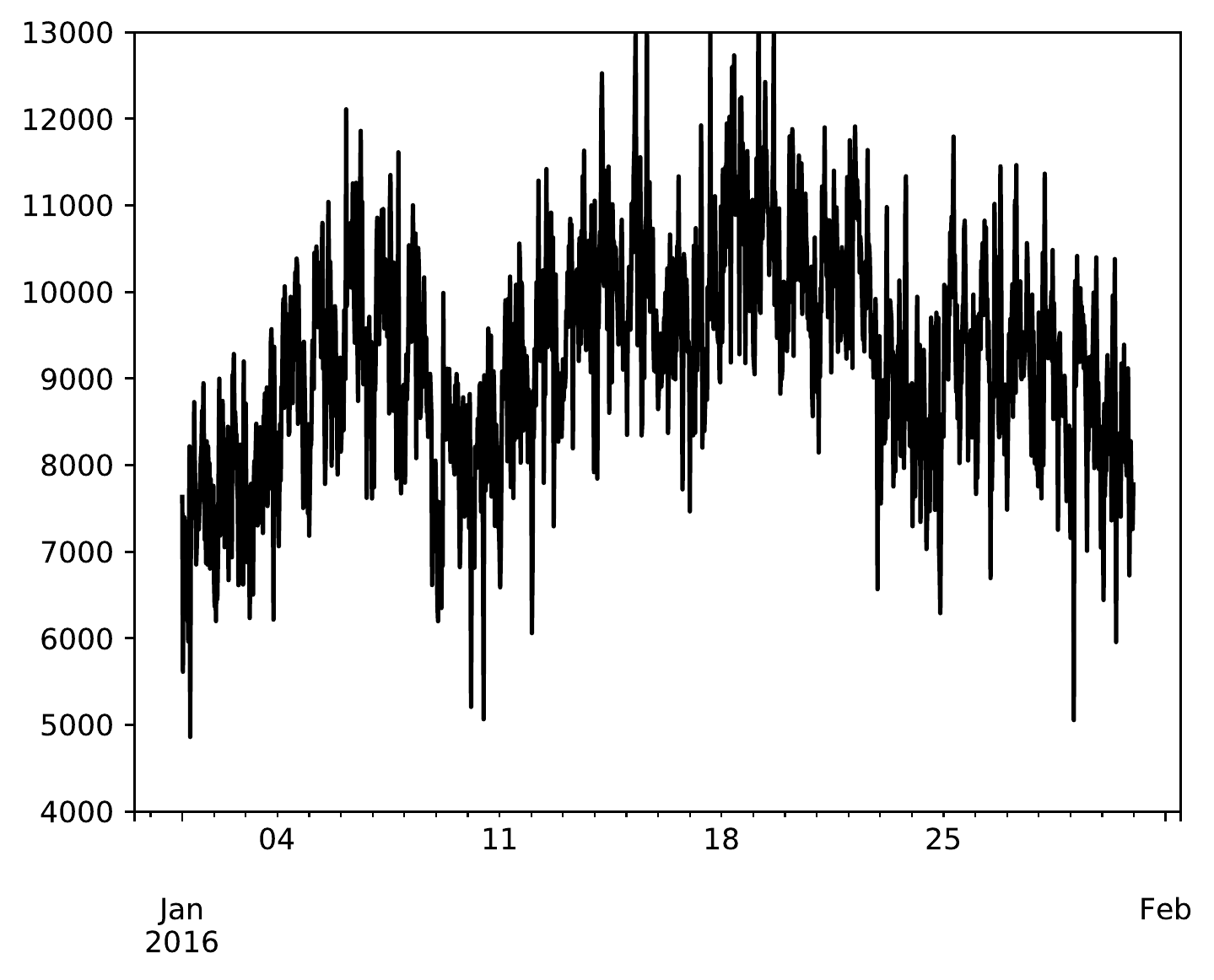} \hspace{-6pt}
    \includegraphics[width=\wlen]{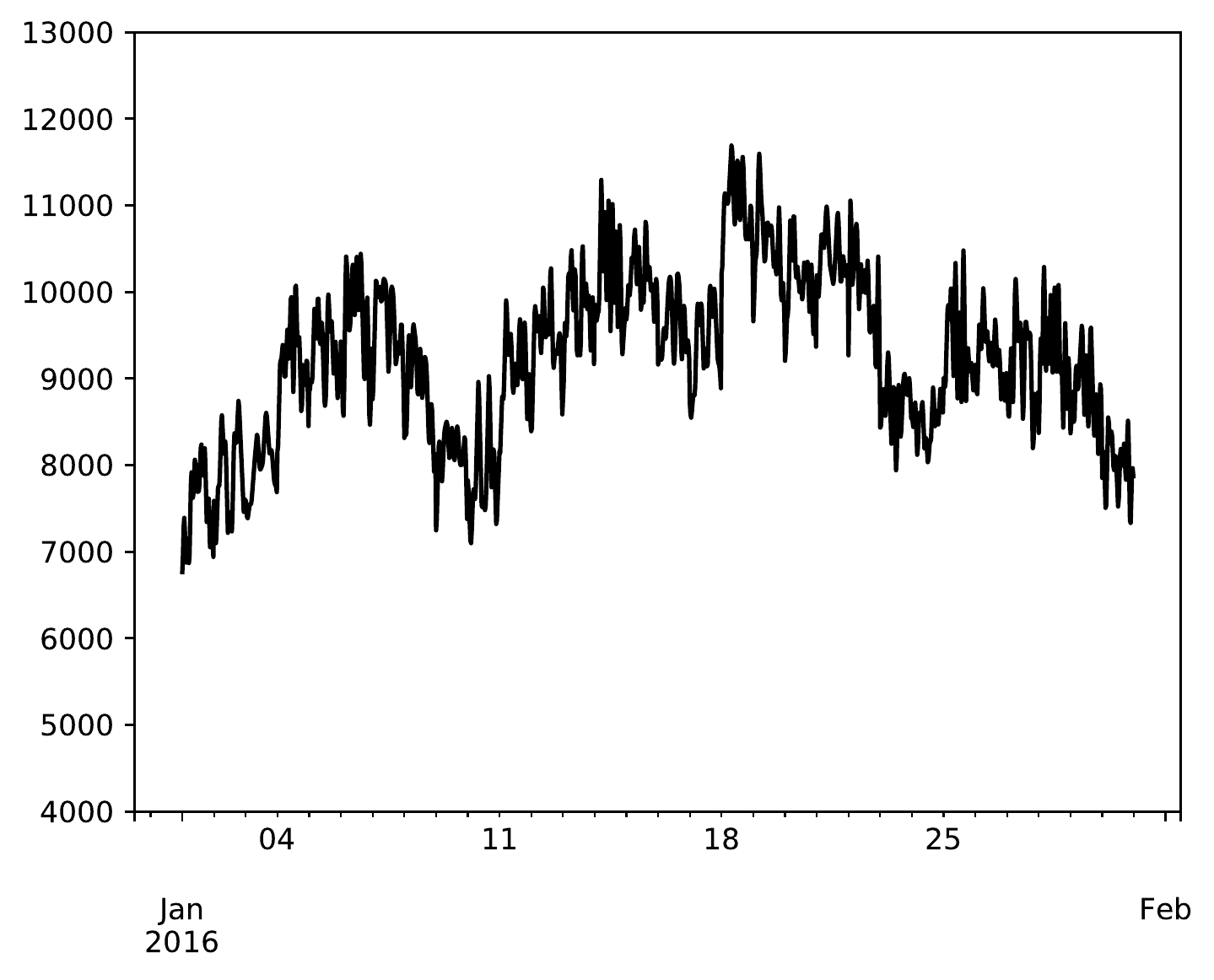}\hspace{-6pt}
    \includegraphics[width=\wlen]{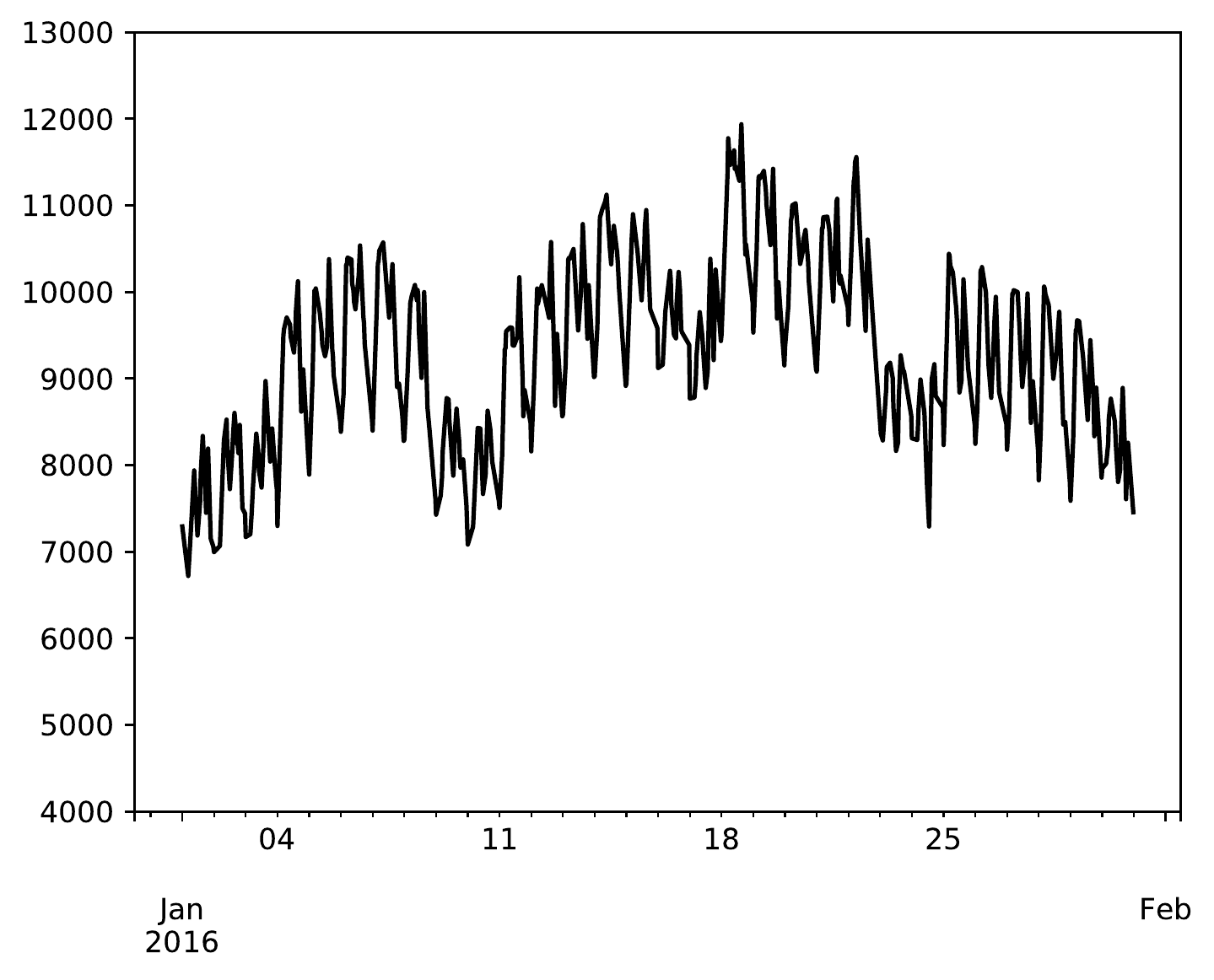}\\
    \includegraphics[width=\wlen]{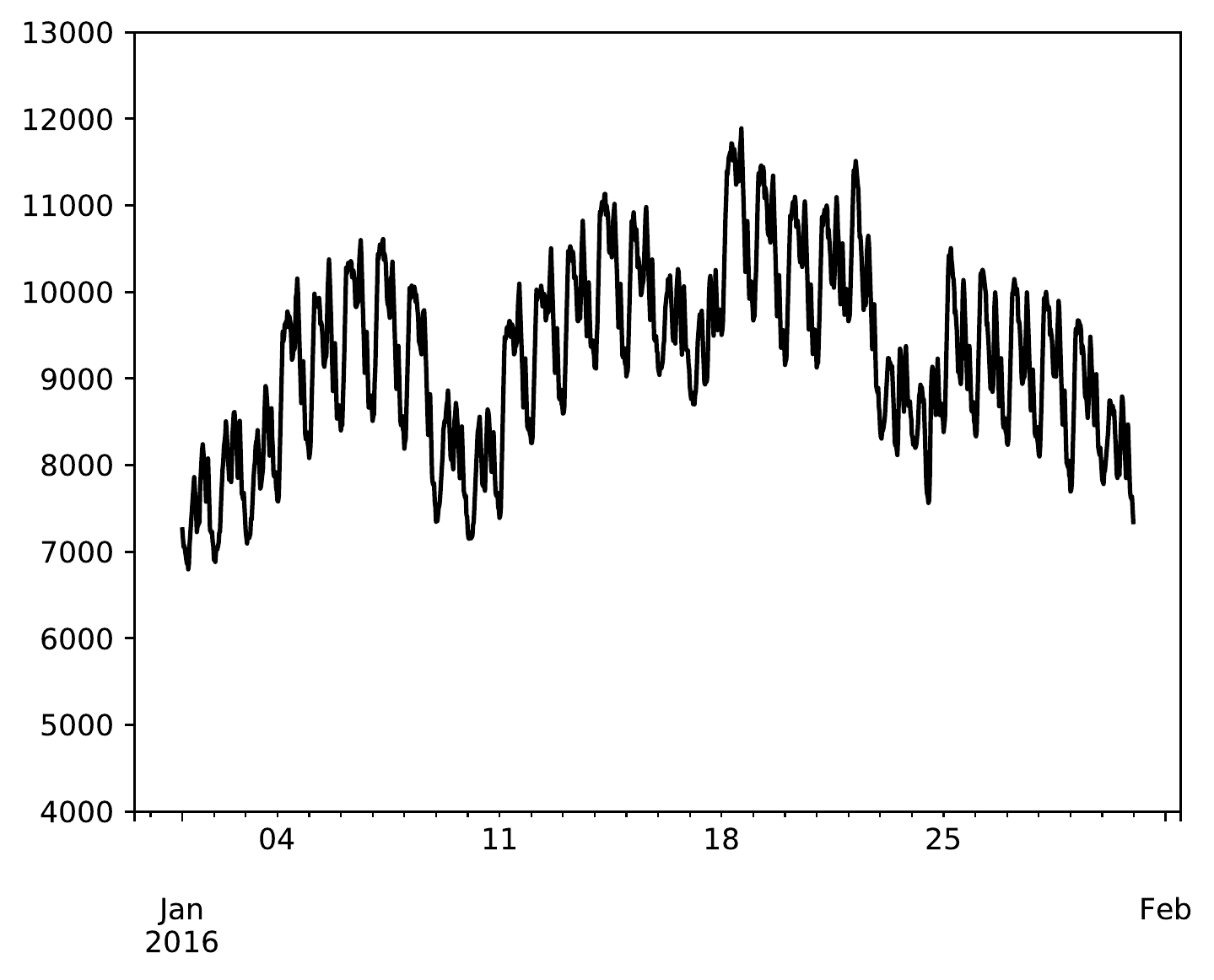} \hspace{-6pt}
    \includegraphics[width=\wlen]{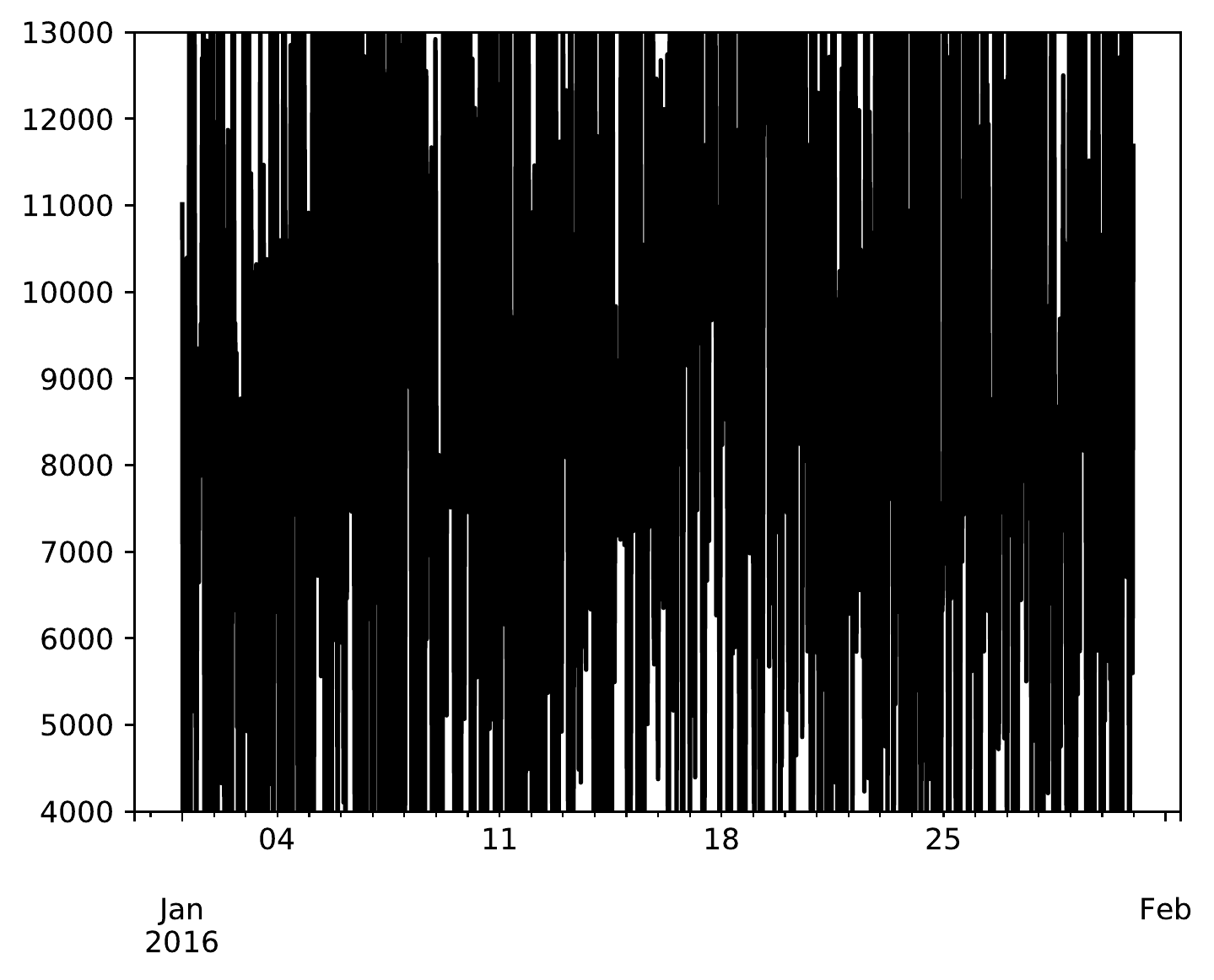} \hspace{-6pt}
    \includegraphics[width=\wlen]{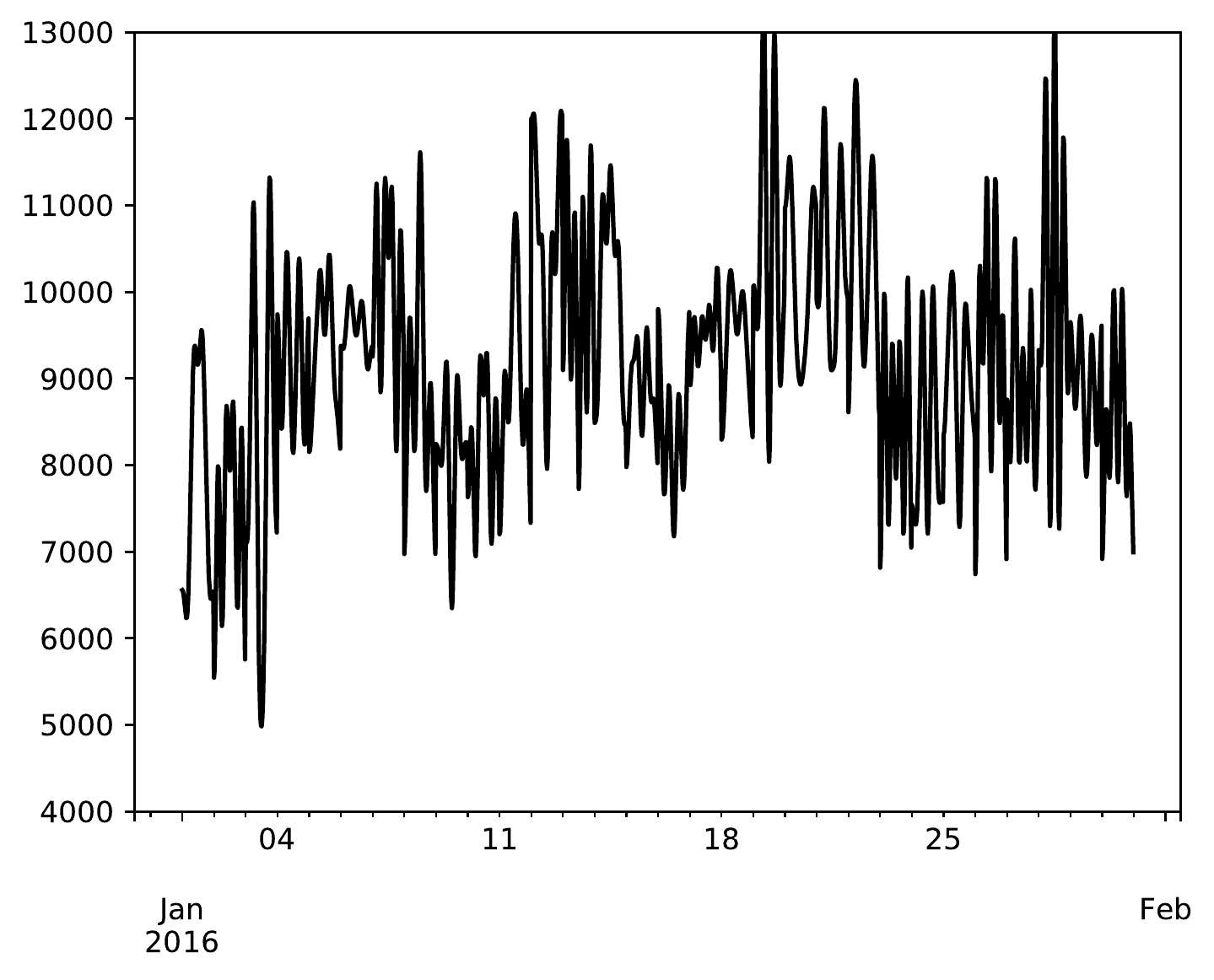}\hspace{-6pt}
    \includegraphics[width=\wlen]{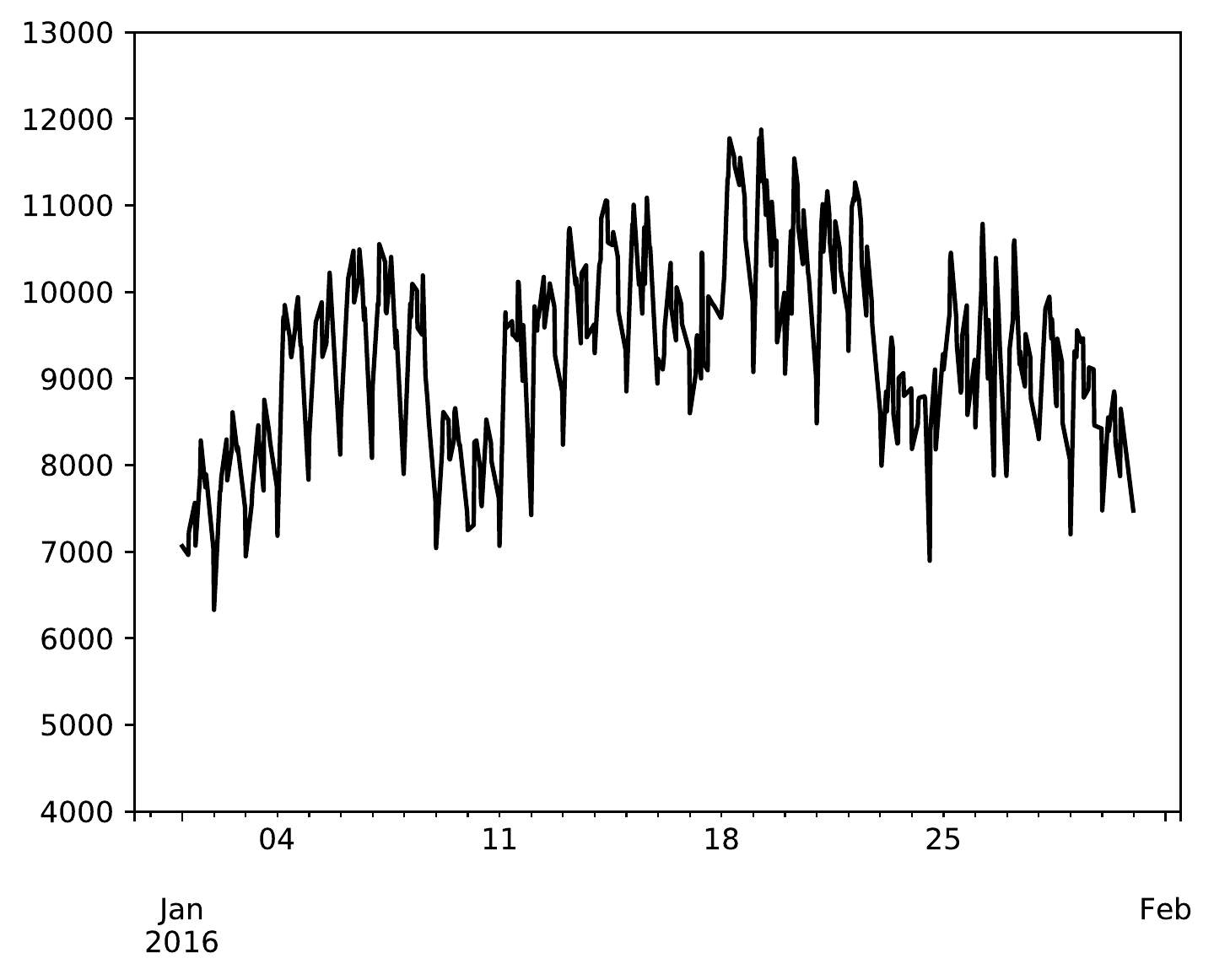}\\
    \caption{Real load consumption data for the Auvergne-Rh\^{o}ne-Alpes region in January 2016 (first column) and its private versions obtained using Laplace (second column), DFT (third column), and \algname{} (fourth column) with privacy budget $\epsilon=1$ (top),  $\epsilon=0.1$ (middle), and $\epsilon=0.01$ (bottom).
     }
    \label{fig:real_and_noisy_data_1} 
\end{figure}

Answering queries over contiguous $w$-periods corresponds to releasing
the private stream over the entire available duration.  Figure
\ref{fig:real_and_noisy_data_1} illustrates the real and private
versions of the data-stream for the Auvergne-Rh\^{o}ne-Alpes region in
January, 2016. It uses $w$-periods of size 48 for given privacy
budgets $\epsilon = 1.0$, $0.1$, and $0.01$, shown in the top, middle,
and bottom rows respectively.  Recall that the choice for the
$w$-period allows the data curator to ensure the protection of the
observed power consumptions within each
period. Thus, the released stream protects loads in each entire day.

The real loads are illustrated in the first column. The figure
compares our proposed \algname\ algorithm (fourth column) against the
\emph{Laplace mechanism} (second column), and the DFT algorithm
\cite{rastogi:10} (third column).

The experiments set the number of Fourier coefficients in the DFT and
sampling steps in \algname{} to $10$. The privacy budget allocated to
perform each measurement is split equally.  Additionally, for
\algname\, $\epsilon_s = \epsilon_p = \epsilon_o = \frac{1}{3}
\epsilon$, and the $L1$-sampling procedure uses a threshold value
$\theta$ of $1000$ (which is about one tenth of the average load
consumption in each region).  This information was publicly revealed.
Finally, \algname{} uses the following feature query set $\sF = \{\bF_1, \bF_2,
\bF_3\}$ in the optimization-based post processing step, with $\bF_1
\prec \bF_2 \prec \bF_3$.  $\bF_1$ is defined as described in Section
\ref{sec:post_process_procedure}; $\bF_2$ partitions each $w$-period
in $4$ sets, the intervals $[0, 14)$, $[14,24)$, $[24, 36)$, and
      $[36,48)$ that correspond to aggregated consumption for the
        following times of the day: [0am-7am), [7am-12pm), [12pm-6pm),
              and [6pm-0am) respectively. $\bF_3$ partitions each
                $w$-period in a single set, listing all the time steps
                within the $w$-period and thus describing the
                aggregated daily energy consumption. The query set
                represents salient moments in the day associated with
                different consumption patterns. These are proxy of
                consumer behaviors and thus energy
                consumption. Because these queries return private
                answers, the privacy is guaranteed by the
                post-processing immunity of DP (see Theorem
                \ref{th:postprocessing}).  Finally, if an algorithm
                reports negative noisy value for a stream point, we
                truncate it to zero.

Figure \ref{fig:real_and_noisy_data_1} clearly illustrates that, for a
given privacy disclosure level, \algname\ produces private streams
that are more accurate than its competitors when visualized.  The next
paragraph quantifies the error reported by the algorithms.

\def \wlen{.32\linewidth}
\def \wlenLabel{.20\linewidth}
\begin{figure}[!t]
    \centering
    \includegraphics[width=180pt]{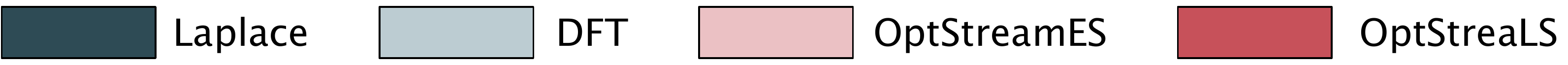}\\
    \includegraphics[width=\wlen]{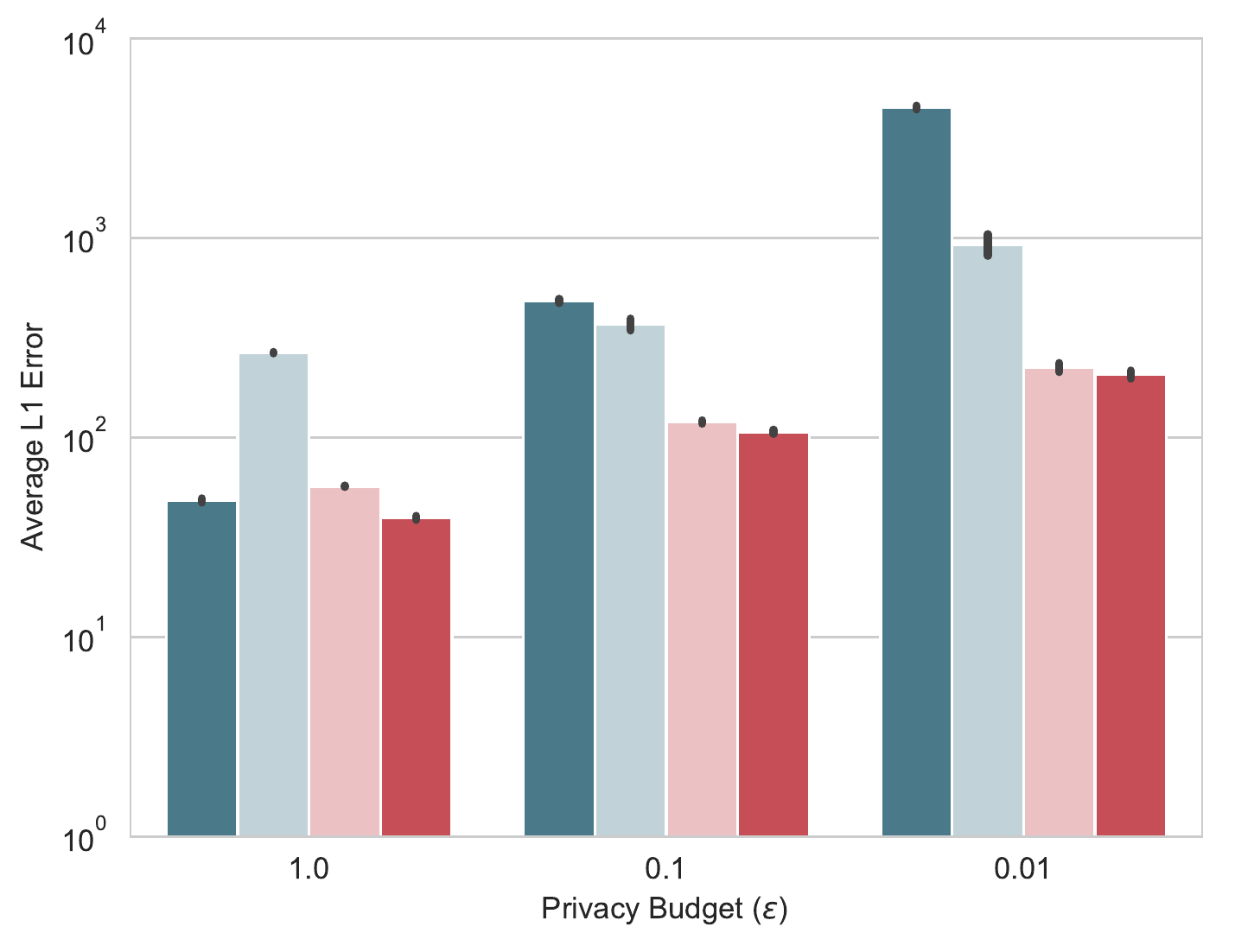}
    \includegraphics[width=\wlen]{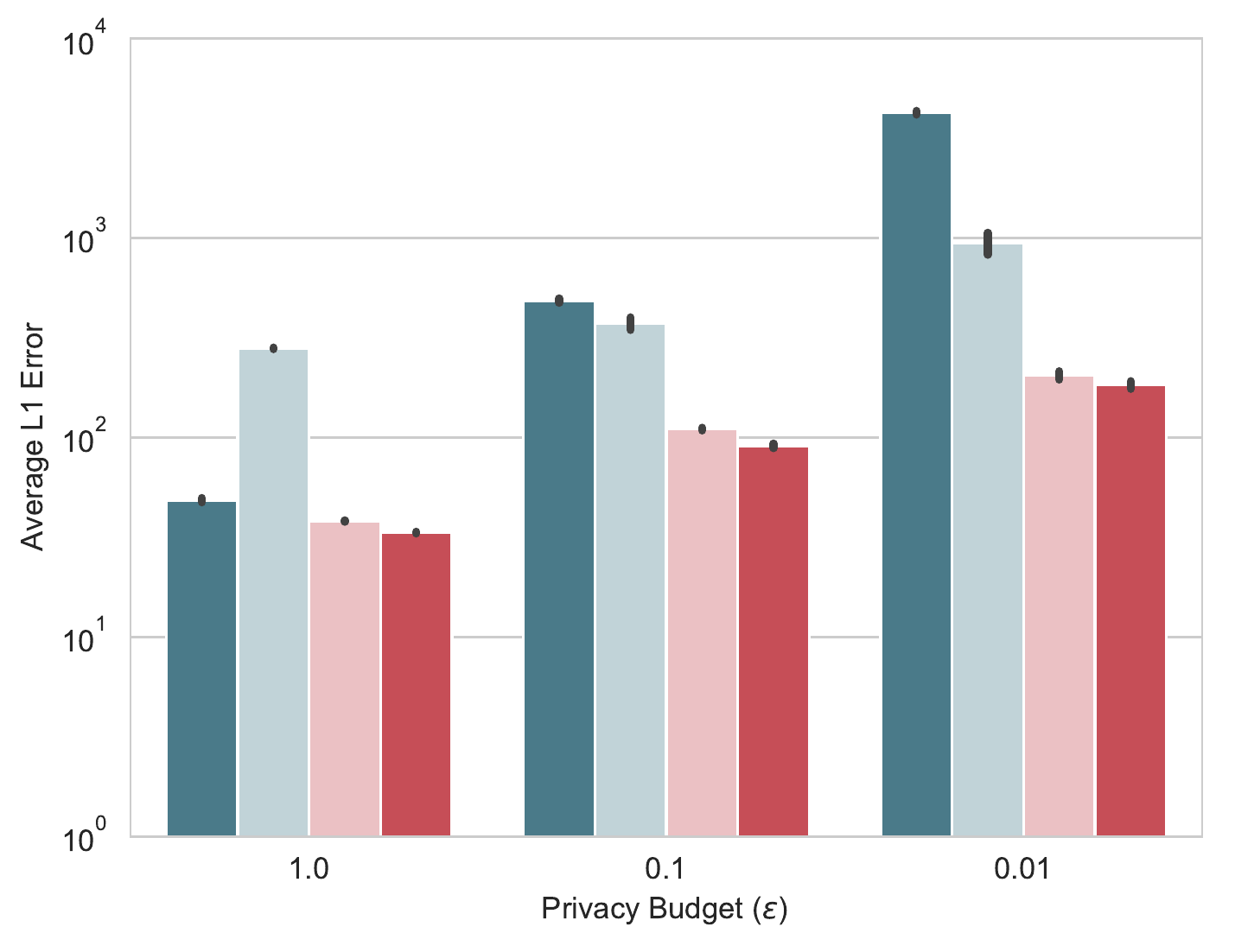}
    \includegraphics[width=\wlen]{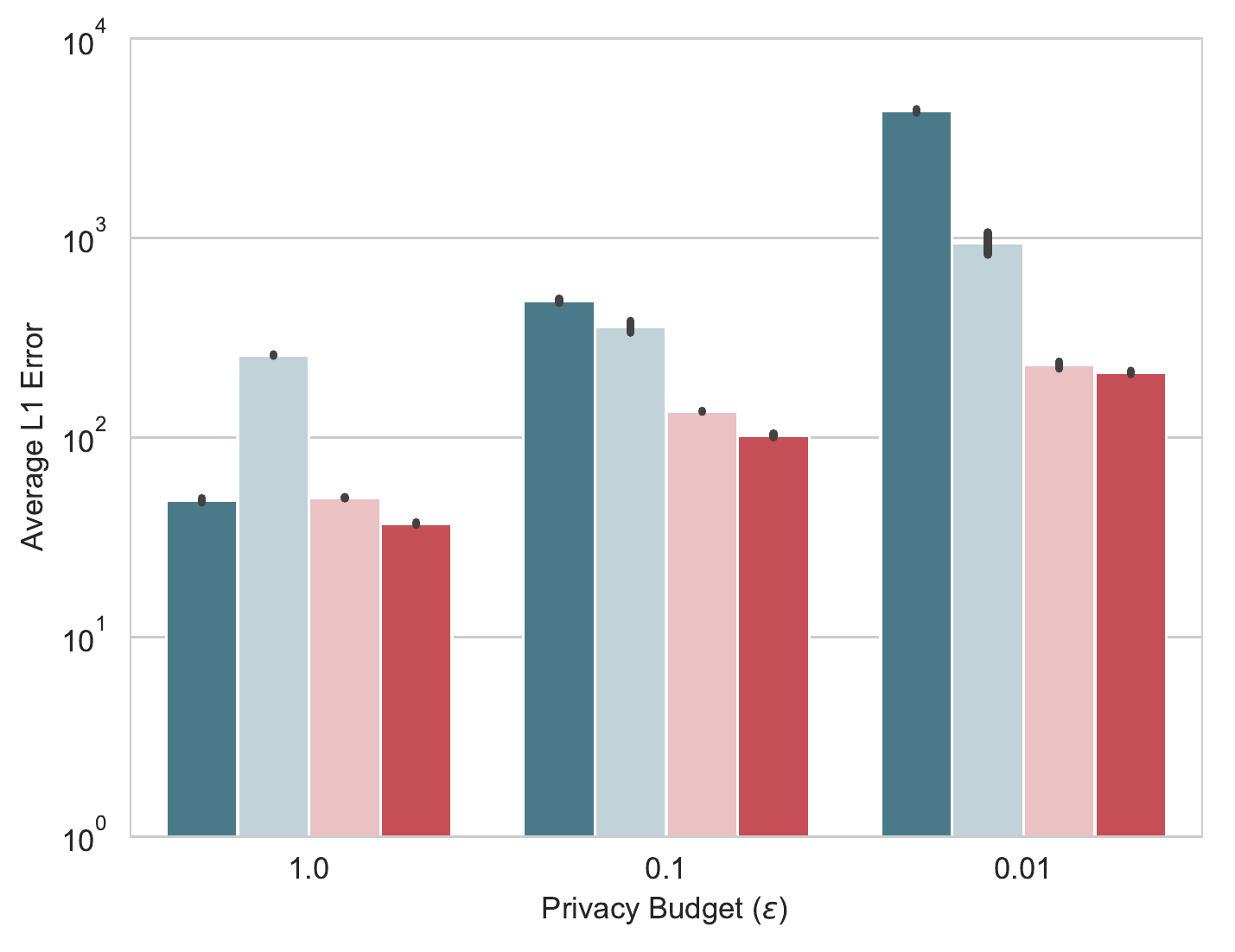}\\
    {\small
    \hspace{60pt}February \hfill ~~~~June~~~~ \hfill October\hspace{50pt}}

    \caption{$L_1$-error analysis: Load stream data for the months of February (left),
     June (middle), and October (right). The $y$-axis reports $\log_{10}$ of the average $L_1$-error for all the stream data streams $R \in \sR$.
     }
    \label{fig:L1_error_1}
\end{figure}

\subsubsection{Average $L_1$-Error Analysis} 

This section reports the average $L_1$-error for each $w$-period
produced by the algorithms. For each input data stream $\bx_R$
associated to a region $R \in \sR$ (illustrated in Table
\ref{tab:france_regions}) and each reported private stream
$\hat{\bx}_R$, we compare the average $L_1$-error defined as:
$\frac{1}{N_R} \| \hat{\bx}_R - \bx_R\|_1$, where $N_R$ is the length
of the data stream associated with region $R$.  Figure
\ref{fig:L1_error_1} reports the average errors across all streaming
regions $r \in \sR$ for the months of February (left), June (middle),
and October (right). The consumptions in these three months capture
different customers load profile behaviors due to different weather
patterns and durations of the day light. Each histogram reports the
log$_{10}$ value of the average error of 30 random trials. Two version
of \algname{} ({\algname ES} and {\algname LS}) are presented and
correspond to the qqually-spaced sampling, and the $L1$-sampling
procedures, respectively.  While all algorithms induce a notable
$L_1$-error which increases as the privacy budget decreases, the
figure highlights that \algname{} consistently outperforms competitor
algorithms. Additionally, \algname{} with $L1$-sampling is found to
outperform its equally-spaced sampling counterpart, especially for
large privacy budgets. For small privacy budgets, the two versions of
the algorithm tend to perform similarly. This is due to the fact that
the $L1$-score becomes less accurate as the amount of noise increases.

\subsubsection{Hierarchical Private Data-Stream Release}

This section evaluates the extensions proposed in Section
\ref{sec:algorithm_ext} for releasing aggregated queries over
hierarchical data streams.  We answer the following queries over
contiguous $w$-periods for the whole duration of the stream: count
queries over the data stream $\bx(R) = x_1(R), x_2(R), \ldots$ for
each region $R \in \sR$ listed in Table \ref{tab:france_regions}, as
well as count queries $\bx = x_1, x_2, \ldots $, where each $x_t =
\sum_{R \in \sR} x_t(R)$ represents the aggregated load consumption at
national level.  Thus, we create a hierarchy of two levels and answer
simultaneously all queries. We allocate a privacy budget of
$\frac{\epsilon}{2}$ to each level of the hierarchy. \algname{} is
compared against the Laplace mechanism applied to each stream data,
using a uniform allocation of the privacy budget
($\frac{\epsilon}{2}$) for each query in a different level of the
stream hierarchy and the DFT algorithm which answers queries over each
data stream by uniformly allocating a portion of the privacy budget at
each level of the hierarchy.

\def \wlen{.32\linewidth}
\def \wlenLabel{.28\linewidth}
\begin{figure}[!t]
    \centering
    \includegraphics[width=180pt]{legend_1.pdf}\\
    \includegraphics[width=\wlen]{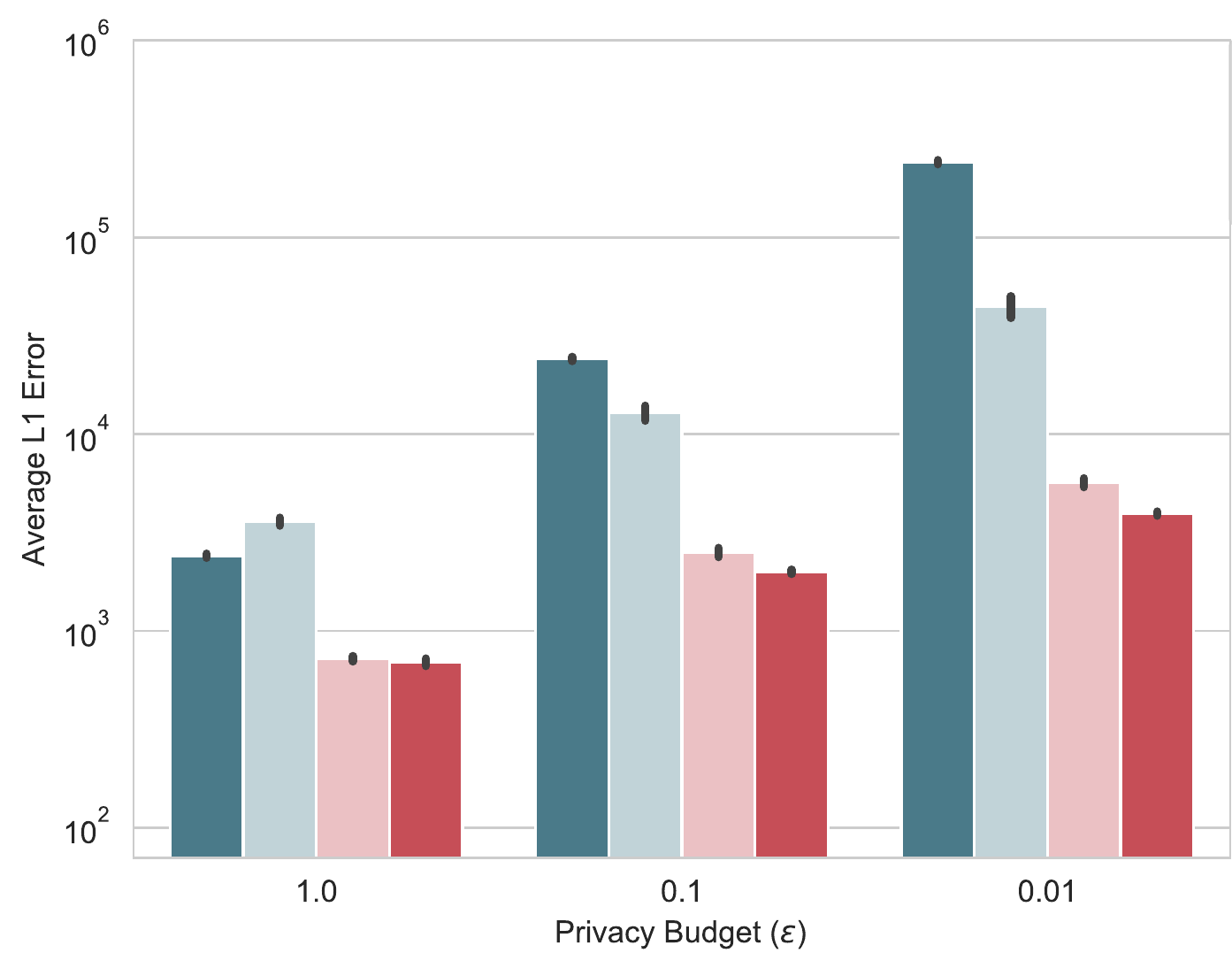}
    \includegraphics[width=\wlen]{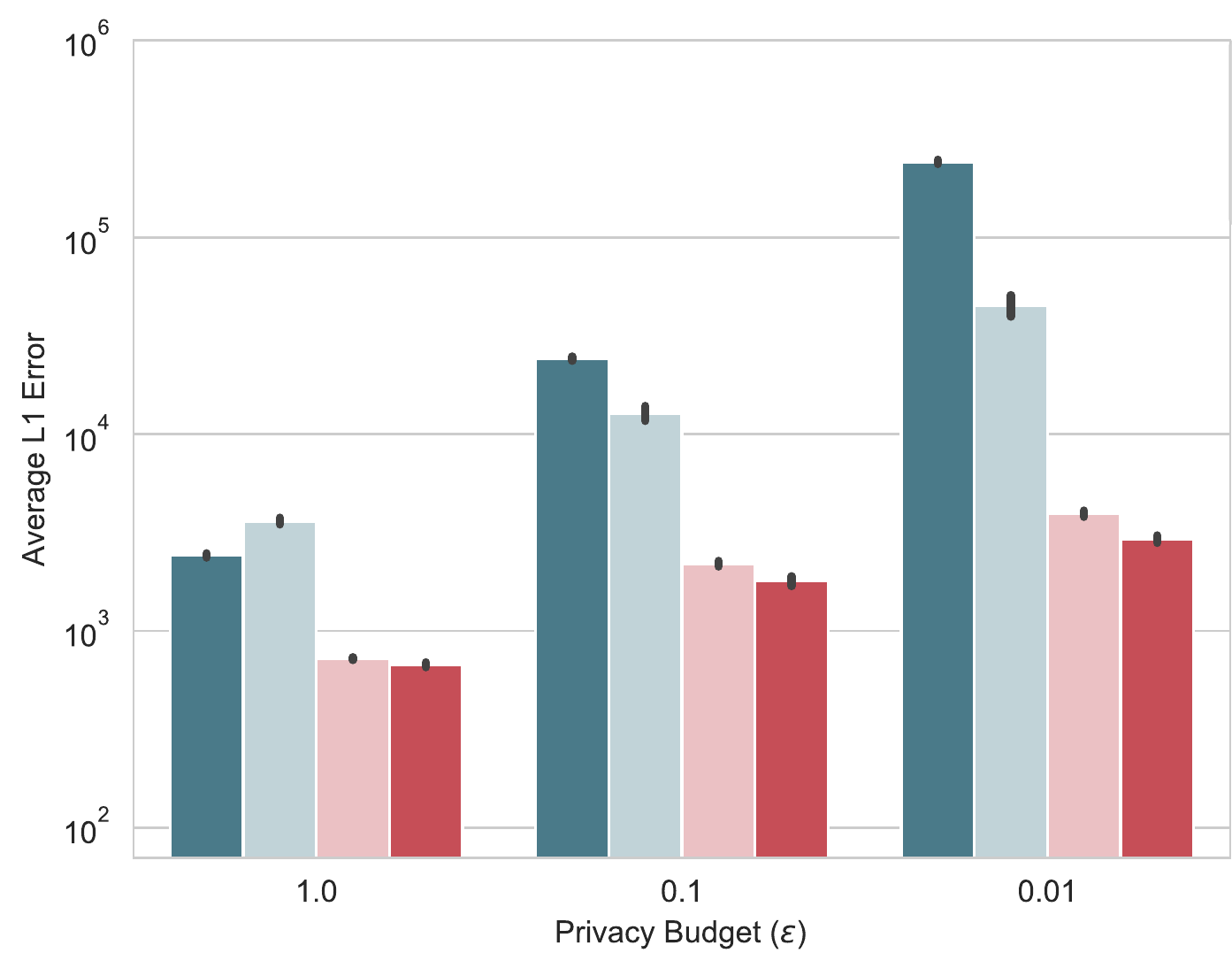}
    \includegraphics[width=\wlen]{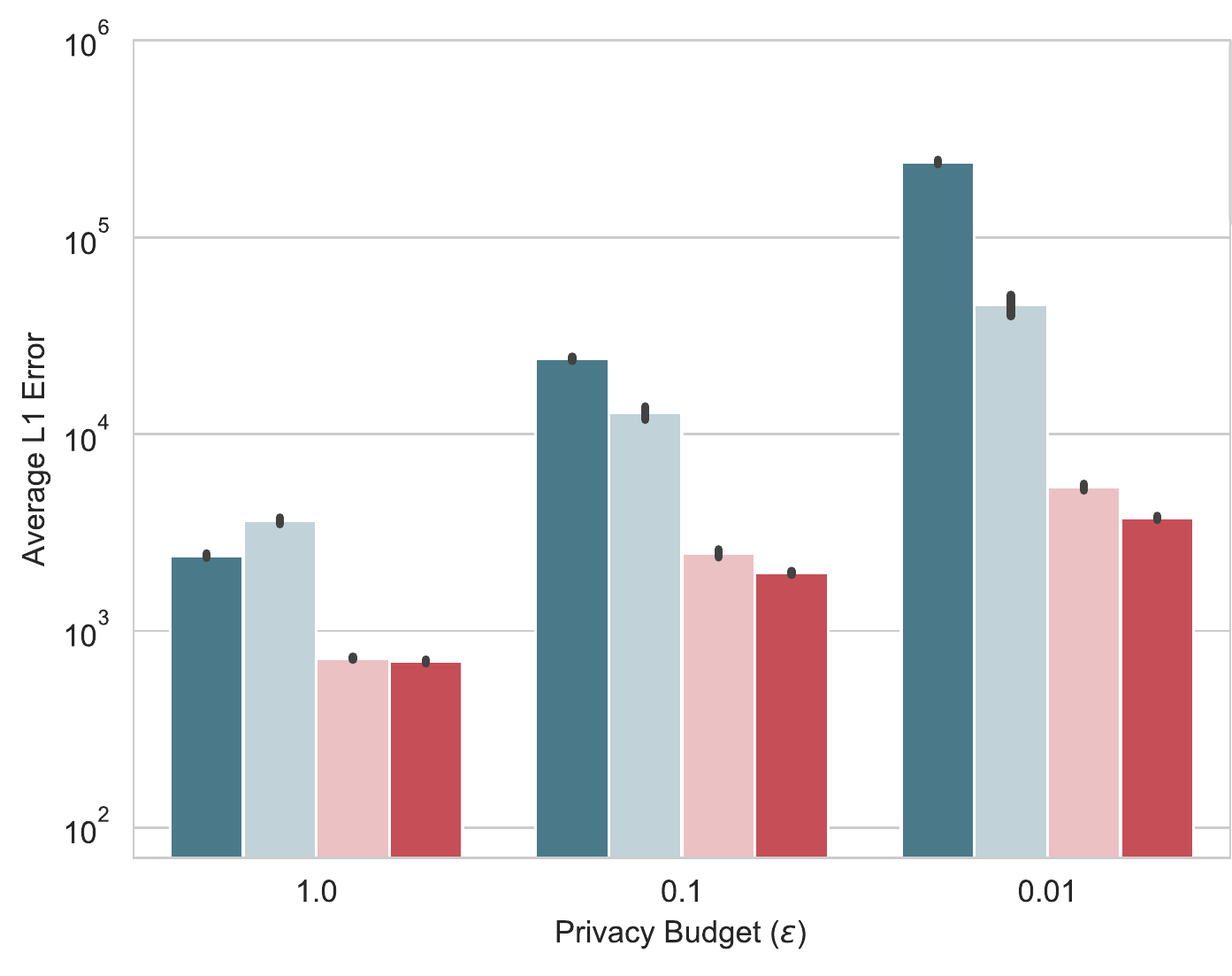}\\
    \includegraphics[width=\wlen]{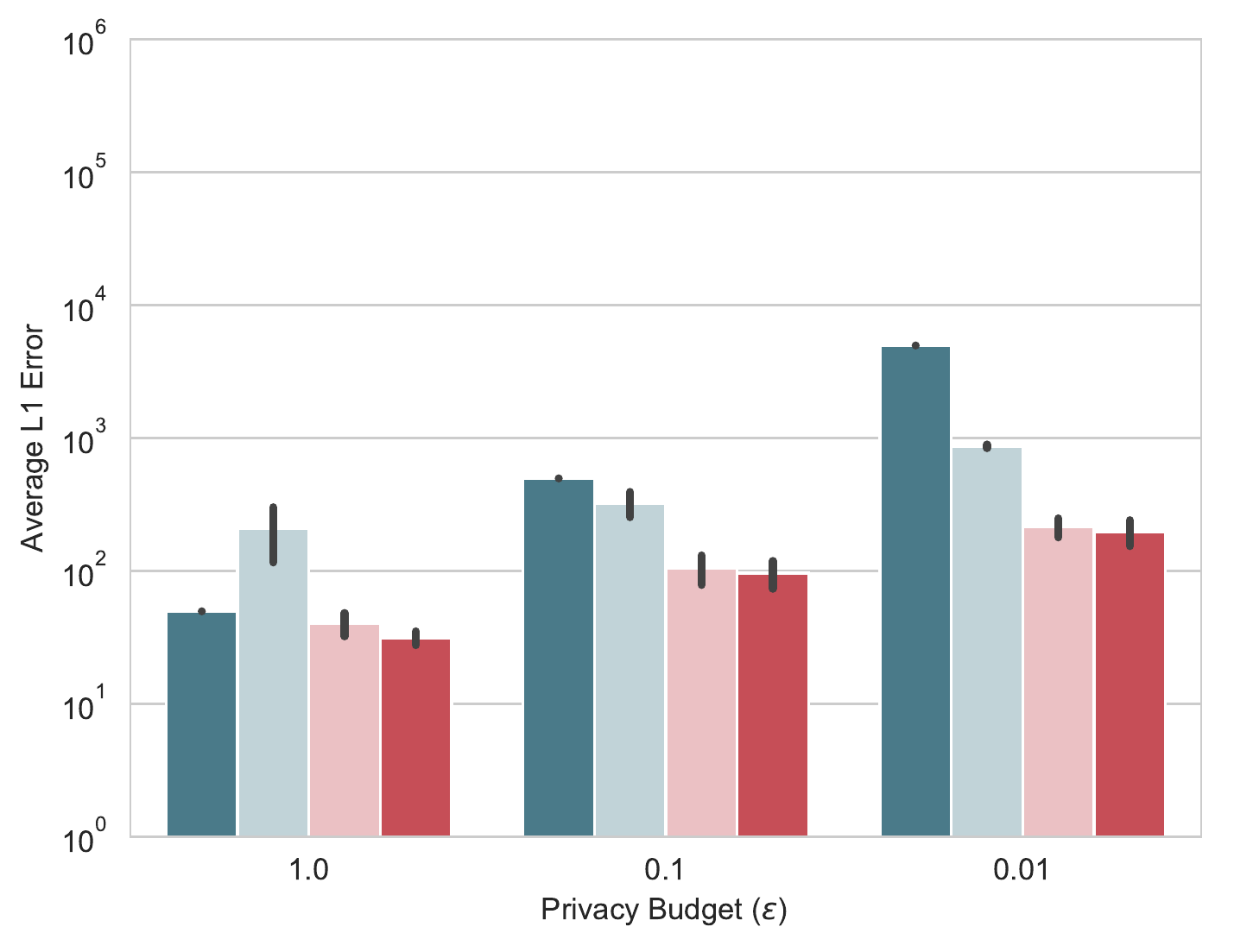}
    \includegraphics[width=\wlen]{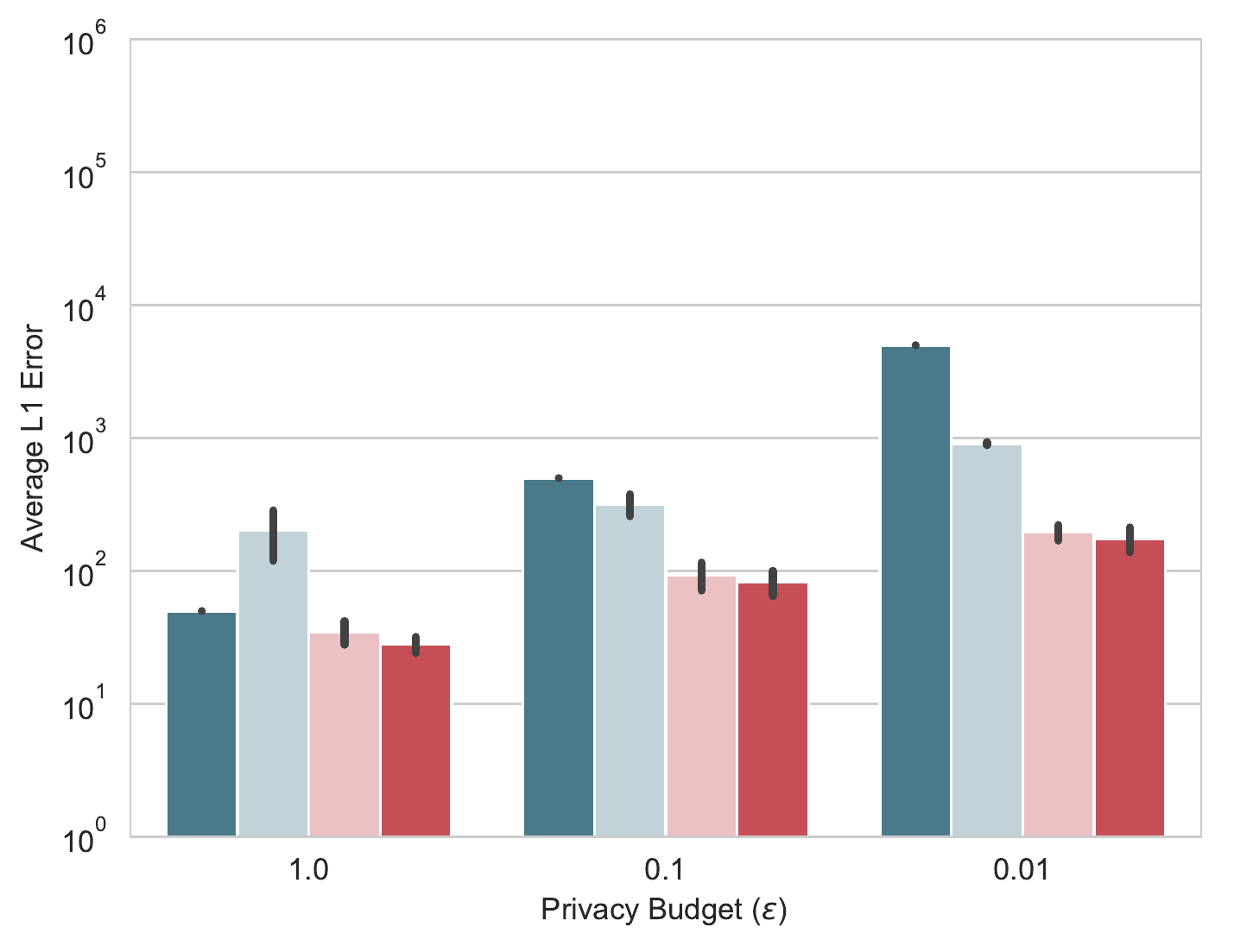}
    \includegraphics[width=\wlen]{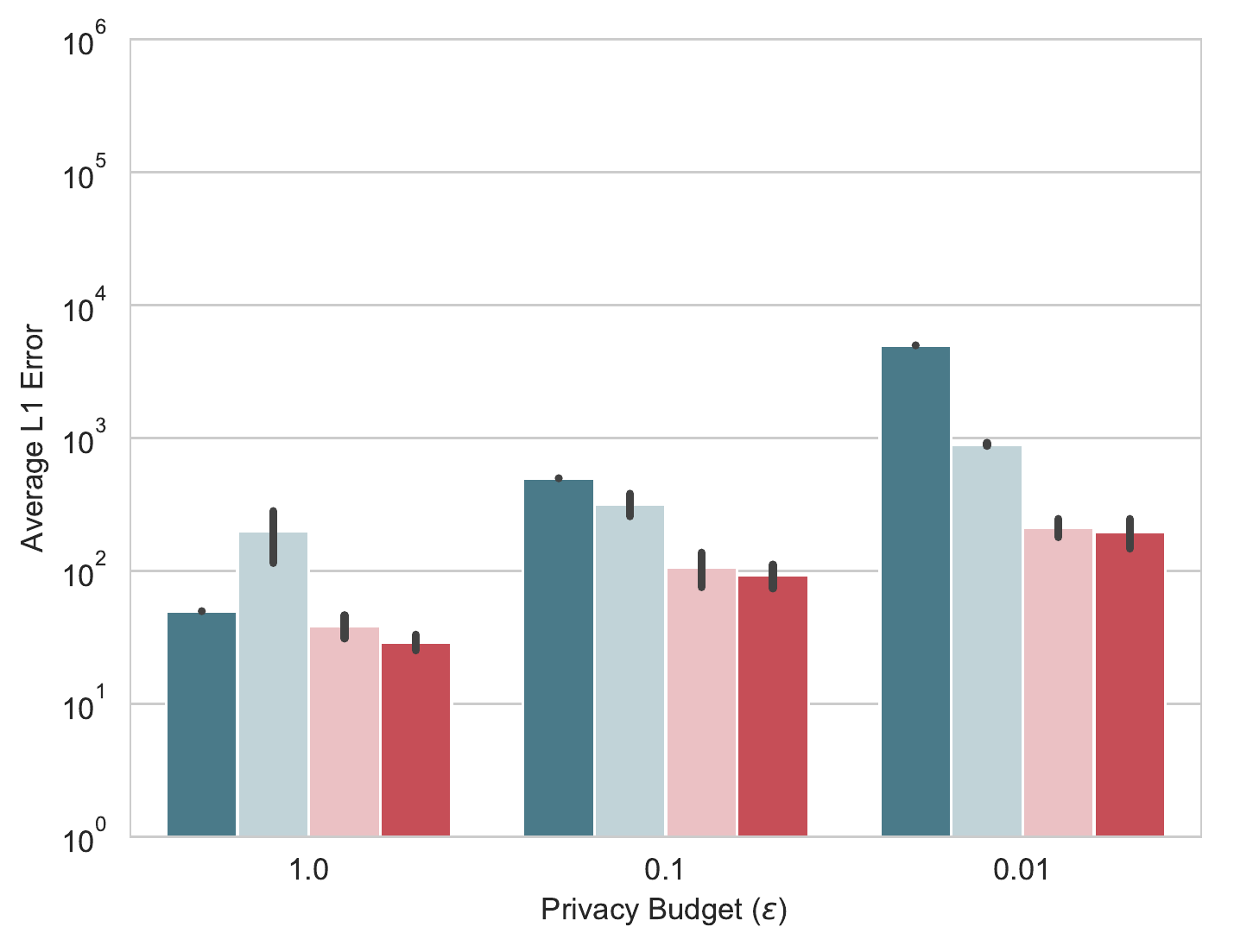}\\
    {\small
    \hspace{60pt}February \hfill ~~~~June~~~~ \hfill October\hspace{50pt}}

    \caption{$L_1$-error analysis: Hierarchical energy load stream data for the months of February, June, and October. 
    The $y$-axis reports $\log_{10}$ of the average $L_1$-error 
    at national level (top) and at the regional level (bottom).
     }
    \label{fig:hierarchical_L1_error}
\end{figure}

Figure \ref{fig:hierarchical_L1_error} shows the results for three
different months of the year: February (left), June (middle), and
October (right), and under different indistinguishability parameters
$\alpha$ for privacy budget $\epsilon=1.0$. The top row of the figure
reports the average $L_1$-errors when releasing stream data $\bx(R)$
associated with each region $R$, while the bottom row gives the
$L_1$-errors when releasing the stream data $\bx$ at national level.
Each histogram reports the log$_{10}$ value of the average error of 30
random trials.  The results illustrates similar trends to those in
previous experiments. Overall, \algname\ with the adaptive
$L1$-sampling produces stream data with the lowest average
$L_1$-errors for both levels of the stream hierarchy.  In addition to
the improved error it is important to note that \algname\ ensures the
consistency of the values of the private stream in the hierarchy,
i.e., for each time step, the reported sum of the loads across all
regions equals the reported load at national level. Neither the
Laplace mechanism nor DFT do ensure such property.

\subsubsection{Impact of Privacy on Forecasting Demand}

The final results evaluate the capability of the released private
streams to accurately predict future consumptions.  To do so, we adopt
the Autoregressive Moving Average (ARMA) model
\cite{alwan1988time,zhang2003time,cochrane2005time,hipel1994time}. ARMA
is a popular stochastic time series model used for predicting future
points in a time series (forecast). It combines an Autoregressive (AR)
model \cite{alwan1988time} and a Moving Average (MA) model
\cite{alwan1988time,cochrane2005time}, i.e., ARMA($p, q$) combines
AR($p$) and MA($q$) and is suitable for univariate time series
modeling.  In an AR($p$) model, the future value of a variable $x_t$
is assumed to be a linear combination of the past $p$ observations and
a random error: $x_t = c + \sum_{i=1}^p \phi_i x_{t-i} + \beta_t$,
where $c$ is constant, $\beta_t$ is a random variable modeling white
noise at time $t$, and the $\phi_i (i=1,\ldots,p)$ are model
parameters. A MA($q$) model uses the past $q$ errors in the time
series as the explanatory variables. It estimates a variable $x_t$
using $\mu + \beta_t + \sum_{i=1}^{q} \theta_i \beta_{t-i}$, where the
$\theta_i (i=1,\ldots, q)$ are model parameters, $\mu$ is the
expectation of $x_t$, and the $\beta_t$ terms are white noise error
terms.  The ARMA model with parameters $p$ and $q$ refers to the model
with $p$ autoregressive terms and $q$ moving-average terms: It
estimates a future time step value $x_t$ as $c + \beta_t +
\sum_{i=1}^p \phi_i x_{t-i} + \sum_{i=1}^q \theta_i \beta_{t-i}$. In
our experiments, we use an ARMA model with parameters $p=q=1$ to
estimate the future 48 time steps (corresponding to a day) when
trained with the past four weeks of the private data stream estimated
using Laplace, DFT, and \algname\ with $L1$-sampling.  All models use
the same parameters adopted in the previous sections.

\def \wlen{.245\linewidth}
\def \wlenLabel{.10\linewidth}
\begin{figure}[!t]
    \centering
    \includegraphics[width=\wlen]{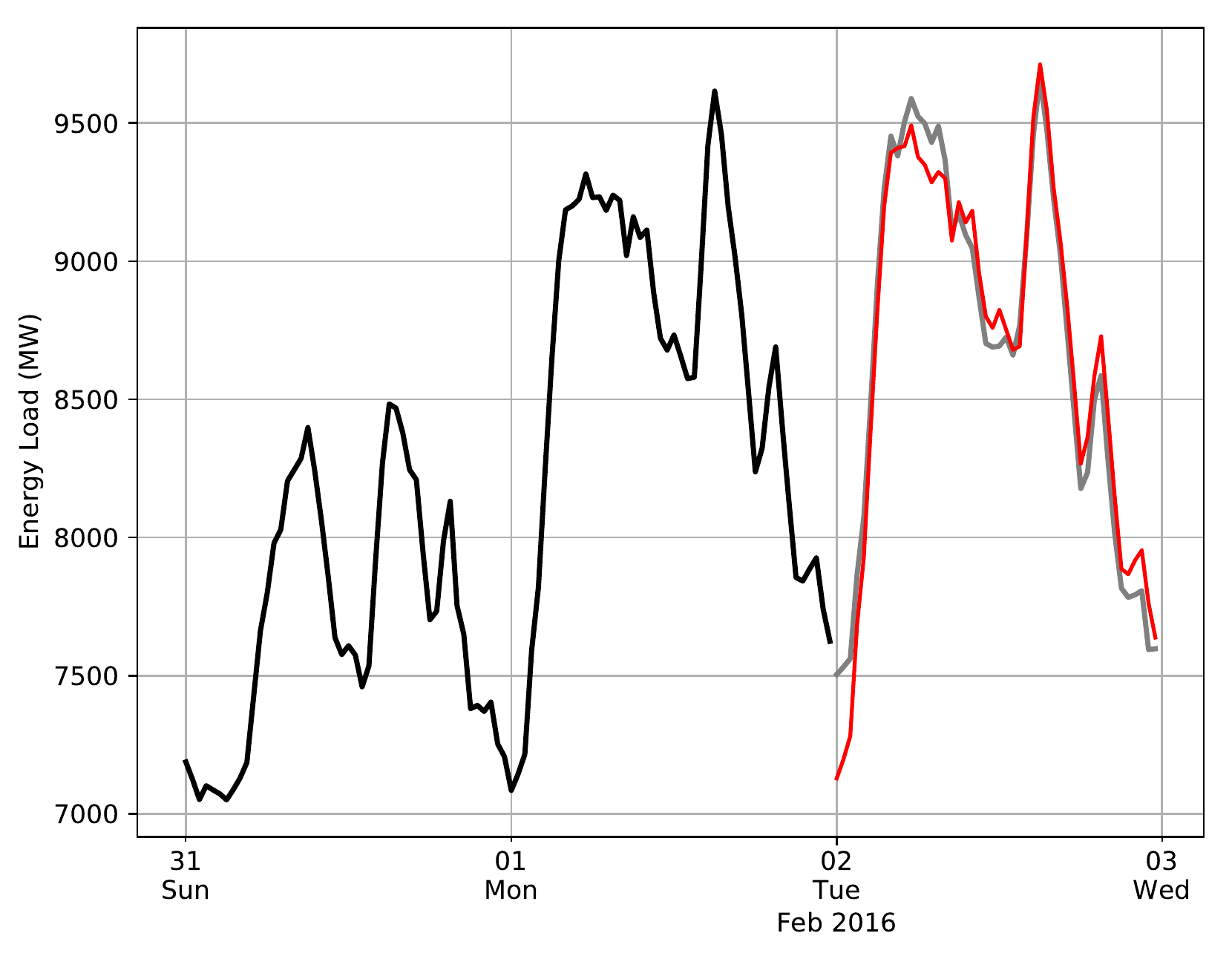}
    \includegraphics[width=\wlen]{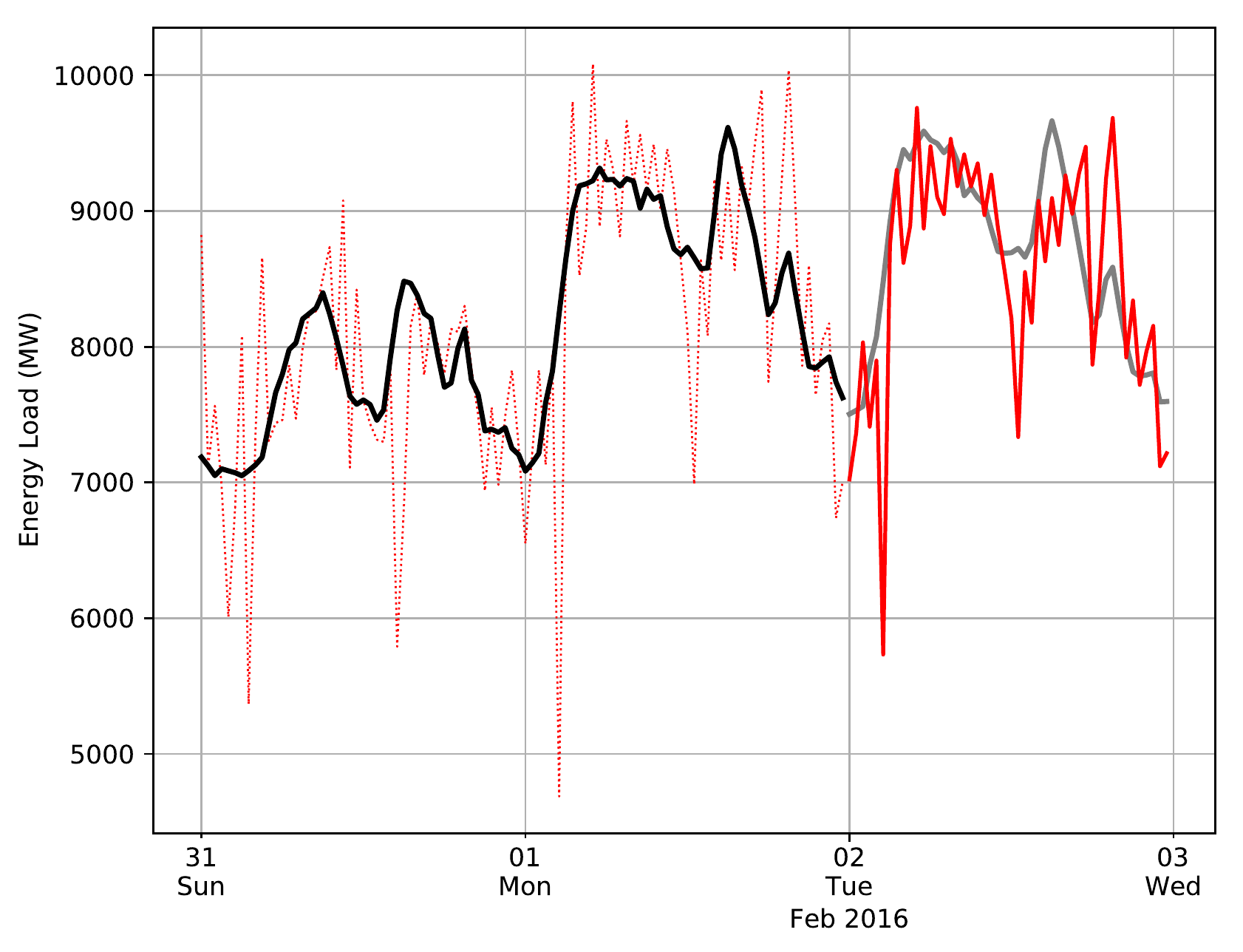}
    \includegraphics[width=\wlen]{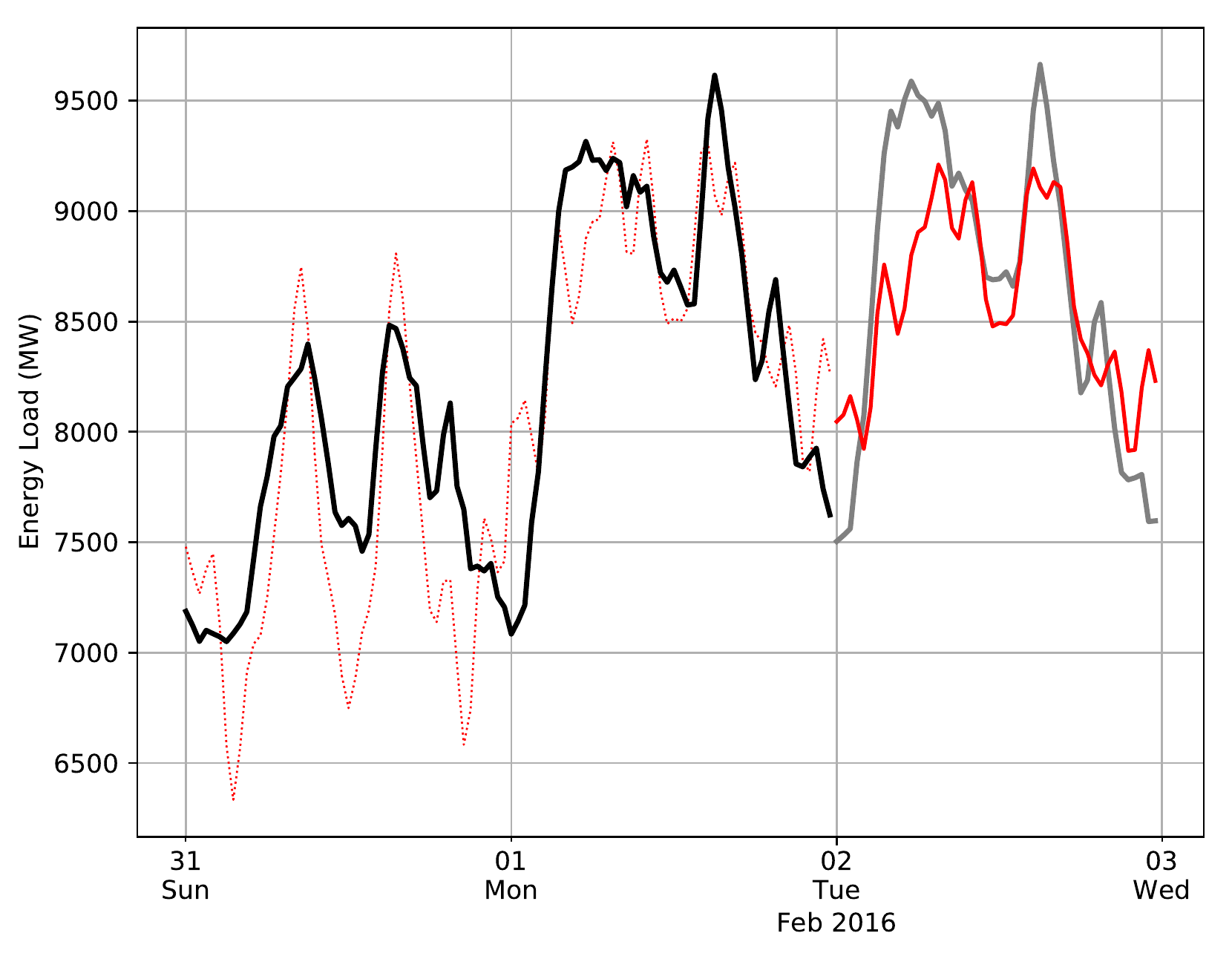}
    \includegraphics[width=\wlen]{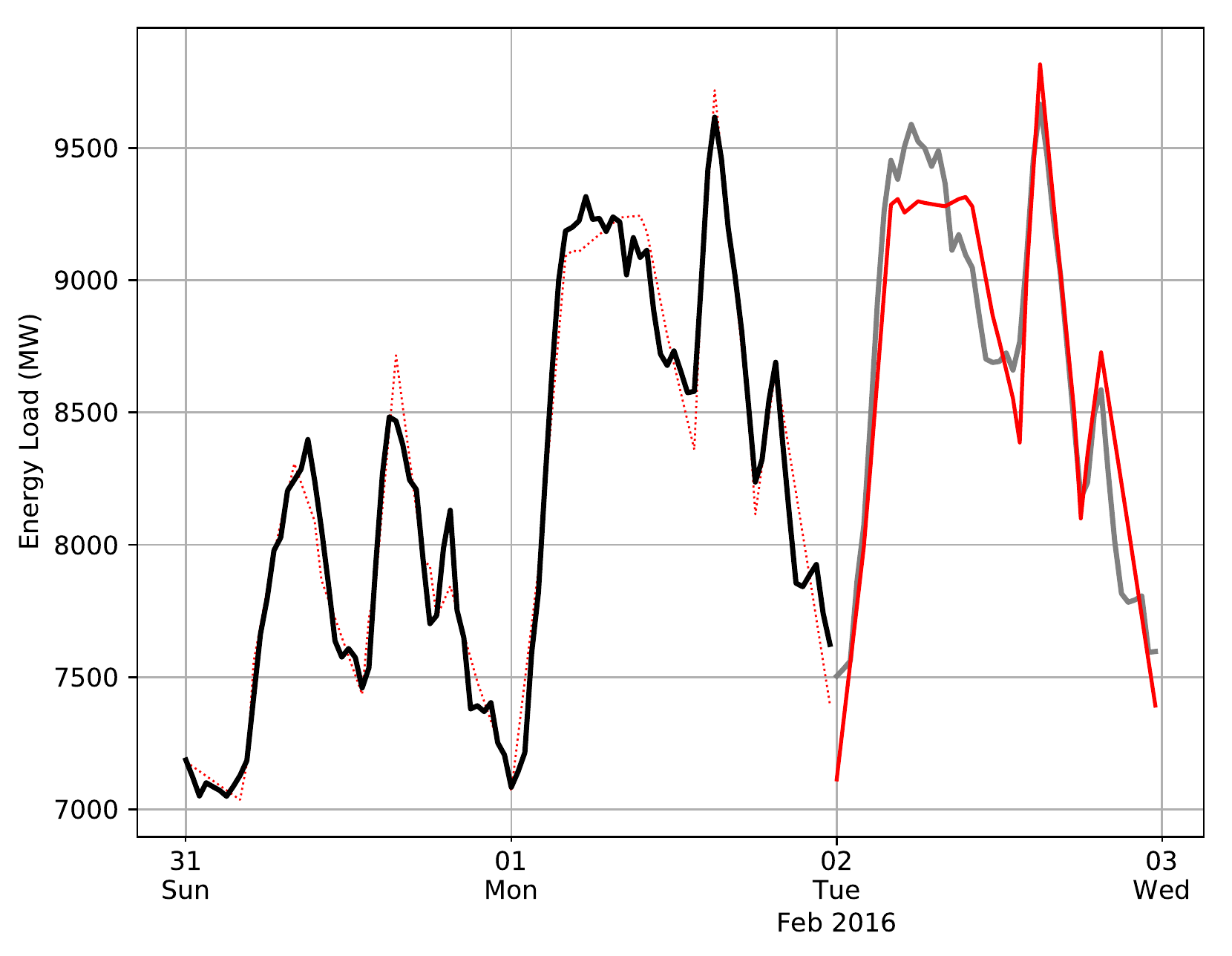}\\ 
    {\small\centering
    \hspace{35pt}
    ~~Real~ \hspace{75pt}
    Laplace \hspace{70pt}
    ~~DFT~~ \hspace{60pt}
    \algname
    }
    \caption{Prediction error: Forecast for a one day load consumption through an ARMA model on the real load consumption data (Real) and its private versions obtained using Laplace, DFT, and \algname{} with privacy budget $\epsilon=0.1$.
     }
    \label{fig:predictions}
\end{figure}

Figure \ref{fig:predictions} visualizes the forecast for the load
consumptions in the Auvergne-Rh\^{o}ne-Alpes region for February 2,
2016.  The black and gray solid lines illustrates, respectively, the real load values observed so far and those of the day to be forecasted.
The dotted red lines illustrates the private stream data estimated so
far (and used as input to the prediction model) and the solid red
lines depict the prediction obtained using the ARMA model.  Figure
\ref{fig:predictions} shows the forecast results using the real
data (Real) and the private stream obtained through Laplace, DFT, and
\algname, respectively. The figure clearly shows that \algname{} is
able to visually produce better estimates for the next day forecast.

\medskip
\def \wlen{.32\linewidth}
\def \wlenLabel{.22\linewidth}
\begin{figure}[htb]
    \centering
    \includegraphics[width=180pt]{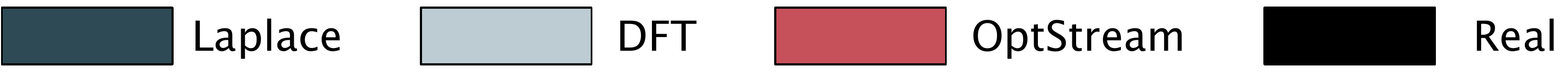}\\
    \includegraphics[width=\wlen]{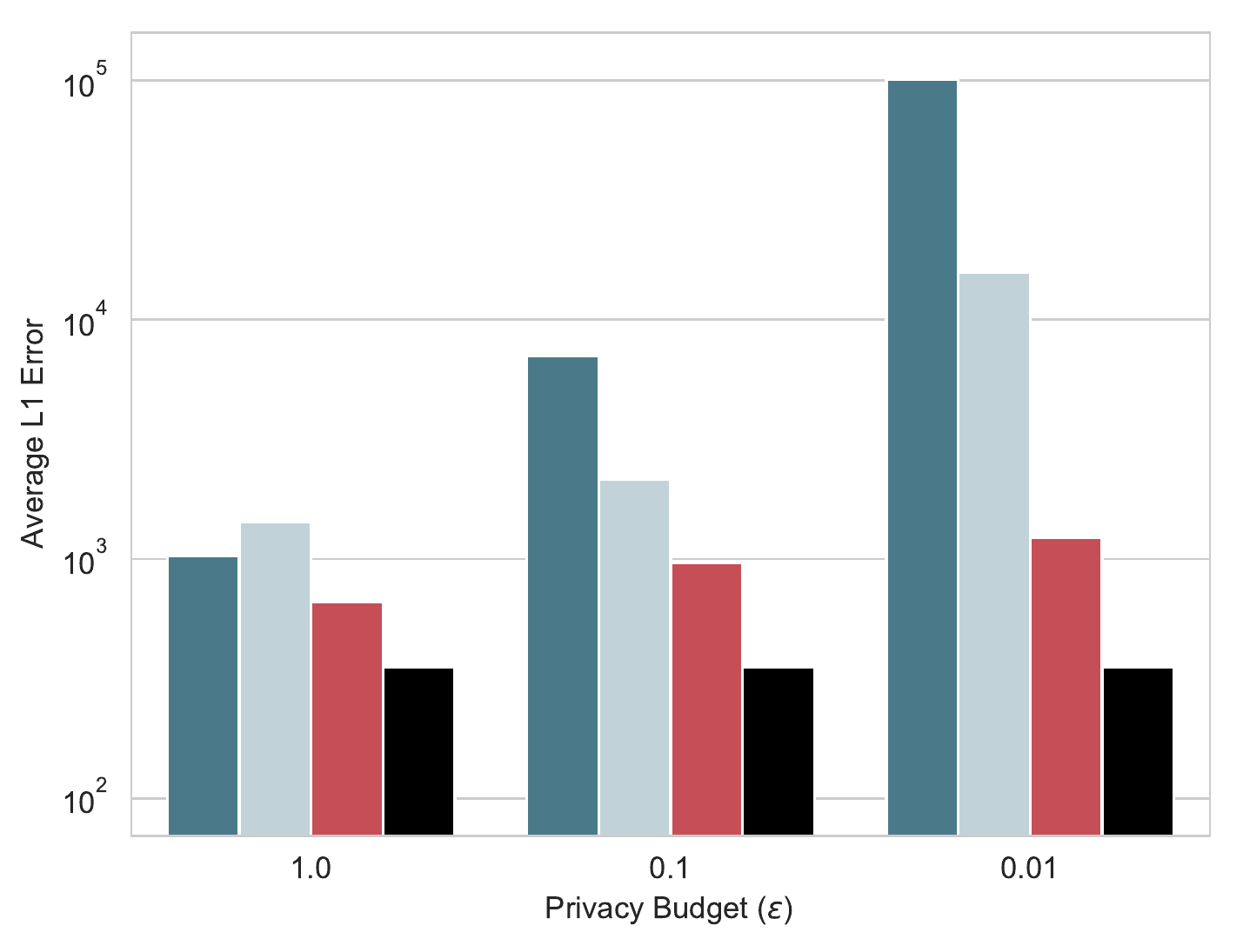}
    \includegraphics[width=\wlen]{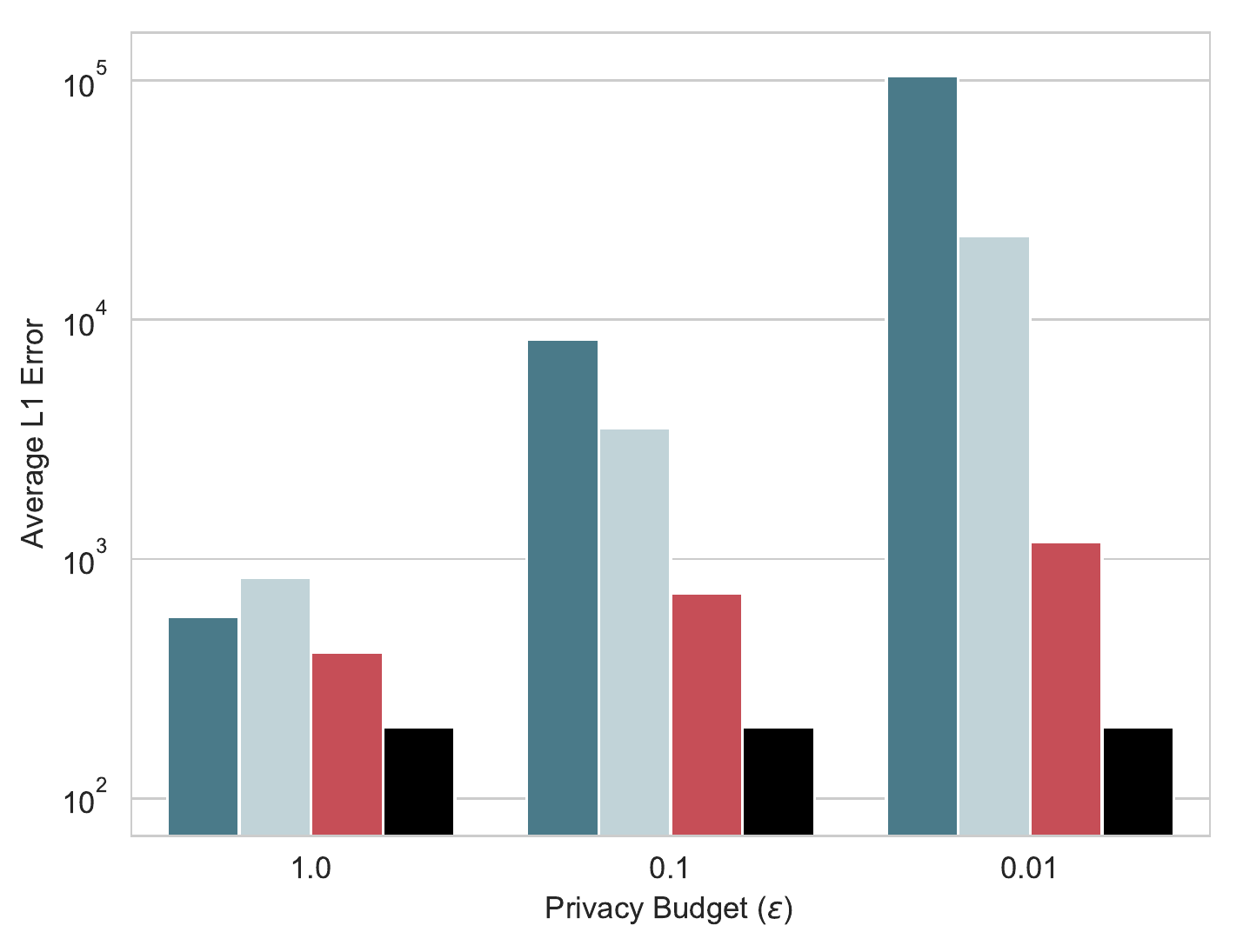}
    \includegraphics[width=\wlen]{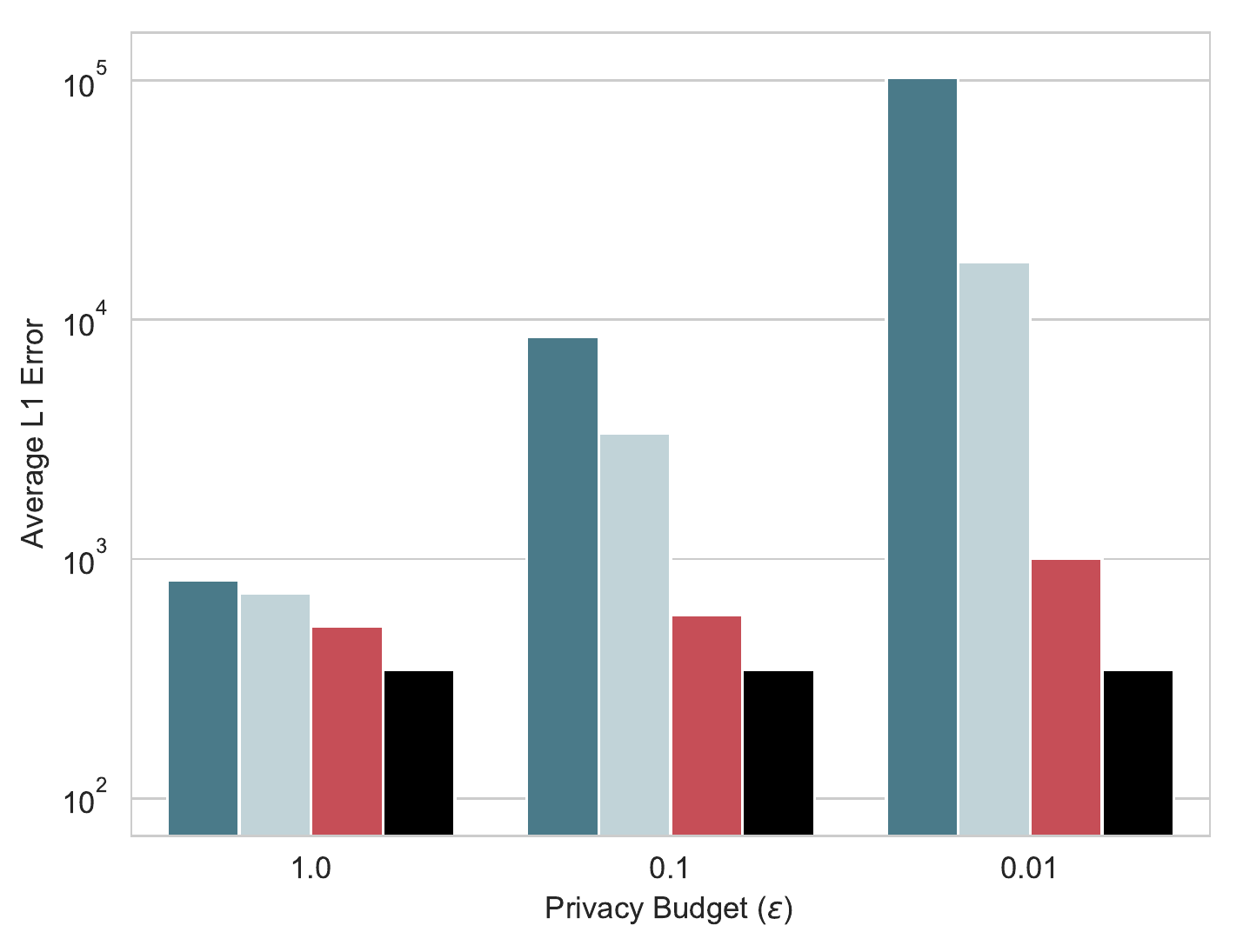} \\
    {\small
    \hspace{60pt}February \hfill ~~~~June~~~~ \hfill October\hspace{50pt}}
    \caption{$L_1$ error analysis: ARMA forecasting model on stream data for the energy loads of the months of February, June, and October for the Auvergne-Rh\^{o}ne-Alpes region. 
     }
    \label{fig:prediction_L1_errors}
\end{figure}

We also quantitatively evaluate the average $L_1$-error for each
prediction produced by the mechanisms. We adopt the same setting as
above for the prediction and report, in Figure
\ref{fig:prediction_L1_errors}, the average $L_1$-error for predicting
each day in the month of February, June, and October for the
Auvergne-Rh\^{o}ne-Alpes region.  Each histogram reports the
log$_{10}$ value of the average error of 30 random trials.  We observe
that \algname\ reports the smallest errors compared to all other
privacy-preserving algorithms, and that the error made by \algname\ in
reporting the next day forecast is closer to the error made in the
forecast prediction using the real data than when using another
method.

\subsection{Evaluation of \algname{} Individual Components}

\def \wlen{.19\linewidth}
\def \wlenLabel{.28\linewidth}
\begin{figure}[!t]
    {\centering\footnotesize
    \hspace{25pt} 
    Real \hspace{60pt}
    $\langle P, \circ , \circ \rangle$ \hspace{50pt}
    $\langle P, \circ , O \rangle$ \hspace{45pt}
    $\langle P, S , \circ \rangle$ \hspace{50pt}
    $\langle P, S , O \rangle$  
    }\\
    \centering
    \includegraphics[width=0.99\linewidth]{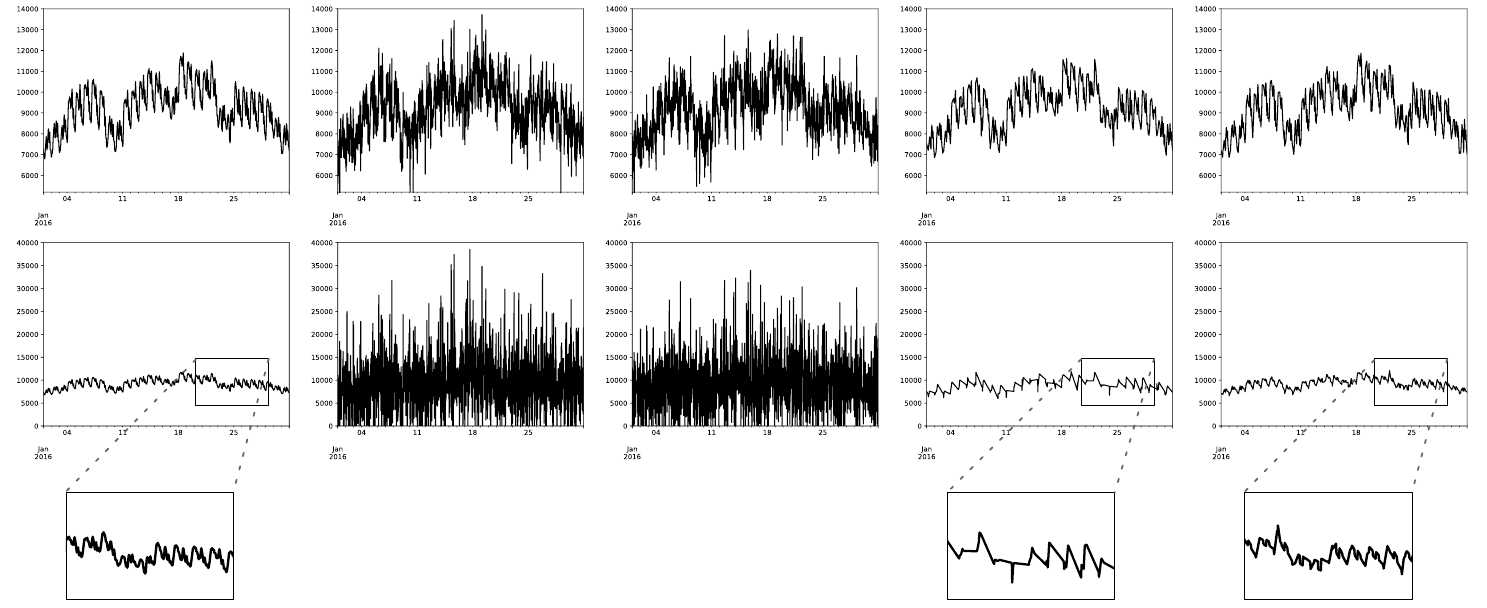}
    \caption{
    Real load consumption data for the Auvergne-Rh\^{o}ne-Alpes region in January, 2016 (Real) and 
    \algname{} activating: 
    exclusively the perturbation step ($\langle P, \circ , \circ \rangle$), 
    perturbation and optimization steps ($\langle P, \circ , O \rangle$),
    perturbation and sampling steps ($\langle P, S , \circ \rangle$), and 
    all steps ($\langle P, S , O \rangle$ ).
    The top row, illustrates the results for privacy budget $\epsilon=0.1$, while the bottom row uses  $\epsilon=0.01$.
    The boxes in the last two quadrants provide a zoom to illustrate in more details the reconstructed time series.
     }
    \label{fig:individual_components}
\end{figure}

We also evaluate the effect of the sampling step as
well as the benefits of the post-processing procedure of \algname{}.
Figure \ref{fig:individual_components} visualizes the real (first
column) and \algname-based private version (other columns) load
consumption data for the Auvergne-Rh\^{o}ne-Alpes region in January,
2016. The top row and bottom rows illustrate results for privacy
budgets of $\epsilon = 0.1$ and $0.01$ respectively.

In order from the second to the last row, Figure
\ref{fig:individual_components} illustrates the results of a version
of \algname{} that enables only the perturbation step ($\langle P,
\circ , \circ \rangle$), the perturbation and the optimization steps
($\langle P, \circ , O \rangle$), the perturbation and the sampling
steps ($\langle P, S , \circ \rangle$), and all steps ($\langle P, S ,
O \rangle$).  The privacy budget is divided equally among all active
components of the algorithm.  When only the perturbation step is
enabled (second column), \algname{} has the same behavior of the
Laplace mechanism.  The addition of the optimization step (third
column) helps increasing the fidelity of the data stream reported (for
instance, the heavy peaks are less pronounced), although the overall
signal is still noisy.  When the perturbation step is combined with
the sampling step under the $L1$ reconstruction procedure (fourth
column), the resulting private stream captures with much higher
fidelity the salient load \emph{peaks} and the streams looks much more
similar to the real one.  However, a careful inspection reveals
several shortcomings: Several load peaks are not captured accurately,
and the loads are often distributed uniformly between sampled points,
making the resulting stream unrealistic.  Finally, the
optimization-based post-processing shines when integrated with the
perturbation and sampling steps (5-th column). The salient features of
the data are effectively captured and the noise redistribution
provides higher resemblance to the real data stream.

\def \wlen{.32\linewidth}
\begin{figure}[!t]
    \centering
    \includegraphics[width=180pt]{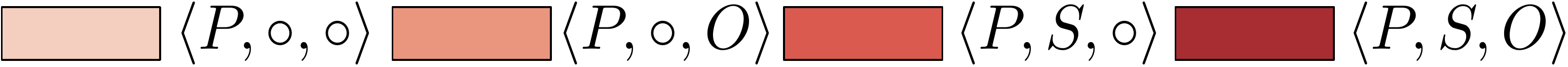}\\
    \includegraphics[width=\wlen]{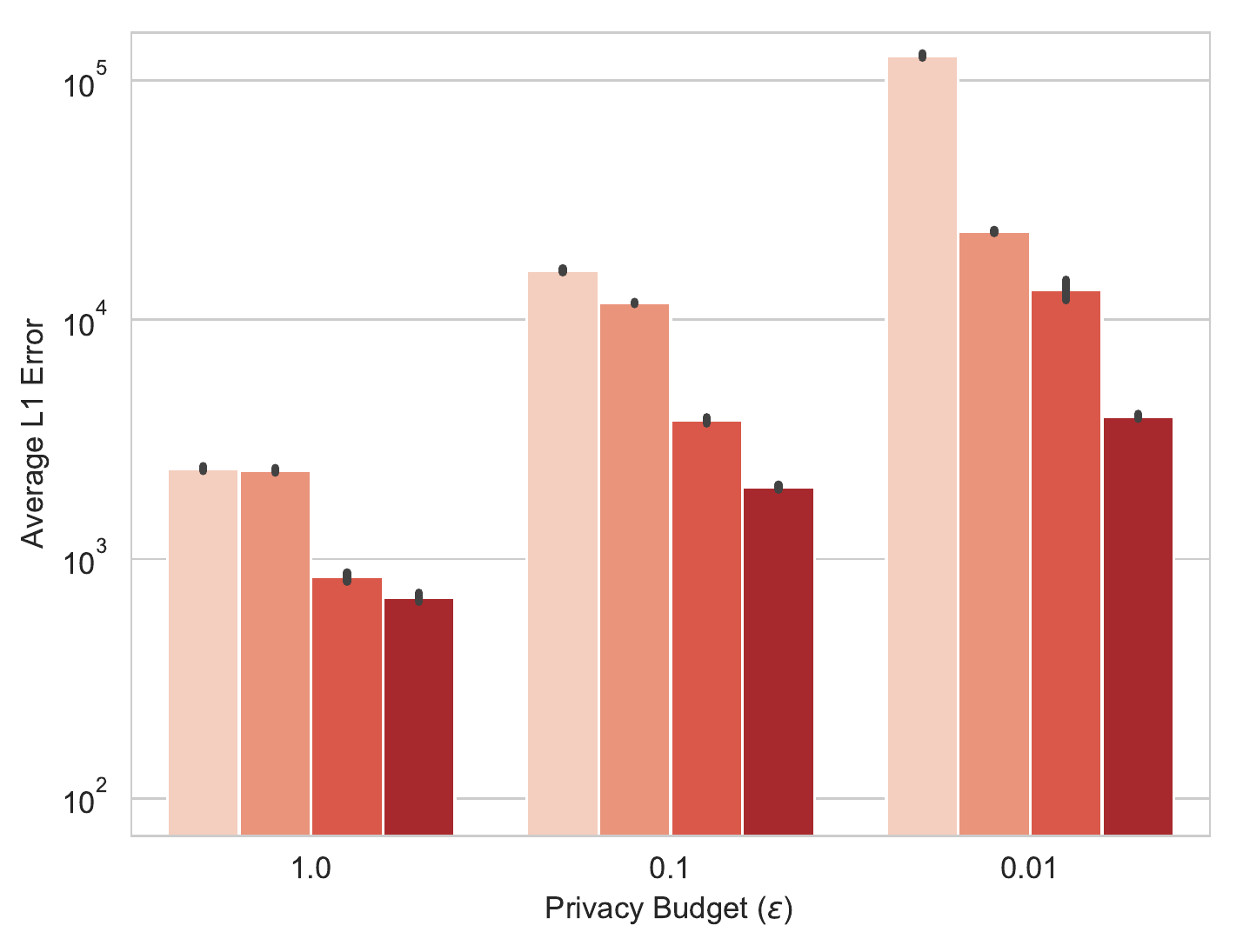}
    \includegraphics[width=\wlen]{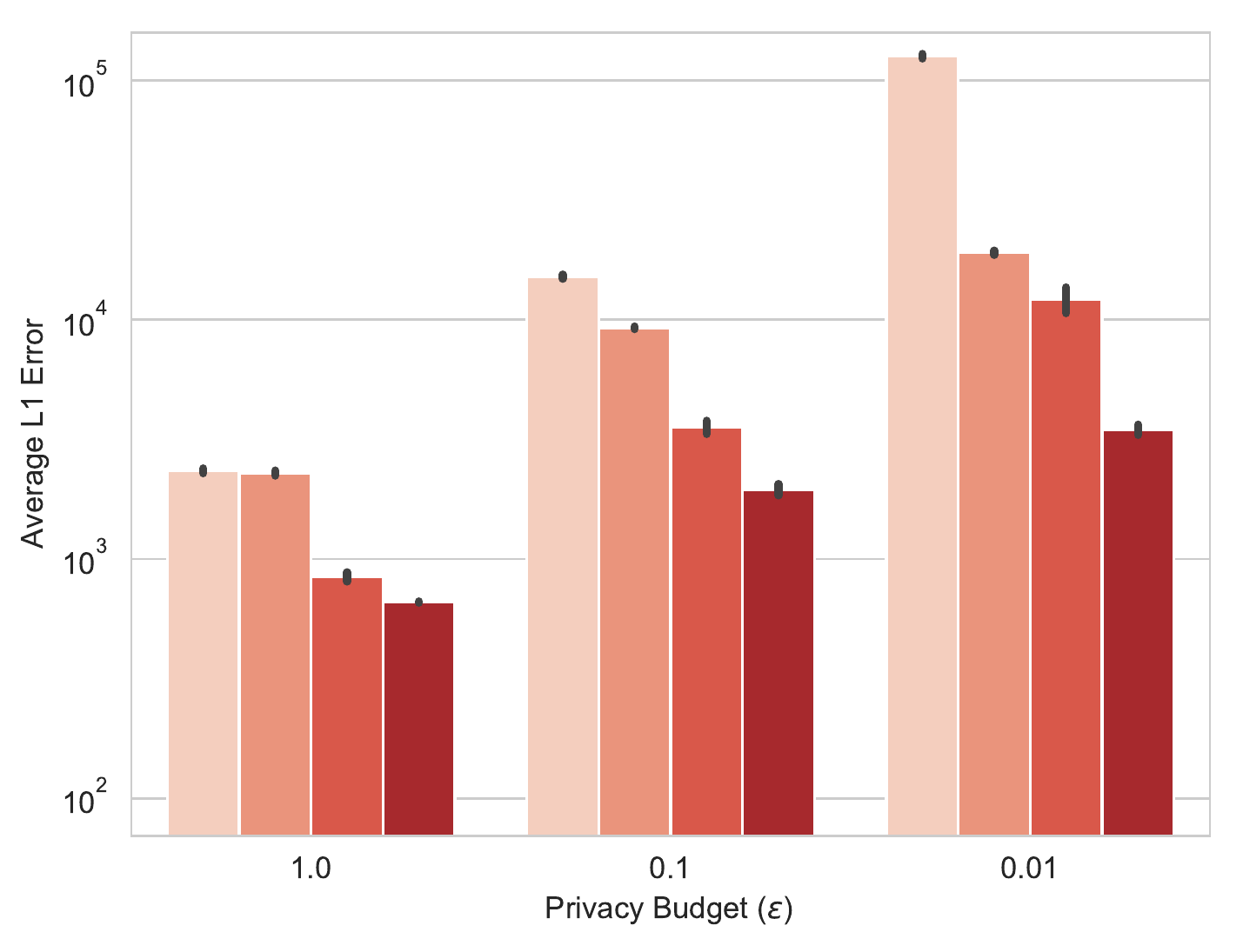}
    \includegraphics[width=\wlen]{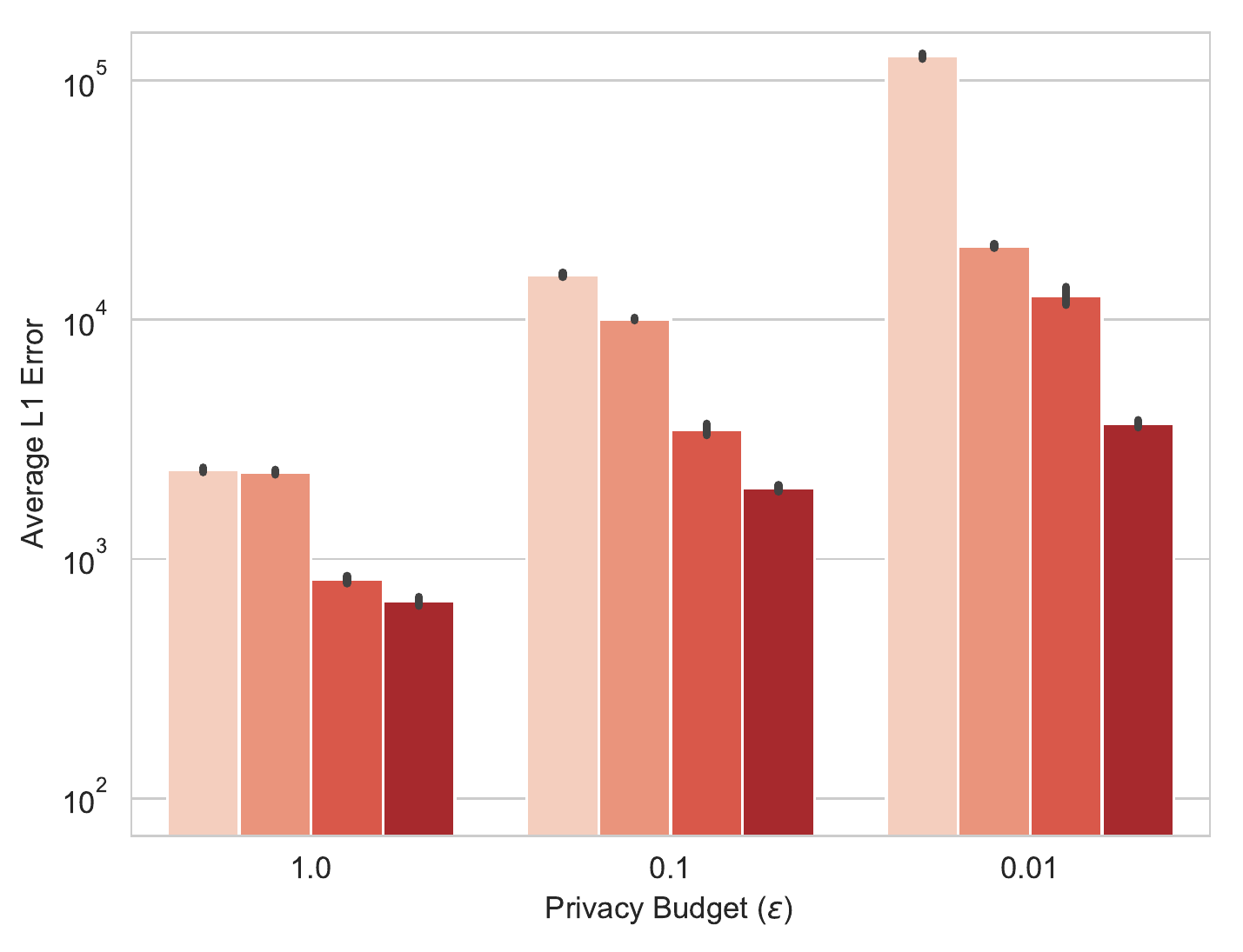} \\
    {\small
    \hspace{60pt}February \hfill ~~~~June~~~~ \hfill October\hspace{50pt}}
    \caption{OptStream components analysis:
    Average $L_1$-error associated with the energy loads of the months of February, June, and October. 
     }
    \label{fig:comp_L1_errors}
\end{figure}

Figure \ref{fig:comp_L1_errors} details the L1-errors associated with
the private energy loads for the months of February, June, and
October.  This quantitative analysis further highlights the importance
of the four components of the proposed \algname{} algorithm. It
illustrates a gradual error reduction of the framework when it adopts,
in order, the perturbation step only, the perturbations and the
optimization steps, the perturbation and the sampling with linear 
interpolation steps, and finally, all four steps.

\subsection{Stream Data-Release in the $(w, \alpha)$-Indistinguishability Model}

Having analyzed the algorithms in the $w$-privacy setting, this
section illustrates the results of the private data streams released
under the combined $w$-event and $\alpha$-indistinguishability models
(see Section \ref{sec:indistig}).  The motivating application requires
releasing the private data streams without disclosing chosen amounts
of load changes within a time window.  We thus analyze the results of
privately releasing a stream of data under the $(w,
\alpha)$-indistinguishability model and use an adapted version of
\algname{}, the Laplace mechanism, and the DFT algorithm. These
versions recalibrate the noise required by the mechanisms to satisfy
the $(w, \alpha)$-indistinguishability definition as detailed in
Section \ref{sec:indistig}.

\def \wlen{.24\linewidth}
\begin{figure}[tbh]
    {\small\centering
    \hspace{35pt}
    ~~Real~ \hspace{75pt}
    Laplace \hspace{70pt}
    ~~DFT~~ \hspace{60pt}
    \algname
    }\\
    \centering
    \includegraphics[width=\wlen]{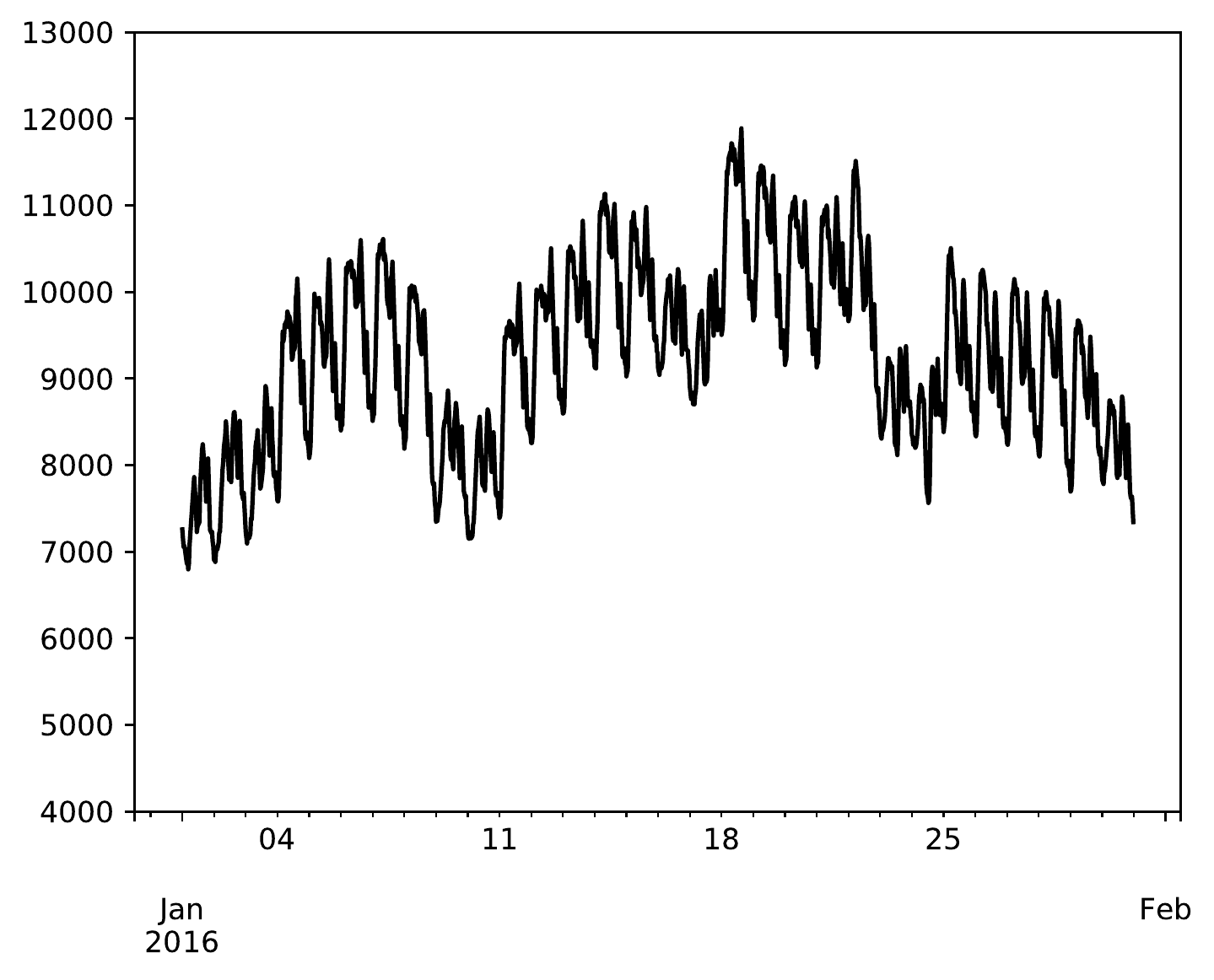}
    \includegraphics[width=\wlen]{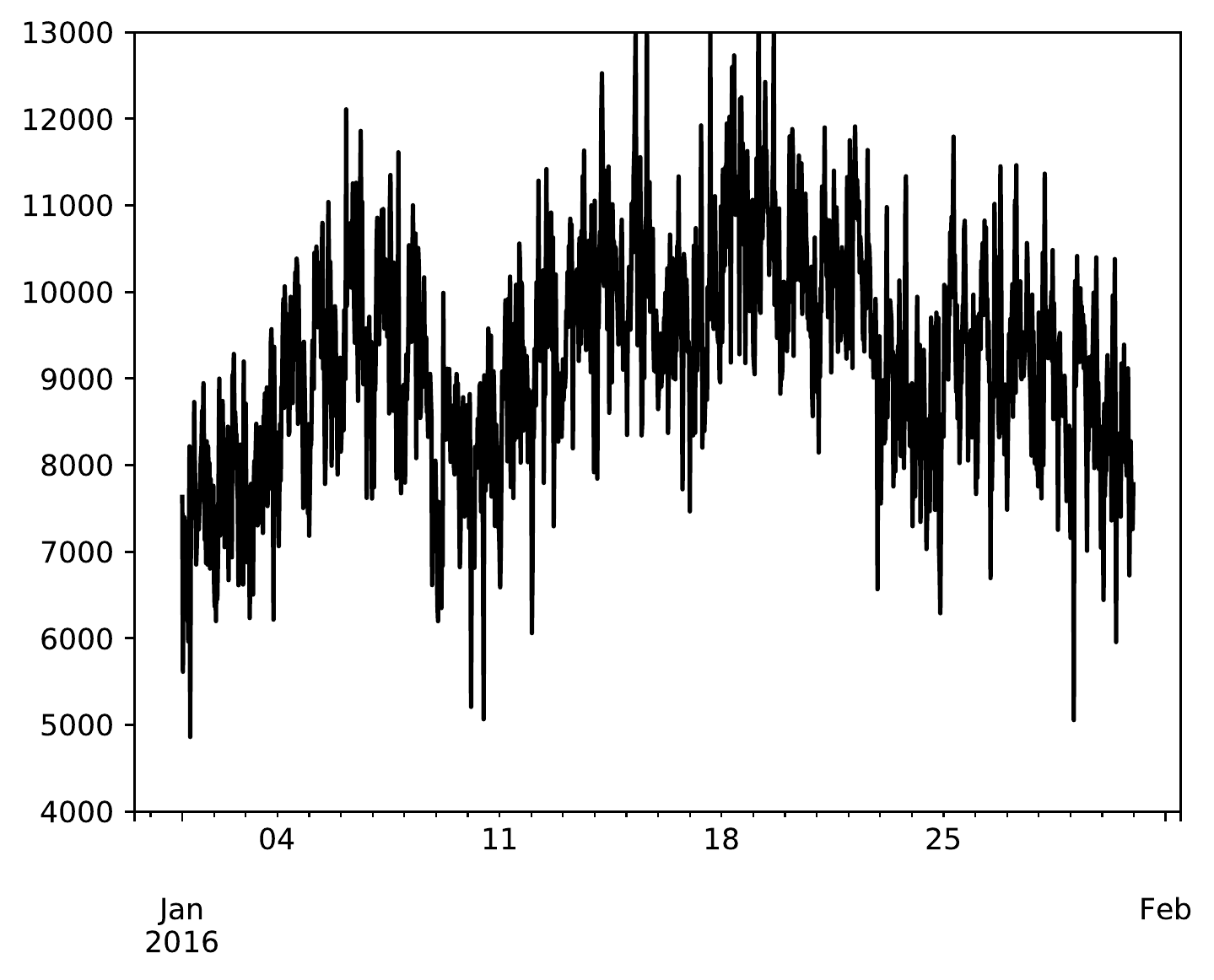}
    \includegraphics[width=\wlen]{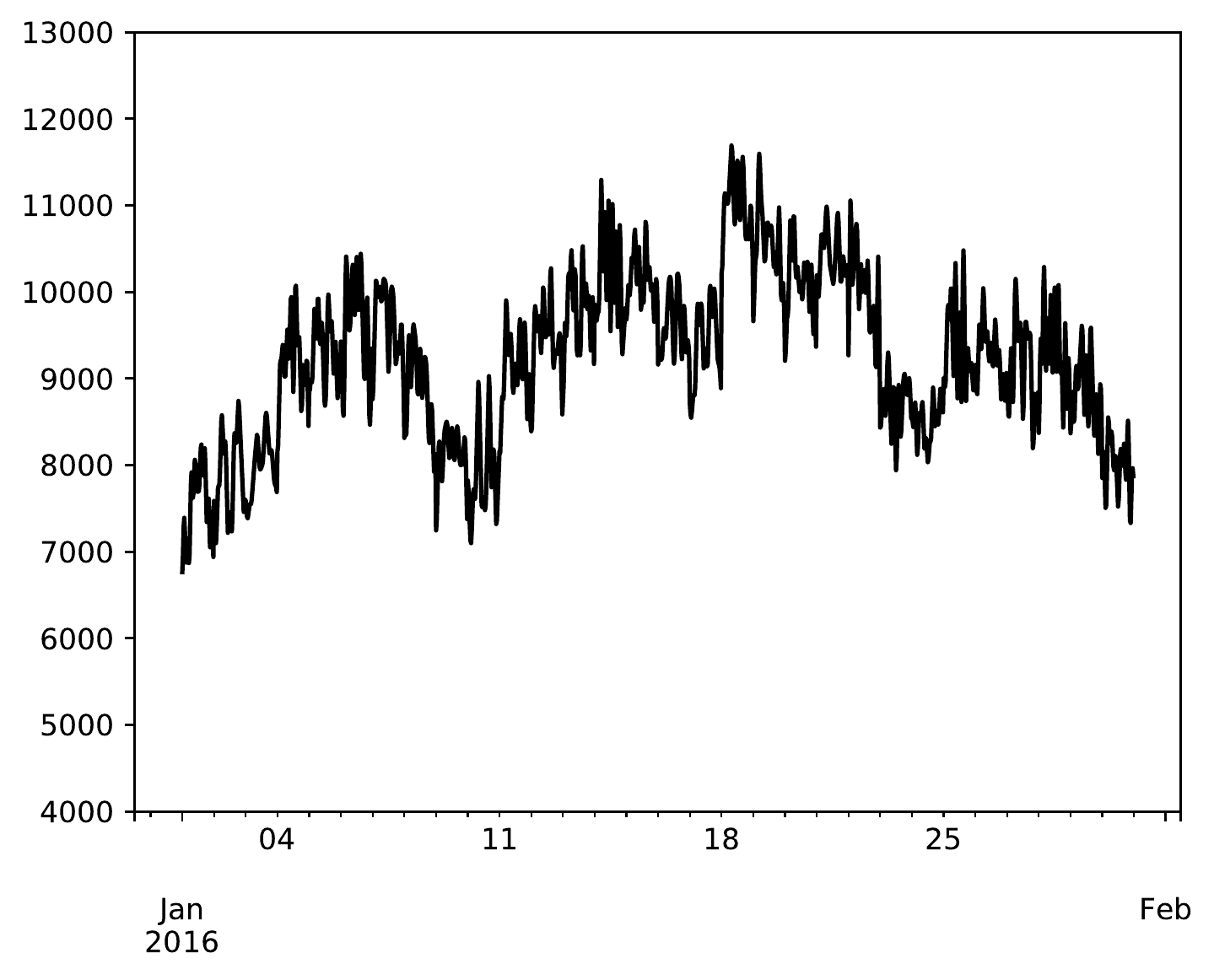}
    \includegraphics[width=\wlen]{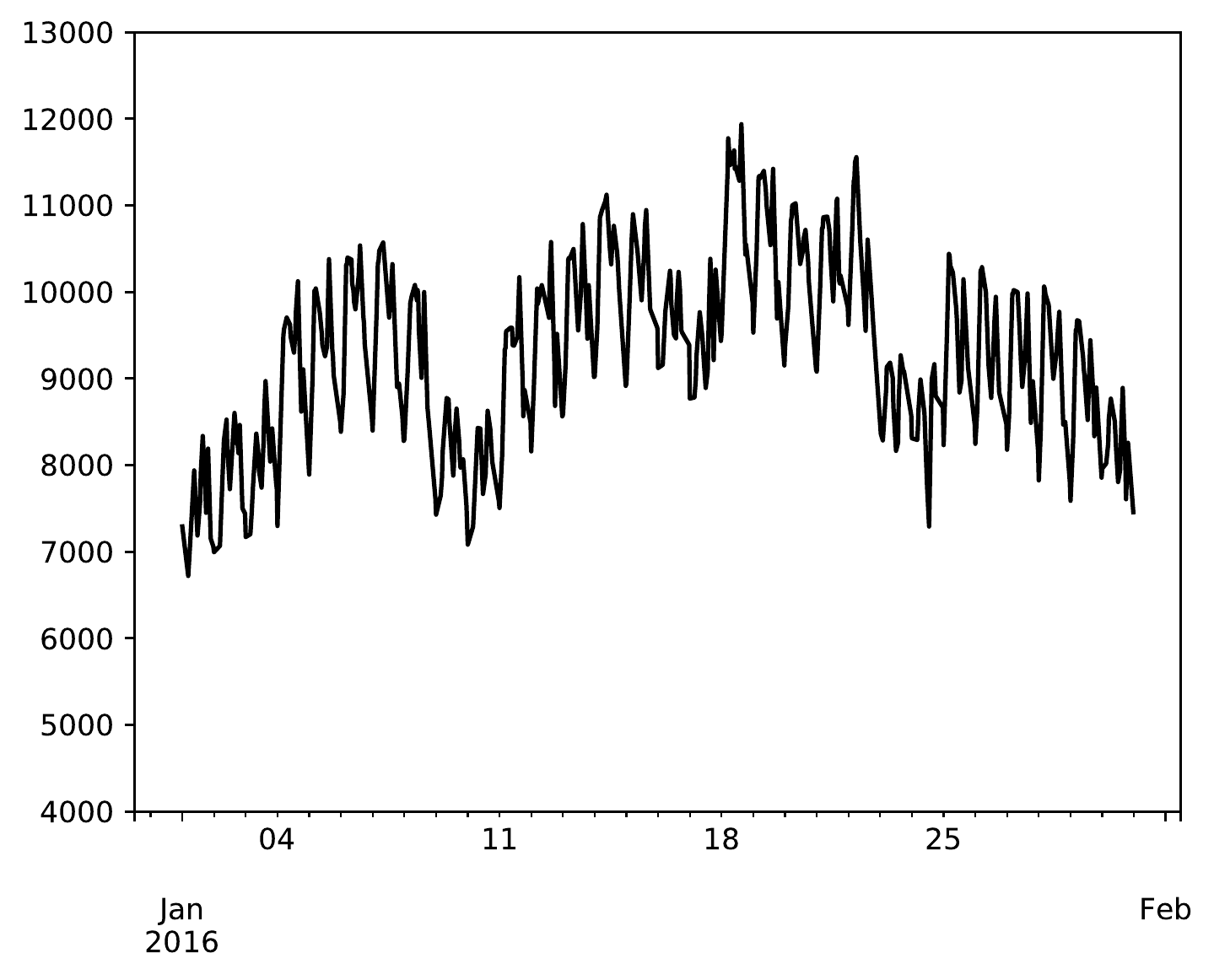}\\
    \includegraphics[width=\wlen]{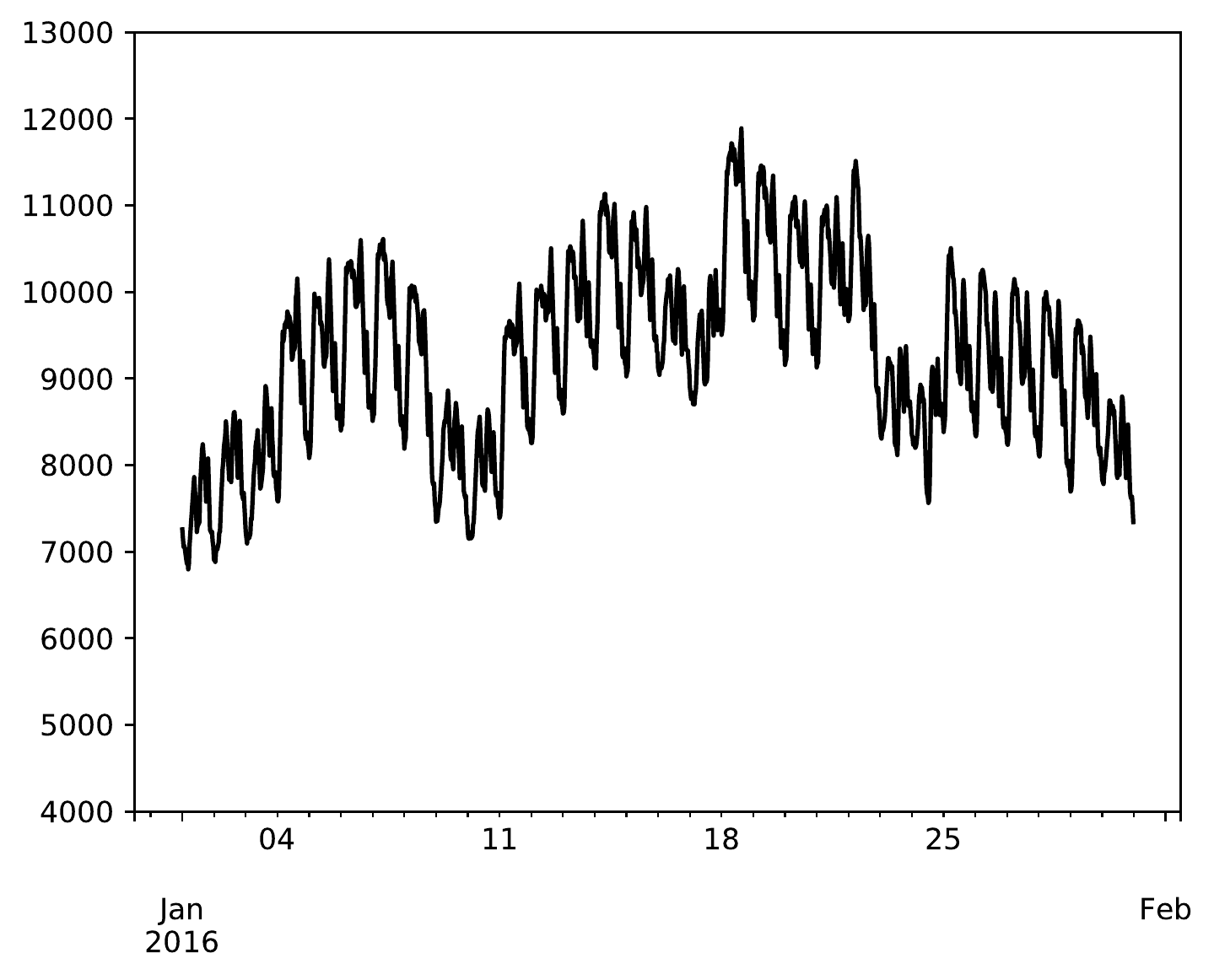}
    \includegraphics[width=\wlen]{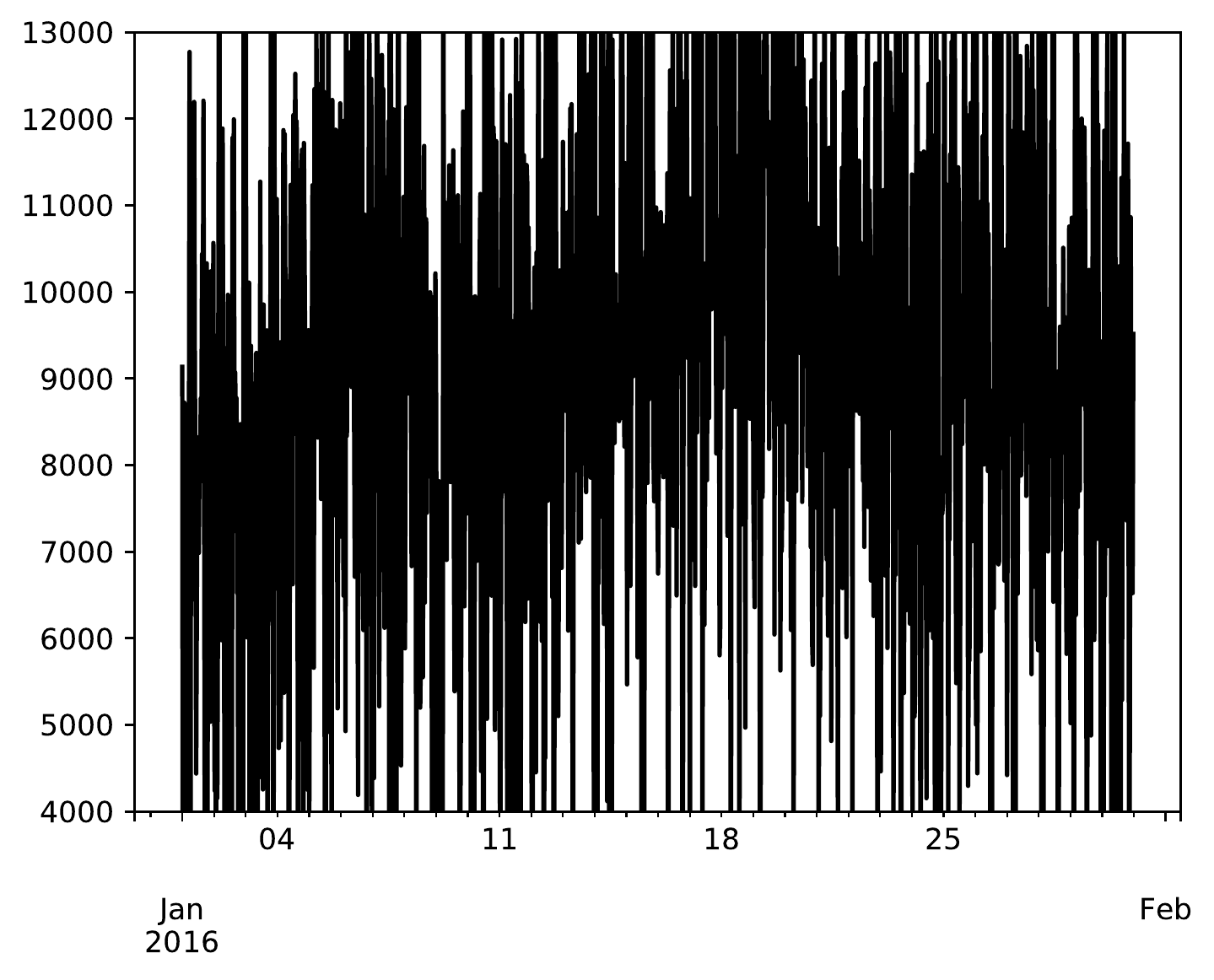}
    \includegraphics[width=\wlen]{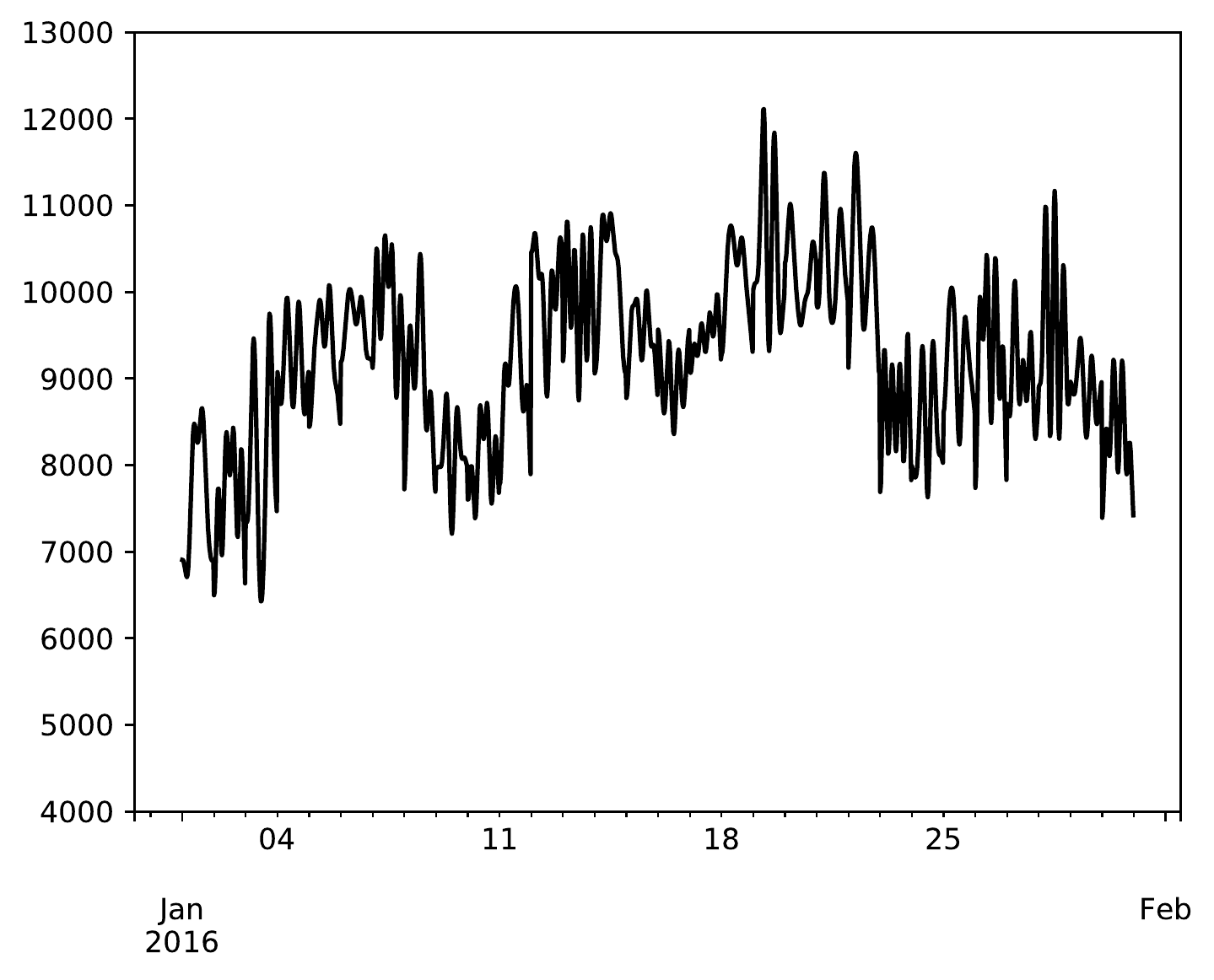}
    \includegraphics[width=\wlen]{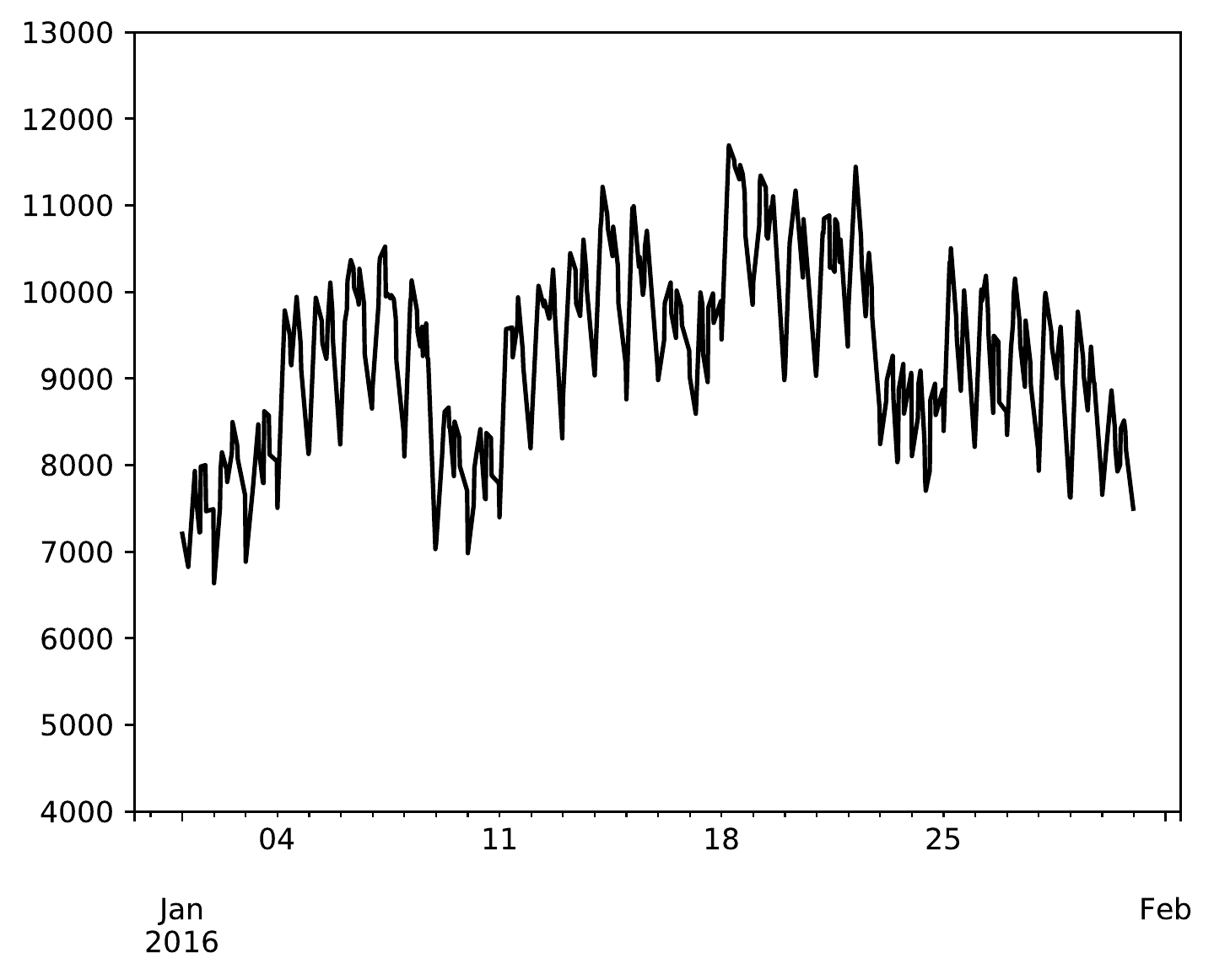}\\
    \includegraphics[width=\wlen]{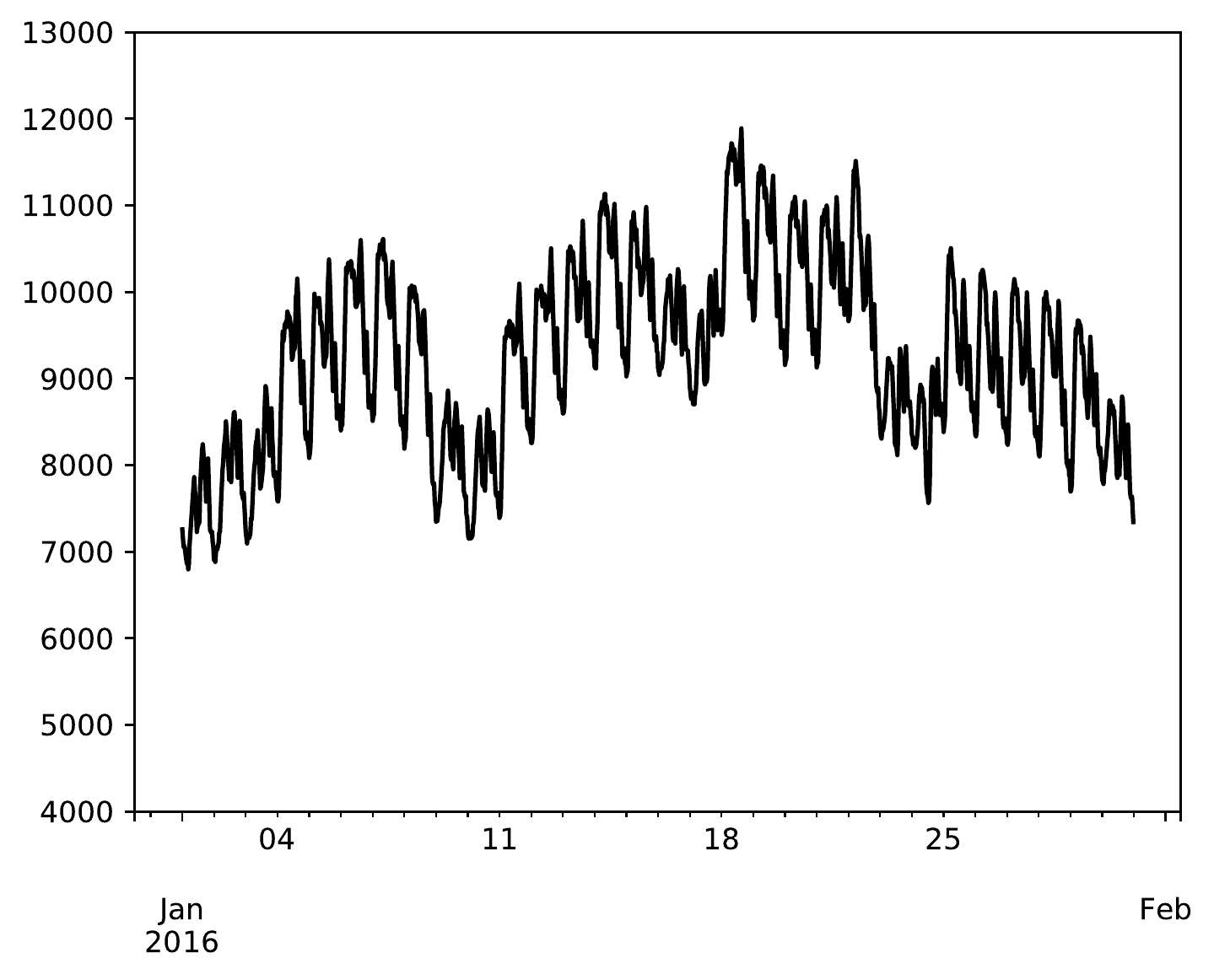}
    \includegraphics[width=\wlen]{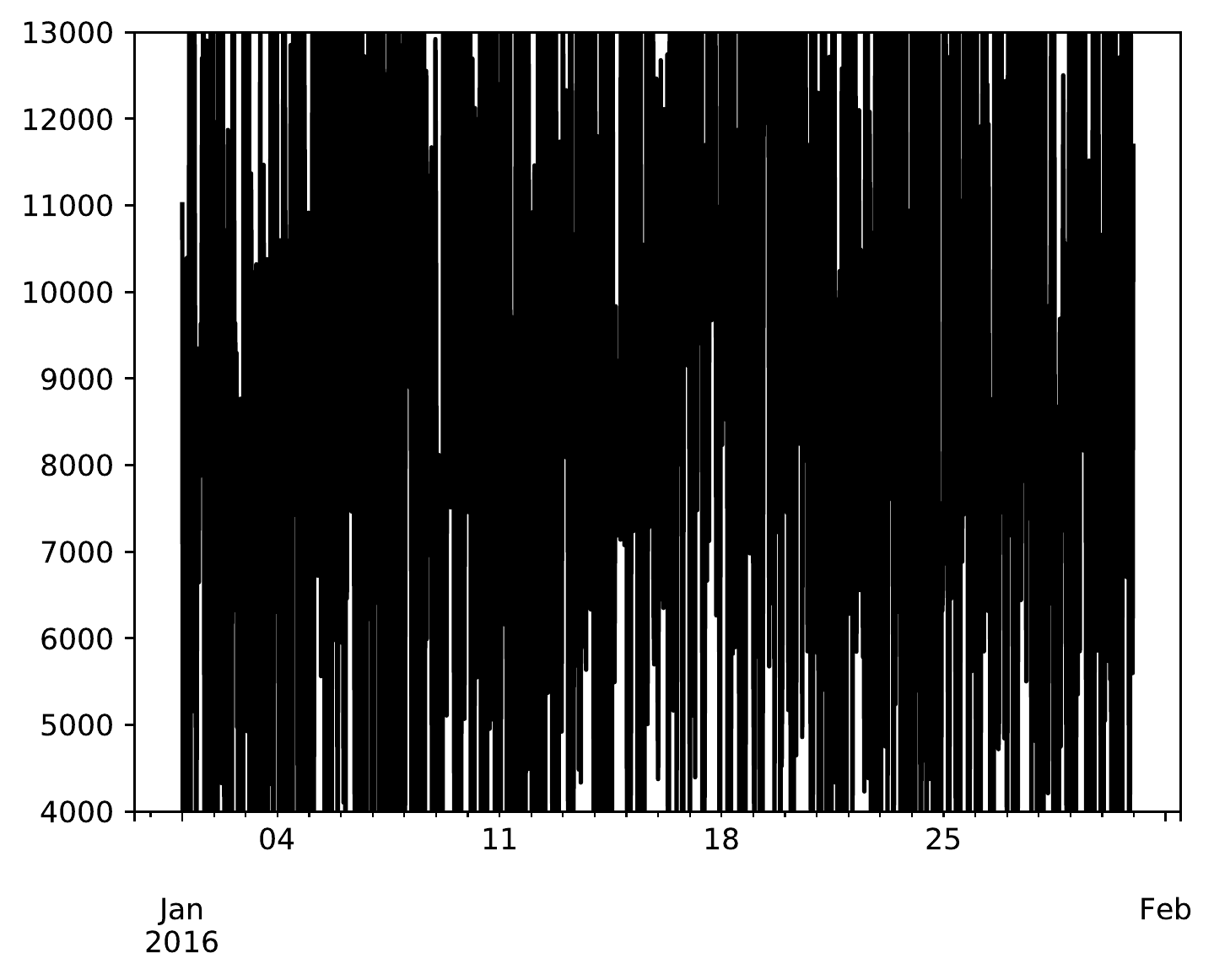}
    \includegraphics[width=\wlen]{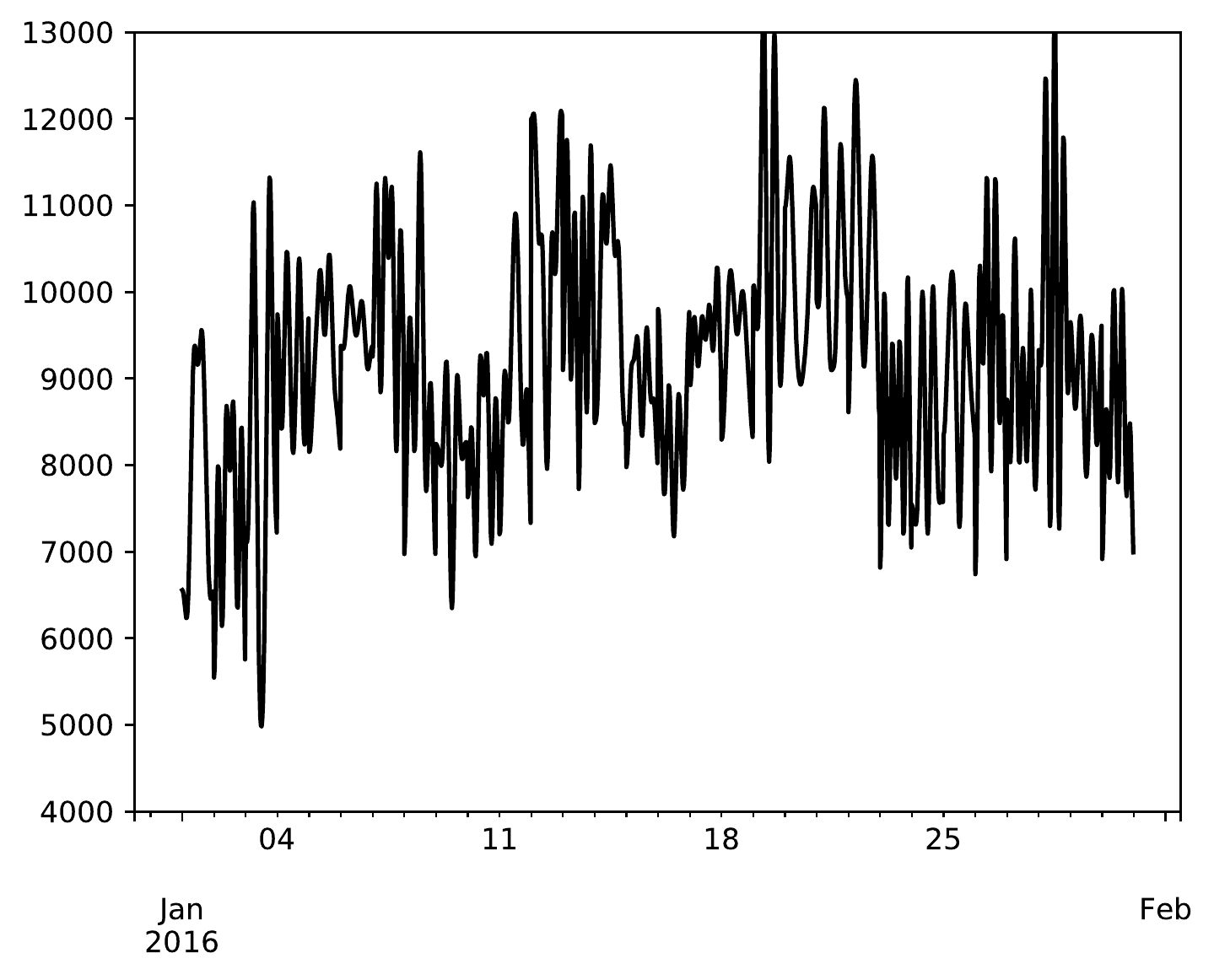}
    \includegraphics[width=\wlen]{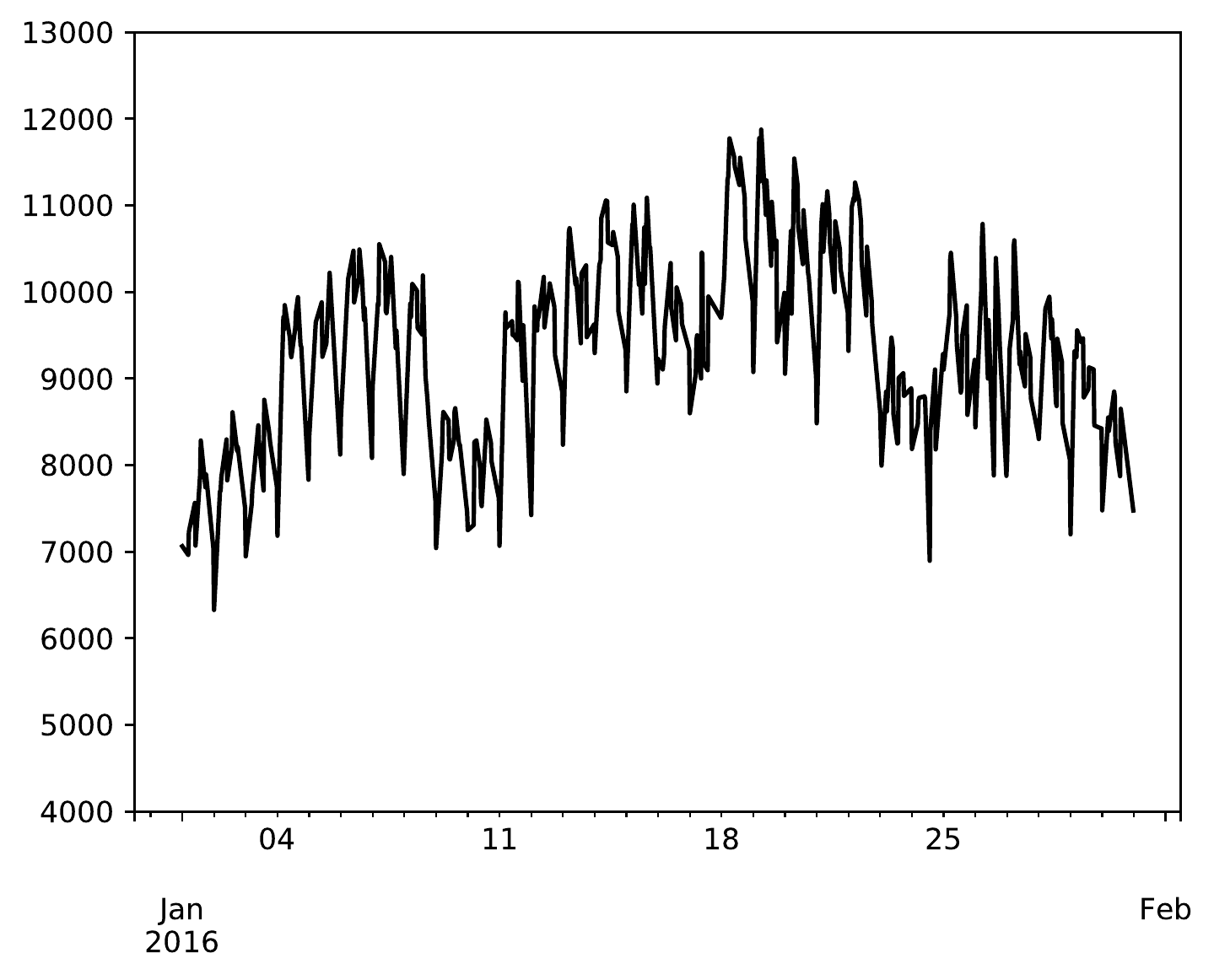}
    \caption{Real load consumption data for the Auvergne-Rh\^{o}ne-Alpes region in January 2016, and its private versions obtained using Laplace, DFT, and \algname{} with privacy budget $\epsilon=1$ and $\alpha=10$ (top), $50$ (middle), and $100$ (bottom).
     }
    \label{fig:real_and_noisy_data} 
\end{figure}

Figure \ref{fig:real_and_noisy_data} illustrates the real and private
versions of the data stream for $w$-periods of size 48 (i.e., one day)
given privacy budget of $\epsilon = 1$ and indistinguishability
parameter $\alpha=10$ (top row), $\alpha=50$ (middle row) and $\alpha
= 100$ (bottom row).  Recall that the choice of the
indistinguishability parameter $\alpha$ allows the data curator to
ensure the protection of the observed increase/decrease power
consumption up to $\alpha$ MegaWatts (MW), while the $w$-period
specifies the length of the obfuscation period.  The figure compares our
proposed \algname\ algorithm against the Laplace mechanism and DFT .

For DFT and \algname{}, the experiments set the number of Fourier
coefficients and sampling steps to $10$ (when $\alpha \leq 50$) and
$5$ (when $\alpha=100$). As in the previous experiments, the privacy
budget allocated to perform each measurement is divided
equally. Finally, \algname{} uses the same feature query set $\sF$
used in the previous experiments. Figure \ref{fig:real_and_noisy_data}
clearly illustrates that \algname\ produces private streams that are
more accurate than its competitors when visualized. This is especially
evident for high indistinguishability parameters, when the Laplace
mechanism can barely preserve any signal from the original data.

\def \wlen{.32\linewidth}
\def \wlenLabel{.20\linewidth}
\begin{figure}[!t]
    \centering
    \includegraphics[width=180pt]{legend_1.pdf}\\
    \includegraphics[width=\wlen]{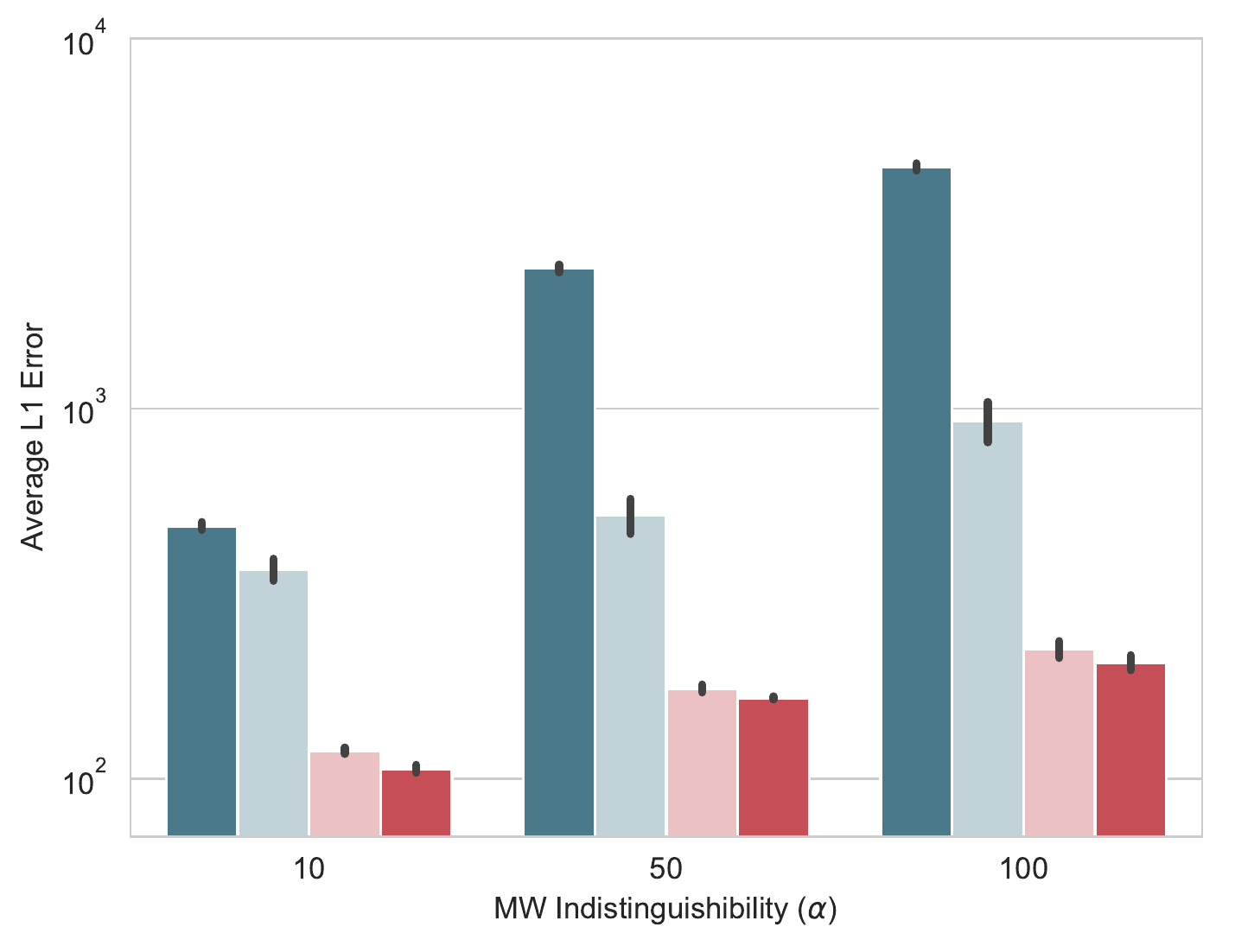}
    \includegraphics[width=\wlen]{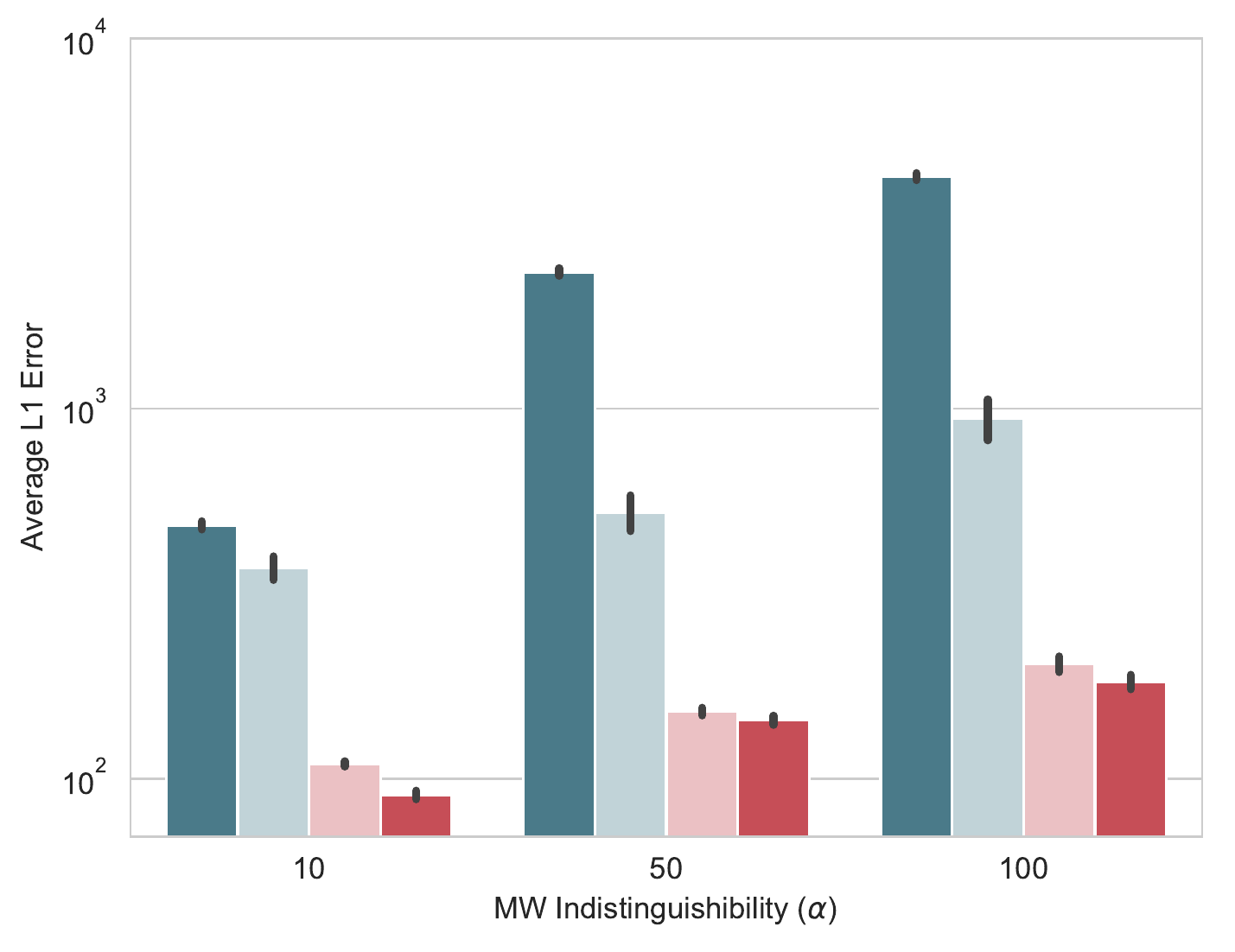}
    \includegraphics[width=\wlen]{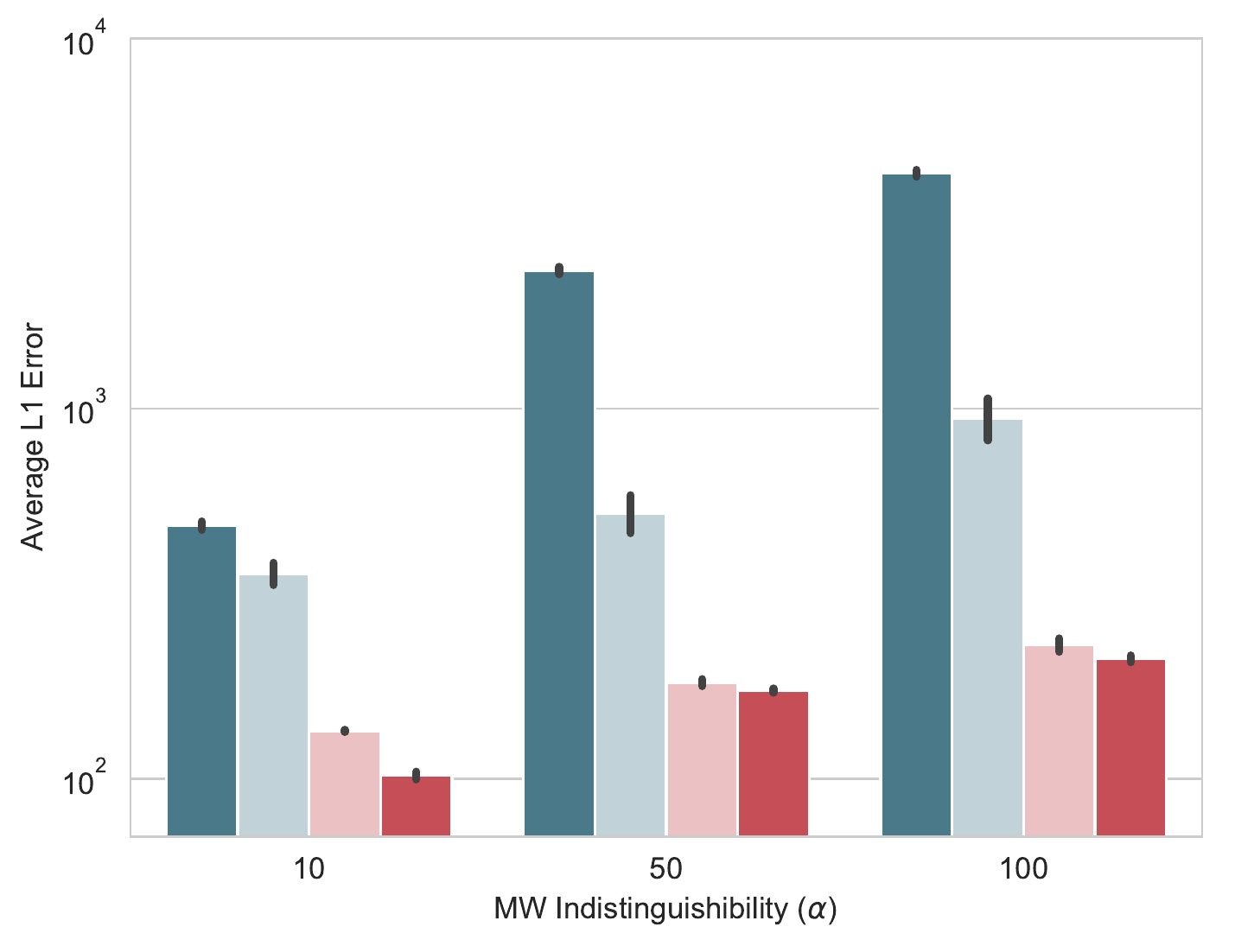}\\
    {\small
    \hspace{60pt}February \hfill ~~~~June~~~~ \hfill October\hspace{50pt}}

    \caption{$L_1$-error analysis: Energy load stream data for the months of February, June, and October. The $y$-axis reports $\log_{10}$ errors averaged across the $L_1$ errors for each stream data $R \in \sR$.}
    \label{fig:L1_error_2}
\end{figure}

We now report the $L_1$-errors averaged for each $w$-period produced
by the algorithms.  Figure \ref{fig:L1_error_2} reports the average
results for each streaming region for the months of February, June,
and October, at varying of the indistinguishability parameter $\alpha
= \{10, 50, 100\}$ given a privacy budget $\epsilon=1.0$.  Each
histogram reports the log$_{10}$ value of the average error of 30
random trials.  While all algorithms induce an $L_1$-error which
increases as the indistinguishability level $\alpha$ increases, the
figure highlights that \algname{} consistently outperforms the other
algorithms.


\section{Discussion}
\label{sec:discussion}

One of the core advantages of the w-privacy model is that it provides
flexibility regarding the size of the time frame to protect.
\algname{} exploits this flexibility to select a small set of
representative points to sample and performs additional optimization
over such values to redistribute noise and enhance accuracy.
Practically, this means that the data stream needs to be batched in
periods of $w$ times steps, each of which is privately released at
once.

Recall that the privacy model adopted in this paper combines the
$w$-privacy model \cite{kellaris:14} with the definition of
$\alpha$-indistinguishability \cite{chatzikokolakis:13}. The resulting
$(w, \alpha)$-indistinguishability model allow us to protect events
related to quantities which are not directly related to user
identities but can be used a proxy to disclose a \emph{secret} (e.g.,
a sensitive information) about a user. For instance, in our
application domain, it is important to both protect the given stream
from inferences about load increases or decreases--which are
typically related to sensitive user activities--while also
protecting the stream in an entire day.

While setting $w=1$ would allow to use \algname{} to release the
stream data online -- without requiring batching --, this setting
offers no room for \algname{} to improve accuracy, since it would be
similar to the Laplace mechanism. To the best of our knowledge, the
only available algorithm that can release a private data stream
online is \emph{PeGaSus} \cite{chen2017pegasus}. It uses a
combination of the Laplace mechanism, a partitioning strategy to group
values in contiguous time steps whose deviation is small, and a
smoothing function. PeGaSus, was not introduced in the w-event model
as it focuses on protecting data in a single time-step.

\def \wlen{.23\linewidth}
\def \wlenLabel{.12\linewidth}
\begin{figure}[!t]
    {\small\centering
    \hspace{60pt}
    ~~Real~ \hspace{65pt}
    PeGaSus \hspace{50pt}
    \algname
    }\\
    \centering
    \includegraphics[width=\wlen]{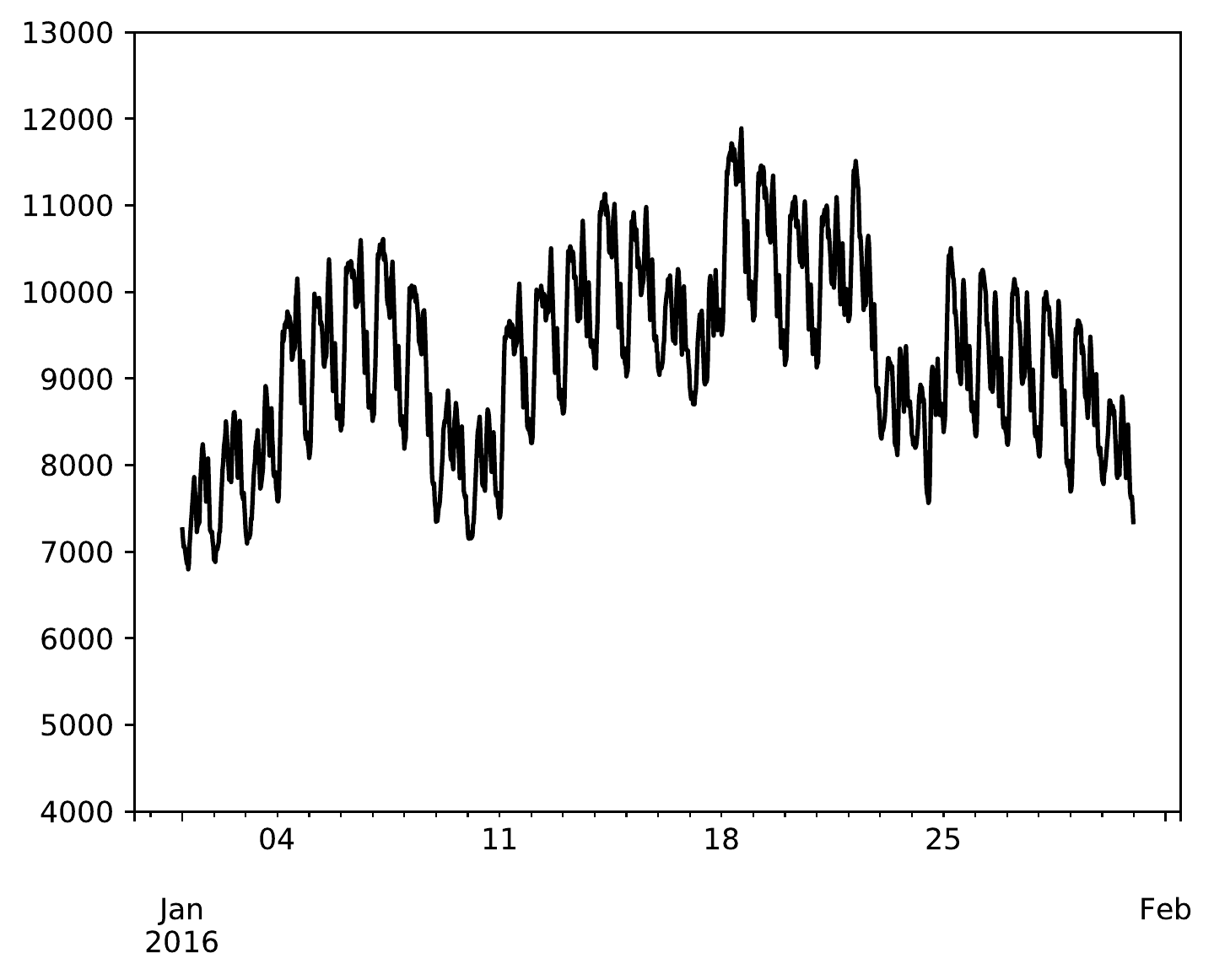}
    \includegraphics[width=\wlen]{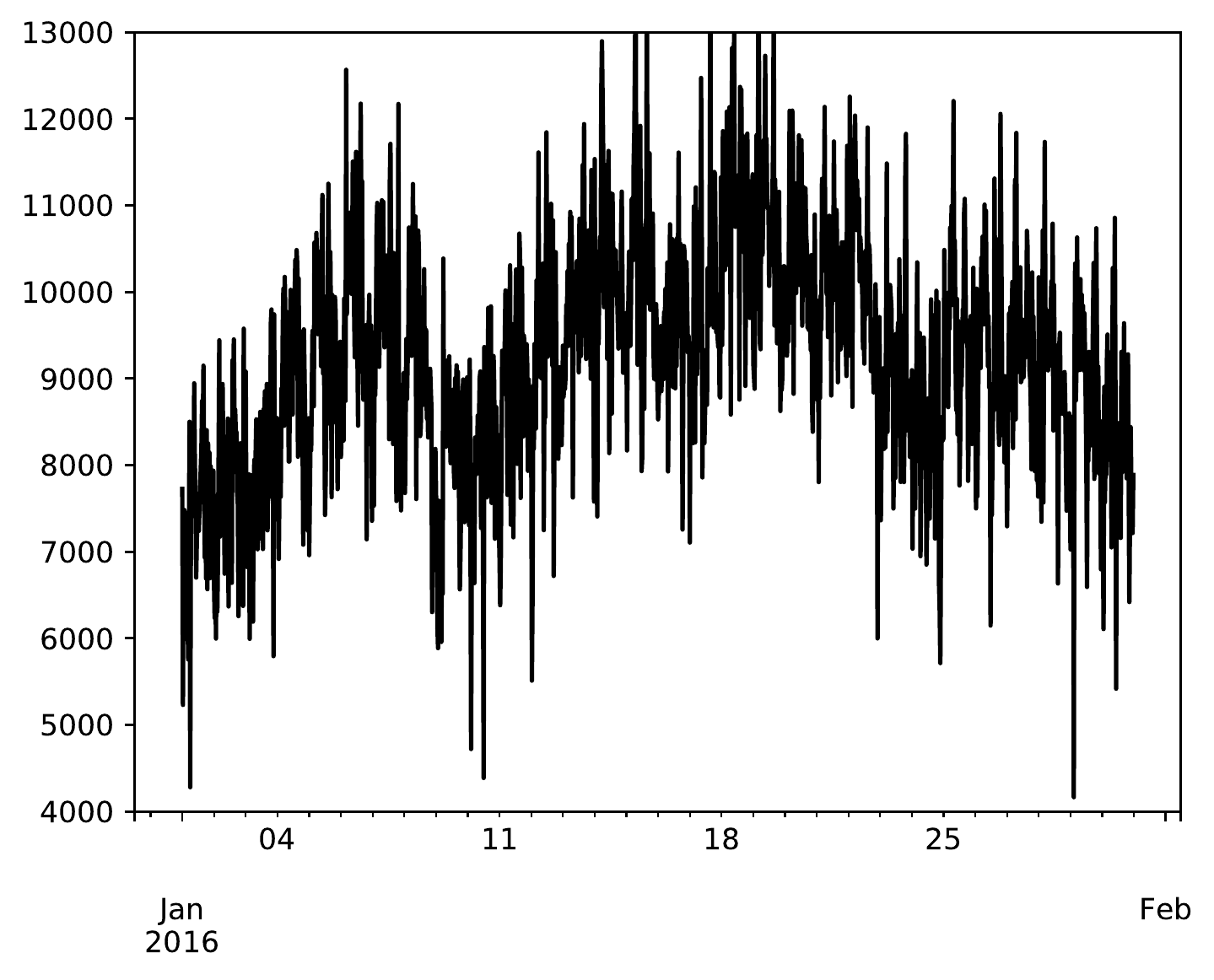}
    \includegraphics[width=\wlen]{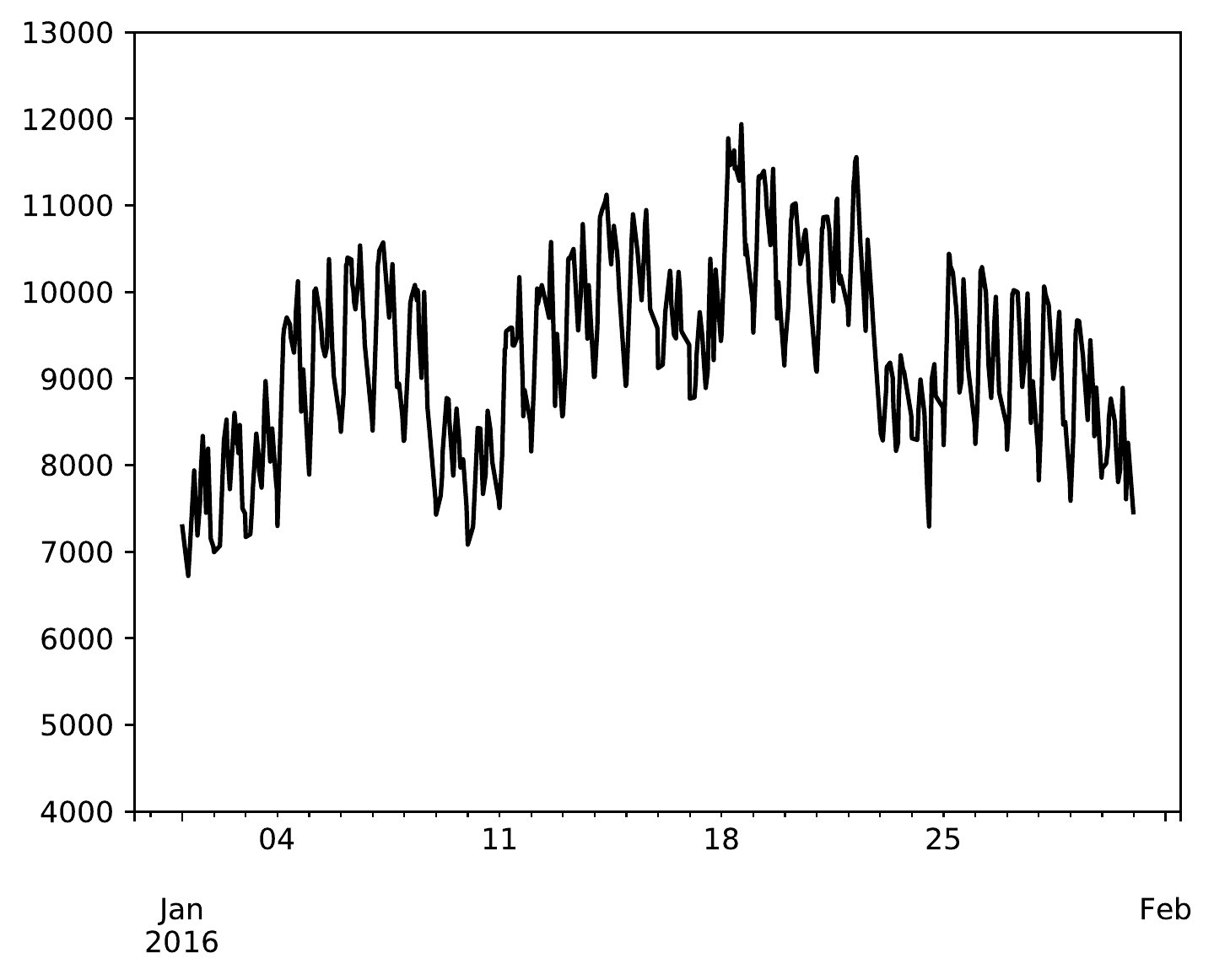}\\
    \includegraphics[width=\wlen]{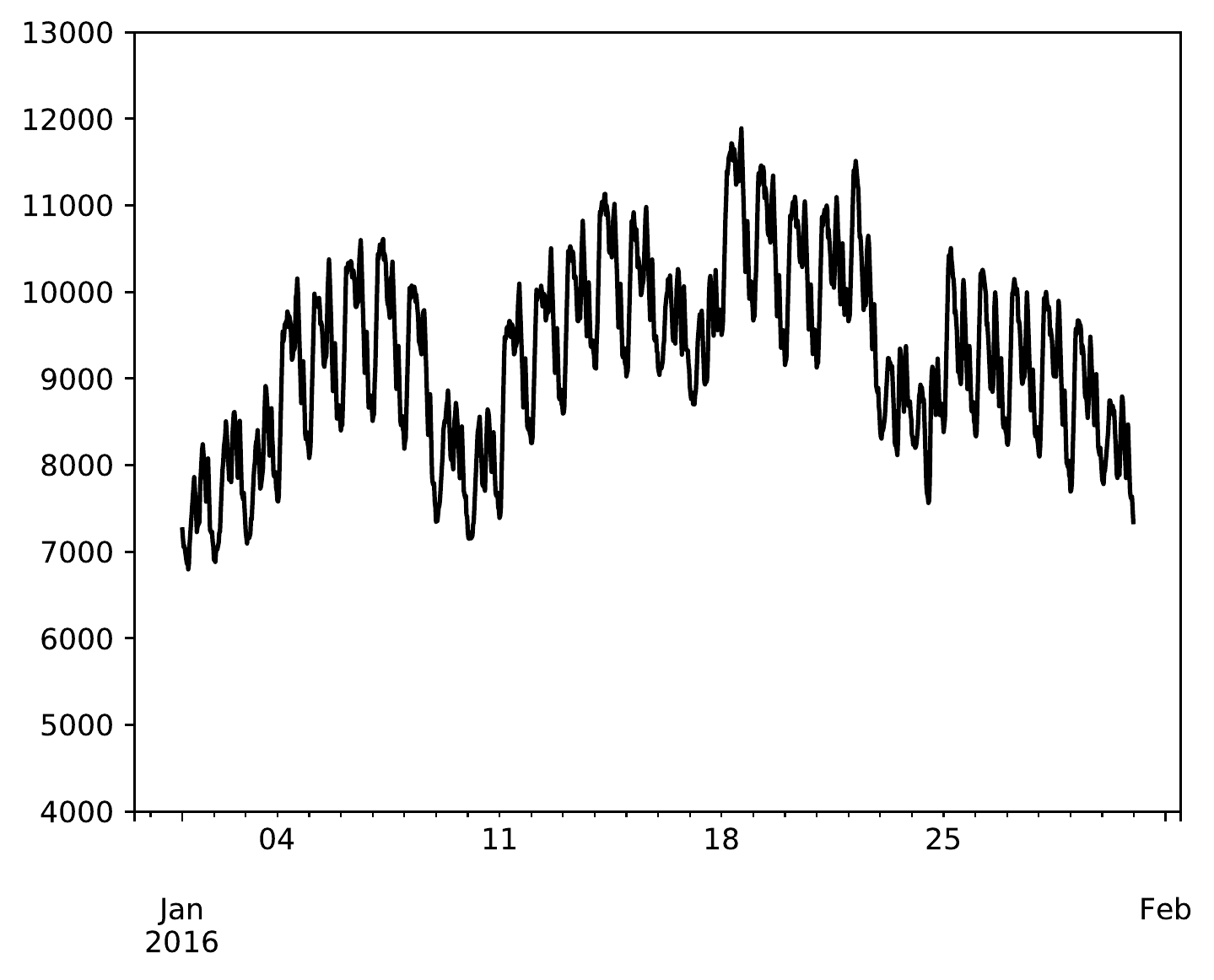}
    \includegraphics[width=\wlen]{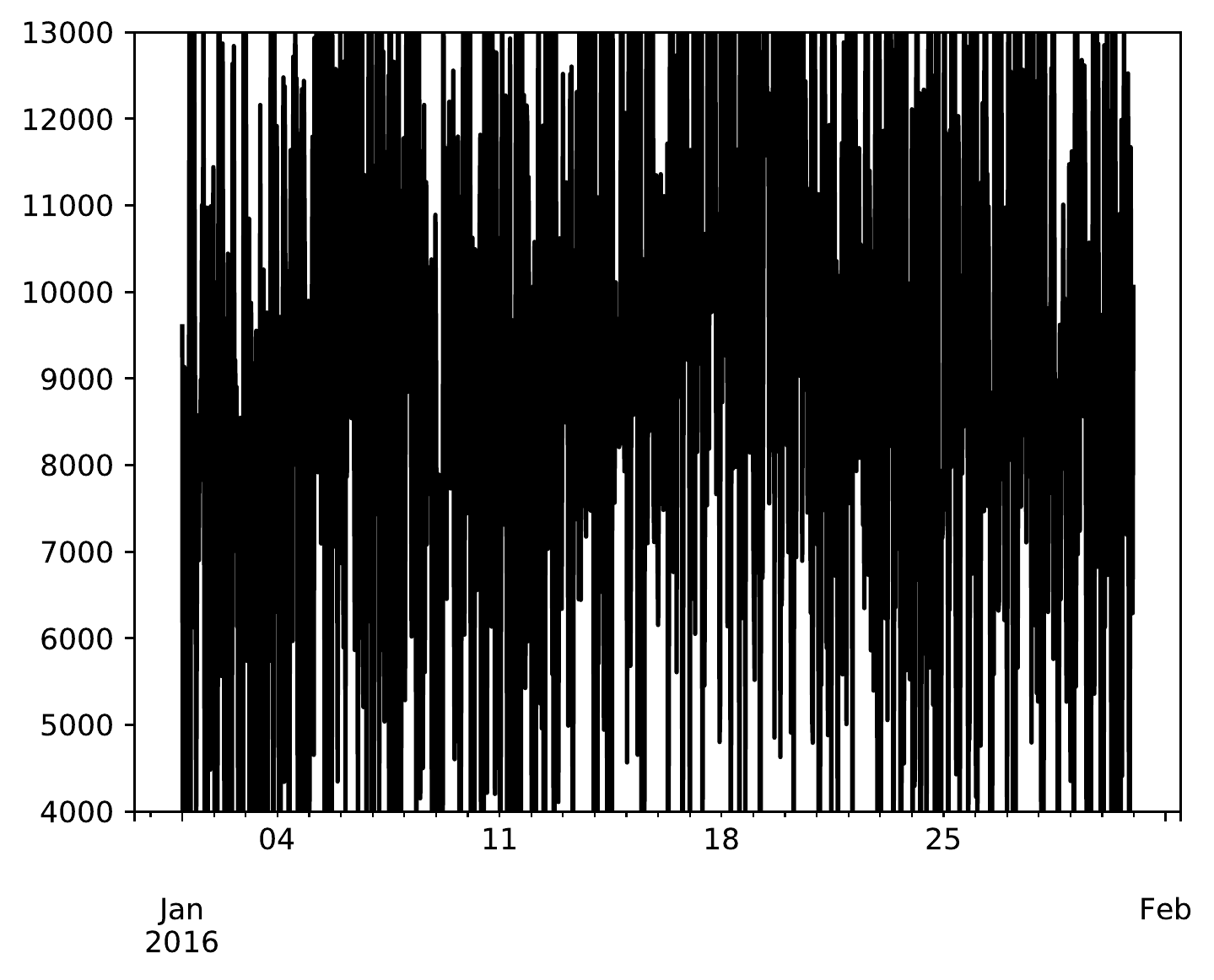}
    \includegraphics[width=\wlen]{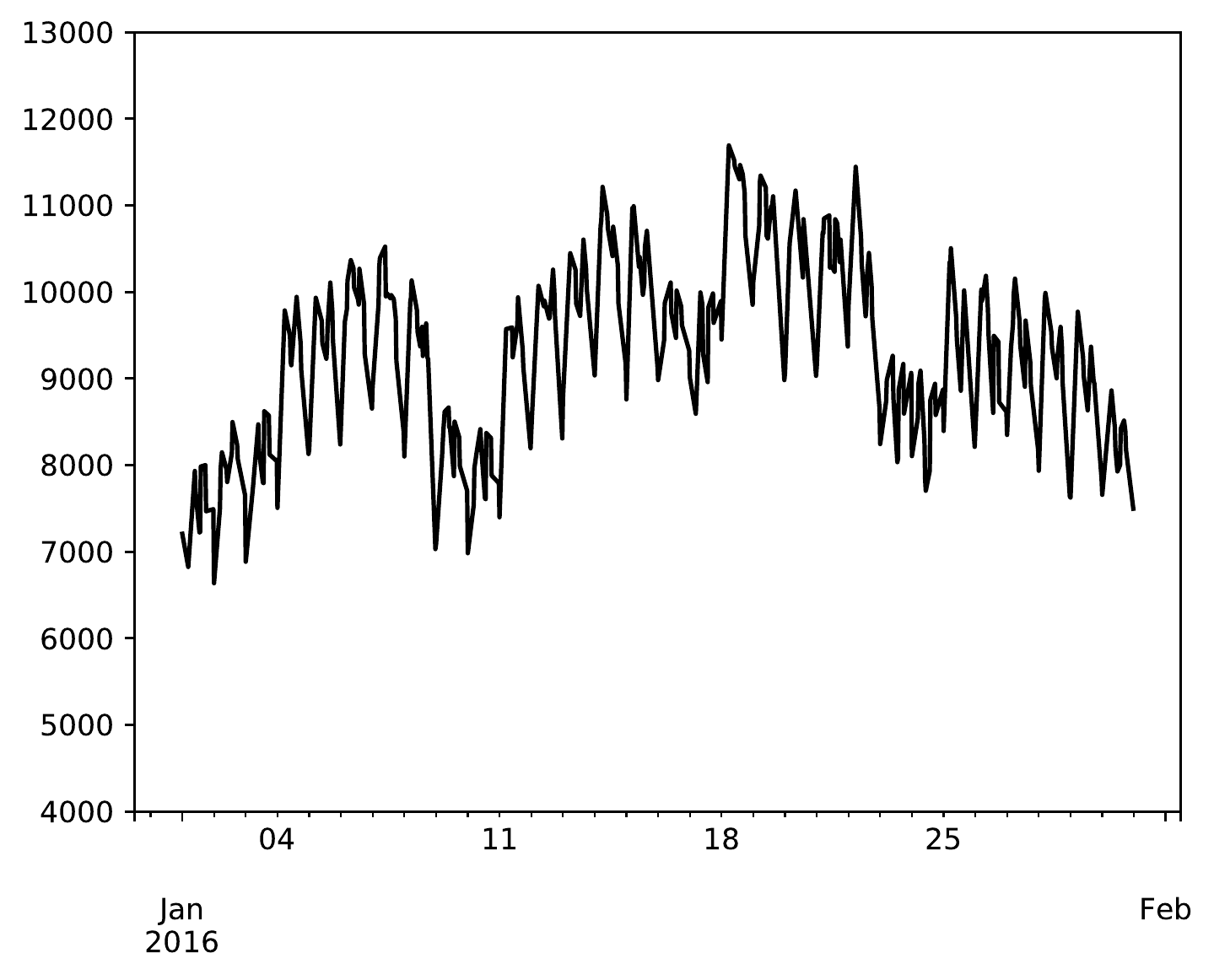}\\
    \includegraphics[width=\wlen]{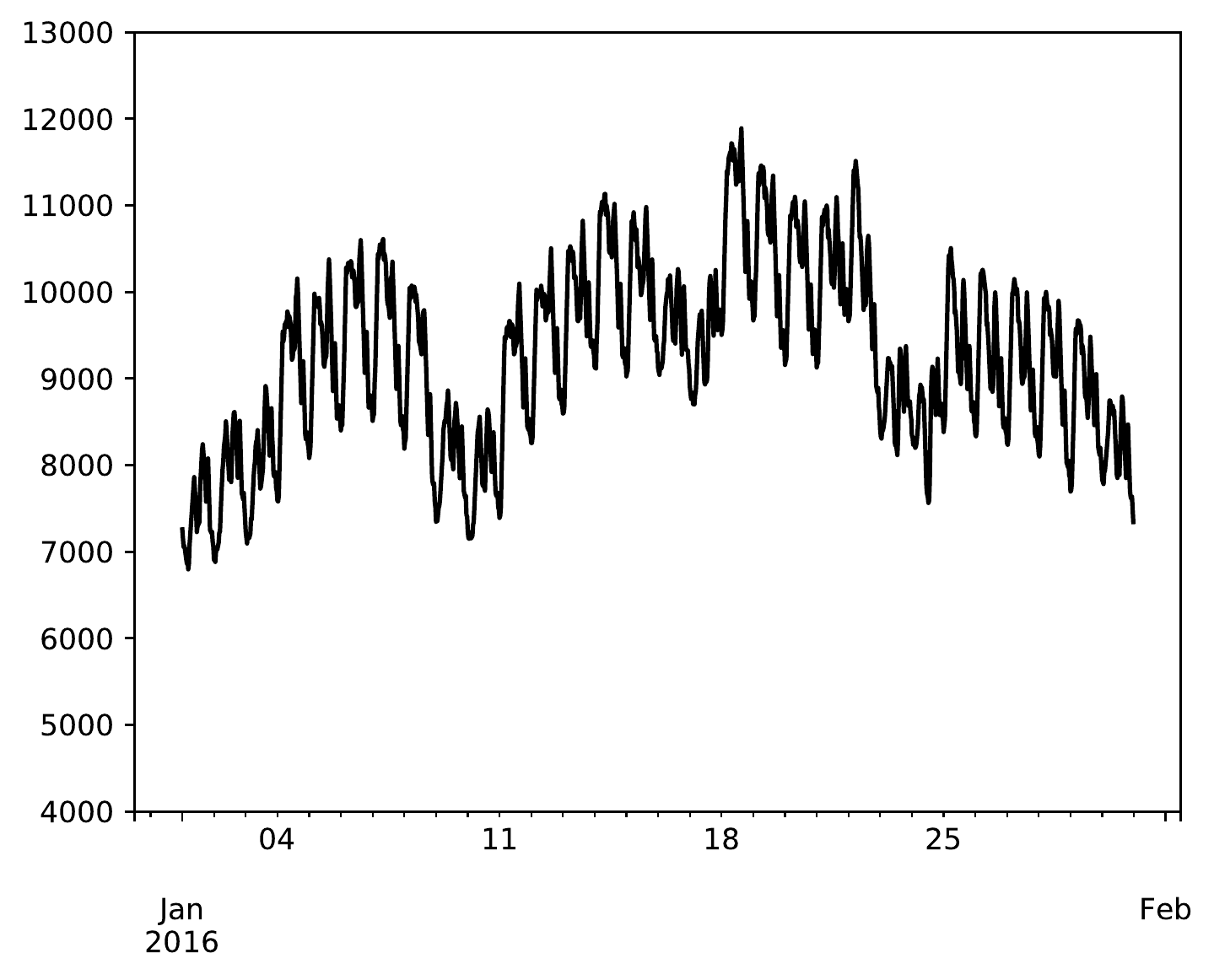}
    \includegraphics[width=\wlen]{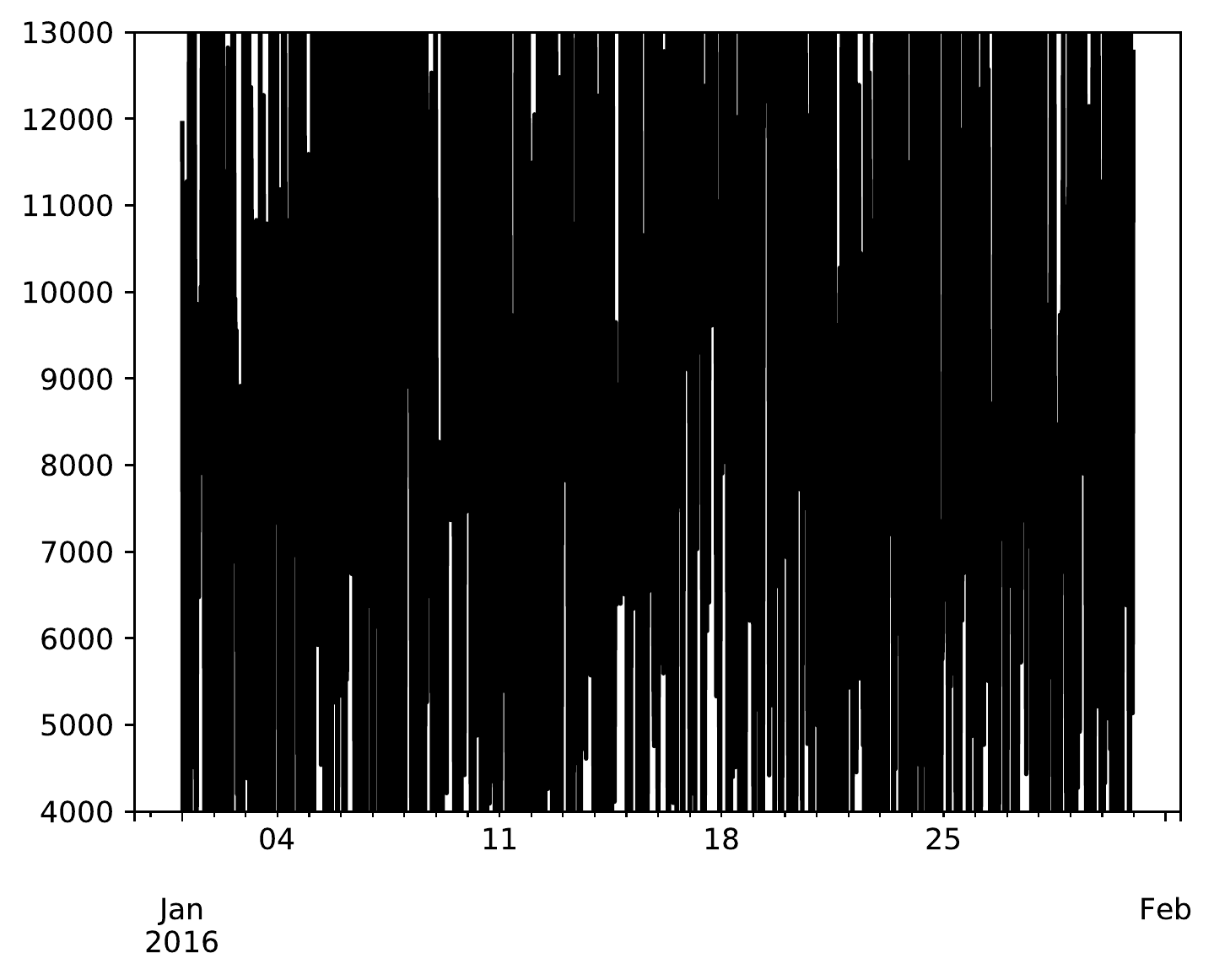}
    \includegraphics[width=\wlen]{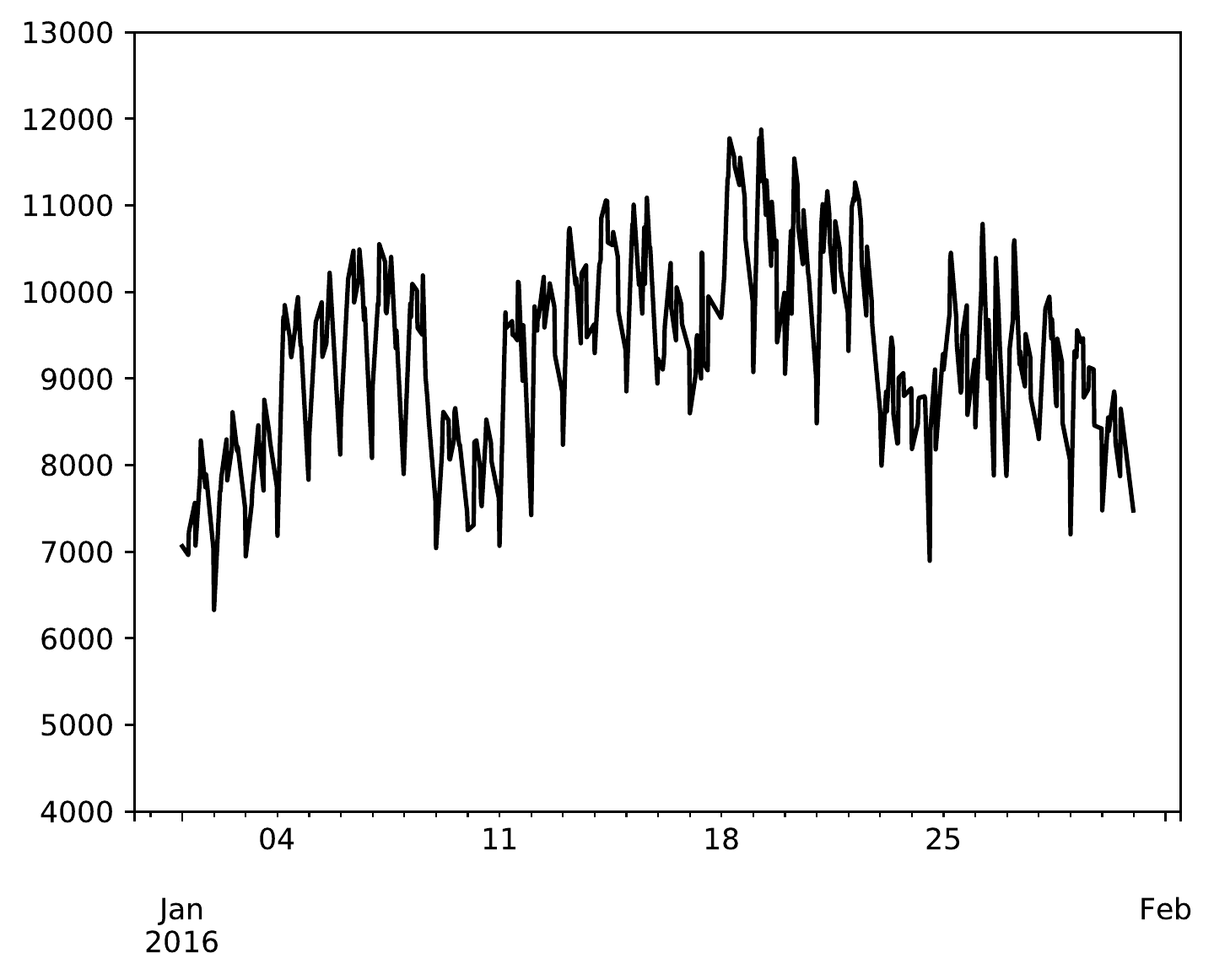}\\    
    \caption{Real load consumption data for the Auvergne-Rh\^{o}ne-Alpes region in January 2016 (left), and its private versions obtained using PeGaSus (center), and \algname{} (right) with privacy budget $\epsilon=1$ (top-row), $\epsilon=0.1$ (middle-row), and $\epsilon=0.01$ (bottom-row).
    }
    \label{fig:com_pegasus} 
\end{figure}

Despite the fact that the two algorithms were designed for solving
different problems, we adapted PeGaSus to the $w$-privacy model and to
satisfy the $(w,\alpha)$-indistinguishability and compare it against
\algname{} for completeness. Figure \ref{fig:com_pegasus} illustrates
a comparison of the algorithms on privately releasing load streams for
the Auvergne-Rh\^{o}ne-Alpes region in January 2016.  The real loads
are highlighted in the first column, while their private counterparts
are shown in the second and third columns. The figure compares our
proposed \algname\ (third column) against PeGaSus (second column) for
indistinguishability parameters $\alpha=10$ (top rows) $\alpha=50$
(middle row) and $\alpha=100$ (bottom row), and fixed a privacy budget
$\epsilon=1$. For PeGaSus, the privacy budget is divided equally for
each time step in the $w$-period. Additionally, the experiments use
the values of the meta-parameters indicated in the original paper
\cite{chen2017pegasus}. The settings for \algname{} are the same as in the
previous sections.  Figure \ref{fig:com_pegasus} clearly illustrates
that \algname\ produces private streams that are more accurate than
those produced by PeGaSus, when the $(w, \alpha)$-indistinguishability
model is adopted. It is an interesting research avenue to combine
PeGaSus and \algname{}.


\section{Related Work}
\label{sec:related_work}

Continuous release of aggregated real-time data has been studied in
previous work including \cite{dwork:2010a,dwork:10}. Most of the state-of-the-art either focuses on event-level privacy on infinite streams \cite{rastogi:10} or on user-level privacy on finite streams \cite{dwork:10}.  Dwork proposed an adaptation of differential privacy to a continuous observation setting \cite{dwork:10}. Her work focused on releasing bit-streams and proposed an algorithm for counting the number of 1s in the stream under event-level differential privacy. \citeA{mir2011pan} proposed \emph{pan-privacy} for estimating counts, moments, and heavy-hitters on data streams while preserving differential privacy even if the internal memory of the algorithms is observed by an attacker.  \citeA{kellaris:14} proposed the notion of $w$-event privacy to balance event-level and user-level privacy, trading off utility and privacy to protect event sequences within a time window of  $w$ time steps. \algname{} applies to the last model.

Within the differential privacy proposal for data streams that fit such model, \citeA{wang:16} proposed \emph{Rescue DP}, which is designed explicitly for spatiotemporal traces. PeGaSus \cite{chen2017pegasus} is another seminal work that allows to  protect  event privacy using a combination of perturbation, partitioning, and smoothing. The algorithm uses a combination of Laplace noise to perturb the data stream for privacy, a grouping strategy which incrementally partitions the space of observed data points by grouping points with small deviations, and a smoothing schema which is used to post-process the private data stream. Both algorithms rely on the idea of contiguous grouping time steps and average the perturbed data within every region in a group. We compared \algname{} against PeGaSus, which is the state-of-the-art on private data-stream release and has been showed to be effective in several settings.

\citeA{rastogi:10} proposed an algorithm which perturbs a small number of Discrete Fourier Transform (DFT) coefficients of the entire time series and reconstructs a released version of the Inverse DFT series. While, in the original paper, the algorithm requires the entire time-series, such approach has been used for $w$-event privacy in the context of data streams \cite{Fan:14}. We also compare \algname{} against DFT. 

\citeA{fan2013fast} proposed Fast, an adaptive algorithm for private release of aggregate statistics which uses a combination of sampling and filtering. Their algorithm is based on user-level differential privacy, which is not comparable to the chosen model of $w$-event level differential privacy.  Additionally, in their work, the authors compare the proposed approach against DFT \cite{rastogi:10} and, while showing improvements, the error remains within the same error magnitude of those produced by DFT. In contrast, our experimental analysis (Section \ref{sec:experiments}) clearly illustrates that \algname\ reduces the error of one order of magnitude when compared to DFT.

Finally, \citeA{chan:11} focused on releasing prefix-sums of the streaming data counts while adopting an event-based privacy model. \citeA{bolot:13} also used an event-based model and proposed an algorithm for answering sliding window queries on data streams.  The model adopted in their work is however incompatible with the privacy model adopted in this work, and hence we did not compare against such approaches.


\section{Conclusions}
\label{sec:conclusions}

This paper presented \algname{}, a novel algorithm for privately
releasing stream data in the $w$-event privacy model. \algname{} is a
4-step procedure consisting of sampling, perturbation, reconstruction,
and post-processing modules. The \emph{sampling} module selects a
small set of data points to measure privately in each period of
interest. The \emph{perturbation} module injects noise to the sampled
data points to ensure privacy. The \emph{reconstruction} module
reconstructs the data points excluded from measurement from the
perturbed sampled points. Finally, the \emph{post-processing} module
uses convex optimization over the private output of the previous
modules, along with the private answers of additional queries on the
data stream, to ensure consistency of quantities associated with salient
features of the data. \algname{} was evaluated on a real dataset from
the largest transmission operator in Europe. Experimental results on
multiple test cases show that \algname{} improves the accuracy of
the state-of-the-art by at least one order of magnitude in this
application domain. The accuracy improvements are measured, not only
in terms of the error distance to the original stream but also in the
accuracy of a popular load forecasting algorithm trained on 
private data sub-streams. The results additionally show that
\algname{} exhibits similar benefits on hierarchical stream data which
is also highly desirable in practice. Finally, the experimental
analysis has demonstrated that both the adaptive sampling and
post-processing optimization are critical in obtaining strong
accuracy. Future work will be devoted to generalizing these results to
the streaming setting where a data element is emitted at each time
step.



\bibliographystyle{theapa}
\bibliography{differential_privacy,time}

\end{document}